\documentclass[twoside]{article}
\usepackage[accepted]{aistats2026}

\usepackage{epsfig}
\usepackage{float}
\usepackage{afterpage}
\usepackage{amsmath}
\usepackage{amstext}
\usepackage{amssymb,bm, bbm}
\usepackage{latexsym}
\usepackage{color}
\usepackage{graphicx}
\usepackage{amsmath}
\usepackage{amsthm}
\usepackage{graphicx}
\usepackage[center]{caption}
\usepackage{booktabs}
\usepackage{multicol}
\usepackage{lipsum}
\usepackage{dblfloatfix}
\usepackage{mathrsfs}
\usepackage[round]{natbib}
\usepackage{tikz}
\usepackage{pgfplots}
\usepackage{enumitem}
\usepackage{makecell}
\usepackage{float}
\restylefloat{table}
\usepackage{subcaption}
\usepackage{mathtools}
\usepackage{csquotes}
\usepackage[utf8]{inputenc}
\usepackage{xcolor}
\usepackage{url}
\usepackage[hidelinks]{hyperref}

\usepackage{tikz}
\usepackage{pgfplots}
\usepgfplotslibrary{groupplots}
\usetikzlibrary{calc}
\pgfplotsset{compat=newest}

\usetikzlibrary{arrows.meta,intersections,calc}
\usepgfplotslibrary{fillbetween}

\allowdisplaybreaks

\newcommand{\rmd}{\mathrm{d}}
\newcommand{\bbE}{\mathbb{E}}\newcommand{\rme}{\mathrm{e}}

\newcommand{\bbN}{\mathbb{N}}

\newcommand{\bbR}{\mathbb{R}}

\newcommand{\rmW}{\mathrm{W}}\newcommand{\rmw}{\mathrm{w}}

\newcommand{\sfL}{\mathsf{L}}

\newcommand{\sfS}{\mathsf{S}}
\newcommand{\sfT}{\mathsf{T}}

\newcommand{\cP}{\mathcal{P}}
\newcommand{\cQ}{\mathcal{Q}}

\newcommand{\cS}{\mathcal{S}}

\newcommand{\cX}{\mathcal{X}}

\theoremstyle{mystyle}
\newtheorem{theorem}{Theorem}
\theoremstyle{mystyle}
\newtheorem{lemma}{Lemma}
\theoremstyle{mystyle}
\newtheorem{prop}{Proposition}
\theoremstyle{mystyle}
\newtheorem{corollary}{Corollary}
\theoremstyle{mystyle}
\newtheorem{definition}{Definition}
\theoremstyle{remark}
\newtheorem{rem}{Remark}
\theoremstyle{mystyle}
\theoremstyle{mystyle}
\theoremstyle{mystyle}
\theoremstyle{discussion}
\theoremstyle{mystyle}
\newtheorem{conj}{Conjecture}
\theoremstyle{mystyle}

\newcommand{\axisfont}{\fontsize{13.3}{15.96}\selectfont}

\begin{document}
	
	\twocolumn[
	
	\aistatstitle{Functional Properties of the Focal-Entropy}
	
	\aistatsauthor{ Jaimin Shah \And Martina Cardone \And  Alex Dytso}
	
	\aistatsaddress{ University of Minnesota,  \\ USA \And  University of Minnesota, \\USA \And Qualcomm Flarion Technology, Inc., \\ USA } ]
	
	\begin{abstract}
		The focal-loss has become a widely used alternative to cross-entropy in class-imbalanced classification problems, particularly in computer vision. Despite its empirical success, a systematic information-theoretic study of the focal-loss remains incomplete. In this work, we adopt a distributional viewpoint and study the \emph{focal-entropy},
		a focal-loss analogue of the cross-entropy. Our analysis establishes conditions for finiteness, convexity, and continuity of the \emph{focal-entropy}, and provides various asymptotic characterizations. We prove the existence and uniqueness of the focal-entropy minimizer, describe its structure, and show that it can depart significantly from the data distribution. In particular, we rigorously show that the focal-loss amplifies mid-range probabilities, suppresses high-probability outcomes, and, under extreme class imbalance, 
		induces an \emph{over-suppression regime} in which very small probabilities are further diminished. These results, which are also experimentally validated, offer a theoretical foundation for understanding the focal-loss and clarify the trade-offs that it introduces when applied to imbalanced learning tasks.
	\end{abstract}
	\section{Introduction}
	In learning-based prediction and decision-making, loss functions shape the decision rule and strongly influence performance. Selecting an appropriate loss is therefore critical. A widely used choice is the \emph{log-loss}, defined for $p \in (0,1]$ as
	\begin{equation}
		\sfL(p) = \log \left( \frac{1}{p} \right).
	\end{equation}
	When predicting with a distribution $Q_X$, the actual outcome $x \in \cX$ occurs with probability $P_X(x)$. The log-loss in this case is $\sfL(Q_X(x)) = \log \tfrac{1}{Q_X(x)}$. Averaging with respect to the true distribution $P_X$ yields the expected log-loss, which equals the \emph{cross-entropy} of $Q_X$ relative to $P_X$. Formally, this cross-entropy over $\cX$ is defined as
	\begin{equation}
		H(P_X,Q_X) \!=\! \begin{cases}
			\bbE_{X \sim P_X} \left [ \sfL(Q_X(X)) \right ]  & P_X \!\ll\! Q_X,\\
			\infty & \text{otherwise.}
		\end{cases}
		\label{eq:CrossEntropy}
	\end{equation}
	Despite its ubiquity, the log-loss often performs poorly under class imbalance -- common in tasks such as object detection~\citep{lin2017focal}, fraud detection~\citep{he2009learning}, natural language processing~\citep{li2020dice}, and various medical applications~\citep{al2019denses,vongkulbhisal2019unifying,shu2019pathological}. To address this,~\citep{lin2017focal} proposed the \emph{focal-loss}, now widely adopted, especially in computer vision. Its success arises from a simple modification of the log-loss: a modulating pre-factor is introduced that downweights well-classified (“easy”) examples and emphasizes misclassified (“hard”) ones, often yielding substantial accuracy gains.
	\begin{definition}[Focal-loss]
		\label{def:FocalLoss}
		Given a \emph{focus parameter} $\gamma \geq 0$ and a prediction score $p \in (0,1]$, the \emph{focal-loss} is defined as
		\begin{equation}
			\sfL_\gamma(p) = \underbrace{(1 - p)^\gamma}_{\text{focal term}} \underbrace{\log \left( \frac{1}{p} \right)}_{\text{log-loss}}.
			\label{eq:FocalLoss}
		\end{equation}
	\end{definition}
	Although the focal-loss is empirically effective, a \emph{comprehensive} information-theoretic treatment is still lacking. The cross-entropy has a clear foundation in information theory—minimizing it is equivalent to minimizing the Kullback–Leibler (KL) divergence between model and data—thus providing a principled optimization landscape. In contrast, the focal-loss alters this landscape in ways not yet fully understood. Characterizing the geometry of the resulting objective and its minimizer(s) is therefore essential.
	
	We adopt a distributional viewpoint and introduce the \emph{focal-entropy}, a focal-loss analogue of the cross-entropy that compares two distributions $P_X$ and $Q_X$.
	\begin{definition}[Focal-Entropy]
		\label{def:FocalEntropy}
		Given a \emph{focus parameter} $\gamma \geq 0$, the focal-entropy of $Q_X$ relative to $P_X$ over $\cX$ is defined as
		\begin{equation}
			H_\gamma(P_X,Q_X) \!=\! \begin{cases} 
				\bbE_{X \sim P_X}\!\left [\sfL_{\gamma}(Q_X(X))\right ] & P_X \!\ll\! Q_X,\\
				\infty & \text{otherwise.}
			\end{cases}
			\label{eq:FocalEntropy}
		\end{equation}
	\end{definition}
	For $\gamma=0$, the focal-entropy in~\eqref{eq:FocalEntropy} reduces to the cross-entropy in~\eqref{eq:CrossEntropy}.  A key defining property of the cross-entropy is that the minimizer over the second parameter $Q_X$ (in the space of all probability mass functions on $\cX$) is equal to $P_X $, that is, 
	\begin{equation}
		\label{eq:Gibbs}
		\min_{Q_X}  H(P_X, Q_X) =   H(P_X, P_X) = H(P_X),
	\end{equation}
	where the last expression is the Shannon entropy. 
    As we show in this paper, this property, however, is no-longer true for the  focal-entropy and, in general, the optimizer is not equal to $P_X$.  
    In this sense, the focal-loss may not appear to be a proper loss \citep{brehmer2020properization,buja2005loss}. However, recent work~\citep{bao2025calm} has shown that the focal-loss is a variant of proper scoring rules.
	The structure and basic properties (e.g., existence) of the minimizer, i.e.,
	\begin{equation}
		\label{eq:FEMinimizer}
		P_\gamma^\star = \arg \min_{Q_X} H_\gamma(P_X,Q_X)
	\end{equation}
	are indeed not fully understood. 
	This paper aims to address this gap by providing a full characterization of $P^\star_\gamma$  and studying many of its properties.

	\subsection{Notation}
	The set of all probability mass functions over a (possibly countably infinite) set $\cX$ is denoted by $\cP(\cX)$. For $P_X \in \cP(\cX)$, its support is
	\begin{equation}
		\cS = \{x \in \cX : P_X(x) > 0\}.
		\label{eq:SetS}
	\end{equation}
	We write $P_X \ll Q_X$ to denote absolute continuity of $P_X$ with respect to $Q_X$. For $P_X$, we let $p_{\min}=\min_{x\in\cS}P_X(x)$ and $p_{\max}=\max_{x\in\cS}P_X(x)$. All logarithms are natural. The Lambert function’s principal and negative branches are denoted by $\rmW_0(\cdot)$ and $\rmW_{-1}(\cdot)$, respectively~\citep{corless1996lambertw}.
	
	\subsection{Literature Review} 
	In~\citep{charoenphakdee2021focal}, the authors implicitly characterized the minimizer of the focal-entropy for the finite alphabet case. They showed that the minimizing class-posterior under the focal-loss is order preserving (i.e., it maintains the same ordering as the true posterior), leading to the conclusion that the focal-loss is \emph{classification-calibrated}.
	Building on this,~\citep{mukhoti2020calibrating} demonstrated that the focal-loss not only improves accuracy on challenging datasets but also yields models that are better calibrated than those trained with the cross-entropy. Their analysis provided both theoretical justification—framing the focal-loss as a regularized KL divergence that implicitly maximizes entropy—and empirical evidence across vision and NLP tasks, including robustness under distributional shift. In this work, we give a theoretical basis for one of their key empirical findings: that focal-loss predictions exhibit higher entropy than those from log-loss.
	Several extensions of the focal-loss have been proposed, from modifying the focal term~\citep{li2022generalized,li2020generalized} to making the focus parameter adaptive~\citep{ghosh2022adafocal}, and even showing that focal- and log-loss arise as special cases of a broader power-series formulation~\citep{leng2022polyloss}. Related directions include class-balanced reweighting~\citep{cui2019classbalanced}, margin-based methods for long-tailed learning~\citep{cao2019ldam}, and balanced softmax corrections~\citep{ren2020balanced}, each tackling class imbalance from complementary angles. On calibration, recent work such as dual focal-loss~\citep{zhang2023dual} and entropy-based generalizations like tilted cross-entropy~\citep{lin2021tce} or $\alpha$-entmax losses~\citep{peters2019sparse} further highlight the central role of entropy-modulated objectives. Collectively, these advances situate the focal-loss within a broader family of class imbalance- and calibration-aware losses that balance sharpness, entropy, and robustness across regimes.
	
	\subsection{Paper Outline and Contributions}
	Section~\ref{sec:AnalyticalPropFL} examines analytical properties of the focal-loss and related functions, focusing on objects such as the inverse of its derivative, which is central to characterizing the focal-entropy minimizer.
	
	Section~\ref{sec:focal_entropy} develops basic functional properties of the focal-entropy. Proposition~\ref{prop: focal-entropy vs gamma} shows that $\gamma \mapsto H_\gamma(P_X,Q_X)$ is non-increasing and convex. Proposition~\ref{Prop:LimitGammaInf} characterizes its asymptotic behavior for large $\gamma$, for both finite and countably infinite supports. Proposition~\ref{prop:focal-finite-iff-cross-finite} states that the focal-entropy is finite if and only if the cross-entropy is finite. Proposition~\ref{prop:convexity-Hgamma} establishes weak lower semicontinuity and convexity of $Q \mapsto H_\gamma(P_X,Q)$. Finally, Theorem~\ref{thm:MinimizationFocalEntropy} proves the existence, uniqueness, and structure of the focal-entropy minimizer $P_\gamma^\star$ in~\eqref{eq:FEMinimizer}. For the special case of two-point support, we further derive explicit bounds on $P_\gamma^\star$, as stated in Proposition~\ref{prop:binary-focal-minimizer}.

	Section~\ref{sec:prop_optimizer} focuses on the properties of the minimizer $P_\gamma^\star$, which shed light on the ability of the focal-loss to address class imbalance. Proposition~\ref{eq:PGmmaStar0Inft} shows that $P_\gamma^\star$ converges to the uniform distribution as $\gamma \to \infty$ and provides rates of convergence. To analyze how the focal-entropy transforms $P_X$ into $P_\gamma^\star$, we study the sequence $p_{(i)} - p^\star_{(i)}$, where $p_{(i)}$ denotes the $i$-th largest entry of $P_X$ and $p^\star_{(i)}$ is the $i$-th largest entry of $P^\star_\gamma$. Theorem~\ref{thm:Signdi} establishes that this sequence has at least one and at most two sign changes. This result enables a rigorous discussion of class imbalance, including the characterization of \emph{exact thresholds} for when small–to–moderate probabilities are amplified (i.e., imbalance is reduced). 
	Furthermore, we show that under extreme class imbalance, i.e., when certain probabilities are very small, the focal-loss does not amplify these probabilities, but instead further suppresses them. We name this \emph{over-suppression regime}. This finding is particularly important for practitioners, as it highlights the need to carefully select $\gamma$ to avoid such regimes. 
	
	To better understand the over-suppression regime, we study trade-offs between the support size, $\gamma$, and the values of the data distribution $P_X$.  Proposition~\ref{prop:binary_regime} shows that for every $\gamma >0$ and support size $|\cS|=2$, the over-suppression regime is not present.  We also conjecture that this holds for $|\cS|=3$ and provide theoretical and numerical evidence supporting this conjecture. 
	Proposition~\ref{prop:SuffCond1Pa=0} and Proposition~\ref{thm:SuffCondGammaMajor} establish sufficient conditions under which the over-suppression regime does not exist for arbitrary values of $|\cS|$.
	Finally, Proposition~\ref{prop:Majorization} shows that whenever the over-suppression regime does not exist, the data distribution $P_X$ majorizes the minimizer $P_\gamma^\star$. 
    Taken together, these results characterize how the focal-entropy reshapes the data distribution. Complementing this distributional viewpoint, Proposition~\ref{prop:RelatEntropyRelation} provides an information-theoretic interpretation by expressing the focal-entropy in terms of relative entropy.
	
	Section~\ref{sec:ExperimentalValidation} validates our results on both {\em synthetic} and {\em real} data (MNIST). Finally, Section~\ref{sec:Conslusion} concludes the paper. Most proofs are in the appendix and the code for our experiments is available at \mbox{\url{https://tinyurl.com/3kjmmua9}}.

	\section{Focal-Loss: Analytical Properties}
	\label{sec:AnalyticalPropFL}
	We derive several analytical properties of the focal-loss and related functions that will be useful for the derivation of the results in the subsequent sections of the paper.
	
	In the next proposition (proof in Appendix~\ref{app:AnalyticalPropFocalLoss}), we investigate the monotonicity and convexity of the focal-loss function in~\eqref{eq:FocalLoss}.
	\begin{prop}
		\label{prop:AnalyticalPropFocalLoss}
		For $p \in (0,1)$ and $\gamma \ge 0$,  we have the following properties:
		\begin{enumerate}
			\item $\gamma \mapsto \sfL_\gamma(p)$ is strictly decreasing;
			\item $p \mapsto \sfL_\gamma(p)$ is strictly decreasing with the derivative given by 
			\begin{equation} \label{eq:firstDer}
				\sfL^{\prime}_\gamma(p) = -(1-p)^{\gamma-1} \left( \gamma \log \frac{1}{p} +\frac{1-p}{p} \right) <  0;
			\end{equation}
			\item $p \mapsto \sfL_\gamma(p)$ is strictly convex with the second derivative given by 
			\begin{align} \label{eq:secondDer}
				\sfL^{ \prime \prime}_\gamma(p) 
				& =  \frac{(1-p)^{\gamma-2}}{p^2} \left[\gamma(1-\gamma)\,p^2\log p \right . \notag
				\\& \qquad  \left. + 2\gamma\,p(1-p) + (1-p)^2\right] >0.
			\end{align}
		\end{enumerate}
	\end{prop}
	As it will be shown later, the inverse of $\sfL^{\prime}_\gamma$ will play an important role in solving the first-order optimality conditions for the focal-entropy minimizer.  
	\begin{figure}
		\raggedright
		\input{Inverse_focal_prime}
		\vspace{-0.4cm}
		\caption{Inverse of $\sfL^{\prime}_\gamma$ in~\eqref{eq:firstDer}.}
		\vspace{-0.5cm}
		\label{fig:Inverse_focal_prime}
	\end{figure} 
	Some of its important properties are summarized next (proof in Appendix~\ref{app:AnalyticalPropFocalLossInverse}) and also illustrated in Figure~\ref{fig:Inverse_focal_prime}.
	\begin{prop}
		\label{prop:AnalyticalPropFocalLossInverse}
		For $p \in (0,1)$, the mapping $p \mapsto  \sfL_\gamma'(p)$ has a well-defined inverse $\left( \sfL^{\prime}_\gamma \right )^{-1}: (-\infty,0) \to (0,1)$, which satisfies the following properties:
		\begin{enumerate}
			\item $\left( \sfL^{\prime}_\gamma \right )^{-1}(-\infty) = 0$ and $\left( \sfL^{\prime}_\gamma \right )^{-1}(0^-) = 1$;
			\item $t \;\mapsto\; \left( \sfL^{\prime}_\gamma \right )^{-1}(t)$ is strictly increasing;
			\item $\left( \sfL^{\prime}_\gamma \right )^{-1}(t) \geq 0$.
		\end{enumerate}
	\end{prop}
	The strict monotonicity of $\left( \sfL^{\prime}_\gamma \right )^{-1}$ ensures that the mapping $t \;\mapsto\; \left( \sfL^{\prime}_\gamma \right )^{-1}(t)$ is injective (see Figure~\ref{fig:Inverse_focal_prime}). 
	For $\gamma=0$ and $\gamma=1$, we have a closed-form expression for the inverse of~\eqref{eq:firstDer}, which is given by
    \begin{equation}
	\left( \sfL^{\prime}_\gamma \right )^{-1}(t) = 
	\begin{cases}
		\frac{1}{\rmW_0(\rme^{1-t})} & \gamma=1,\\
		-\frac{1}{t} & \gamma = 0.
	\end{cases}
\end{equation}
	One final function that we will require is 
	\begin{equation}
		\label{eq:PhiFunc}
		\phi_\gamma(p) = - p \,\sfL_\gamma'(p), \,  p \in (0,1). 
	\end{equation}
	The function $\phi_\gamma$  will play an important role in Section~\ref{sec:prop_optimizer}, when we study properties of the minimizer of the focal-entropy.  Analyzing its unimodality and bounds helps explain how the focal-loss redistributes probability masses of the data distribution, particularly in imbalanced settings. The next proposition (proof in Appendix~\ref{app:properties_of_phi_gamma}) analyzes properties of $\phi_\gamma$ in~\eqref{eq:PhiFunc}.
	\begin{prop} \label{prop:properties_of_phi_gamma} For all $\gamma \ge 0$, the mapping $p \mapsto \phi_\gamma(p)$ satisfies the following properties:
		\begin{enumerate}
			\item $\phi_\gamma(0^+)=1$ and $\phi_\gamma(1) =0$;
			\item  $ p\mapsto \phi_\gamma(p)$ is unimodal. More specifically, the mapping $p \mapsto \phi_\gamma(p) $ is strictly decreasing if $\gamma > \kappa(p)$ and strictly increasing if $\gamma < \kappa(p)$ where 
			\begin{equation}
				\label{eq:kappaP}
				\kappa(p) = \frac{1}{p}- \frac{2(1-p)}{p \log \left( \frac{1}{p} \right )}, \ p \in (0,1). 
			\end{equation} 
			In addition, 
			\begin{itemize}
				\item For $p \in (0,1)$, $p \mapsto \kappa(p)$ is strictly decreasing;
				\item $\kappa(p) \leq 0$ for all $p \in [\rmw,1)$, where $\rmw = -\frac{1}{2}\rmW_0 \left( - \frac{2}{\rme^2} \right ) \approx 0.2032$;
			\end{itemize}		
			\item 
			$
			\max \limits_{ p \in (0,1)} \phi_\gamma(p) \le 1+\gamma;
			$
			\item Let $p^{+}_\gamma =\arg \max_{p \in (0,1)} \phi_\gamma(p)$, then
			\begin{equation}
				p^{+}_\gamma    \leq  \min \left\{ \rmw, \frac{1}{\left( 1+ \sqrt{1+\gamma}\right )^2} \right\};
				\label{eq:UBPgammaStar}
			\end{equation}
			\item  For $u \in (0, 1/3]$, it holds that
			$    
			\phi_\gamma(u) \ge \phi_\gamma \left(\frac{1-u}{2}\right). 
			$
		\end{enumerate}
	\end{prop}
	\section{Focal-Entropy: Functional Properties}
	\label{sec:focal_entropy}
	In this section, we develop functional properties of the focal-entropy, proving its convexity, finiteness, and the existence and uniqueness of its minimizer.
	\subsection{Focal-Entropy vs. $\gamma$}
	We here derive a few properties that highlight how the focal-entropy in~\eqref{eq:FocalEntropy} varies as a function of $\gamma$.
	The first property is directly derived from the inspection of~\eqref{eq:FocalEntropy} (see also Appendix~\ref{app: focal-entropy vs gamma}, where some additional properties of the mapping $\gamma \mapsto H_\gamma(P_X,Q_X)$ are provided).
	\begin{prop} \label{prop: focal-entropy vs gamma}
		Fix some $P_X,Q_X \in \cP(\cX)$ such that $P_X \ll Q_X$.
		For $\gamma \geq 0$, the mapping
		$\gamma \mapsto H_\gamma(P_X,Q_X)$ is non-increasing and convex. 
		Thus, we have that 
		\begin{equation}
			0 \le H_\gamma(P_X,Q_X) \le H_0(P_X,Q_X) \!= H(P_X,Q_X).
		\end{equation}
	\end{prop}
	The next proposition (proof in Appendix~\ref{app:LimitGammaInf}) presents a limit for the focal-entropy when the focus parameter $\gamma$ tends to infinity.
	\begin{prop}
		\label{Prop:LimitGammaInf}
		Fix some $P_X,Q_X \in \cP(\cX)$ such that $P_X \ll Q_X$ and $H(P_X,Q_X)<\infty$.  Then, it holds that
		\begin{equation}
			\lim_{\gamma \to \infty}  \!\!H_\gamma^{ \frac{1}{\gamma}}\!(P_X, Q_X) \!=\! \left \{ \begin{array}{ll} \!\!\!\max_{x \in \cS} \bar{Q}_X(x)  &  \!|\cS|\!<\!\infty,  \\
				\!\!\!1 &  \!|\cS| \!=\!\infty,
			\end{array} \right. 
			\label{eq:LimitgammaInfinity}
		\end{equation}
		where $\bar{Q}_X(x) = 1 - Q_X(x)$ and $\cS$ is defined in~\eqref{eq:SetS}.
	\end{prop}
	We highlight that, when the support size $|\cS|$ is finite, minimizing the right-hand side of~\eqref{eq:LimitgammaInfinity} over $Q_X$ yields the uniform distribution over $\cS$. A more precise statement of uniformity as $\gamma \to \infty$ will be given in Proposition~\ref{eq:PGmmaStar0Inft}.

	\subsection{Basic Functional Properties}
	We derive some basic functional properties of the focal-entropy in~\eqref{eq:FocalEntropy}. 
	We start with the next proposition (proof in Appendix~\ref{app:focal-finite-iff-cross-finite}), which gives a sufficient and necessary condition for the focal-entropy to be finite.
	\begin{prop}
		\label{prop:focal-finite-iff-cross-finite}
		Fix some $\gamma \ge 0$. Then,  $H_\gamma(P_X, Q_X) < \infty$ if and only if $H(P_X, Q_X) < \infty$.
	\end{prop}
	For the remainder of the paper, we consider only distributions with \emph{finite} focal-entropy; see \citep{cover1999elements} for an example of infinite entropy. The next proposition (proof in Appendix~\ref{app:convexity-Hgamma}) presents additional functional properties of the focal-entropy.
	\begin{prop}
		\label{prop:convexity-Hgamma}
		For any $\gamma\ge0$ and fixed $P_X$, it holds that the mapping $Q \;\mapsto\; H_\gamma(P_X,Q)	$:
		\begin{enumerate}
			\item is weakly lower semicontinuous\footnote{ 
                A functional $F:\cP(\cX)\to(-\infty,\infty]$ is said to be
				\emph{weakly lower semicontinuous} if for every sequence $(Q_n)_{n \ge 1}$ such that
				$Q_n \to Q$ (weak convergence), we have that $F(Q)\le \liminf_{n\to\infty} F(Q_n).$}; and 
			\item is convex.  Moreover, it is \emph{strictly} convex over the probability space defined on $\cS$ in~\eqref{eq:SetS}.
		\end{enumerate}
	\end{prop}
	\begin{rem}
		Proposition~\ref{prop:convexity-Hgamma} shows that the focal-entropy is weakly lower semicontinuous. We note that, in general, the focal-entropy is not weakly continuous.
        To see this, let $\cX=\mathbb{N}$ and consider the distribution
\begin{equation}
	P_X(k) = 2^{-k}, \quad k=1,2,\dots.	
\end{equation}
Now, for each $n$, define
\begin{equation}
	Q_n(k)=
	\begin{cases}
		P_X(k) & k<n,\\
		\sum \limits_{j=n}^\infty P_X(j) & k=n,\\
		0 & k>n.
	\end{cases}	
\end{equation}
Then, for every fixed $k$, $Q_n(k)\to P_X(k)$ as $n\to\infty$, so $Q_n$ converges weakly to $P_X$. However, the fact that $Q_n(k)=0$ for all $k>n$ with $P_X(k)>0$, leads to
$H_\gamma(P_X,Q_n)=\infty$, whereas
$H_\gamma(P_X,P_X)<\infty$.  This example shows that $H_\gamma$ is not weakly continuous.
	\end{rem}
	In~\citep{mukhoti2020calibrating}, the authors provided an inequality between the focal-entropy and the cross-entropy, which allows certain learning conclusions to be made. In Proposition~\ref{prop:MinFocalCrossRel}, we will provide a novel relationship that bounds the difference between their infimum values.
	\subsection{Optimization of the Focal-Entropy}
	\label{sec:OptFocalEntropy}
	The minimum of the cross-entropy in~\eqref{eq:CrossEntropy} (see~\eqref{eq:Gibbs}) follows from the non-negativity of the relative entropy, or equivalently Gibbs’ inequality. We now show a similar result for the focal-entropy. 
	
	\begin{theorem} 
		\label{thm:MinimizationFocalEntropy}
		Fix some $P_X \in \cP(\cX)$ with support $\cS$ as in~\eqref{eq:SetS} and $\gamma \ge 0$.  Then, there exists  a unique $0<\alpha_\gamma^\star <\infty$ such that 
		\begin{equation} \label{eq:optimizer}
			P^\star_\gamma(x)
			=\begin{cases}
				(\mathsf L_\gamma')^{-1} \left (-\tfrac{\alpha_\gamma^\star}{P_X(x)}\right )&x\in\cS,\\
				0&x\notin\cS,
			\end{cases}
		\end{equation}
		is a valid distribution.
		Moreover,  for all $Q_X \in \cP(\cX)$, it holds~that
		\begin{align}
			H_\gamma ( P_X, Q_X) 
			&\!\ge \!   H_\gamma ( P_X, P^\star_\gamma) \!+\! \alpha_\gamma^\star \! \left(\!1 \!-\!\! \sum_{x \in \cS} \! Q_X(x) \!\right)\\
			&\geq  H_\gamma ( P_X, P^\star_\gamma),
		\end{align}
		and $P^\star_\gamma = \arg \min_{Q_X} H_\gamma(P_X,Q_X)$
		is the unique minimizer for the focal-entropy. 	
	\end{theorem}
	\begin{proof}
			For $P^\star_\gamma$ in~\eqref{eq:optimizer} to be a valid distribution, the constant $\alpha_\gamma^\star$ must be a positive root of the equation $F(\alpha_\gamma^\star ) = 1$, where
			\begin{equation} \label{eq:normalization}
				F(\alpha) = \sum_{x \in \cS}(\sfL^\prime_\gamma)^{-1} \left(-\tfrac{\alpha}{P_X(x)}\right), \quad \alpha >0.	
			\end{equation}
			From Proposition~\ref{prop:AnalyticalPropFocalLossInverse}, each summand in~\eqref{eq:normalization} is continuous and strictly decreasing in $\alpha>0$ having values on $(0,1)$. Thus, $\alpha \mapsto F(\alpha)$ is continuous and strictly decreasing with $F(0^+)\!=\!|\cS|$ and $F(+\infty)\!=\!0$ (see Proposition~\ref{prop:AnalyticalPropFocalLossInverse}).  
			Then, by the intermediate value theorem, there is a unique positive root $\alpha^\star_\gamma$, yielding a valid~$P^\star_\gamma$.  
			
			To show the second part of the theorem, recall from Proposition~\ref{prop:AnalyticalPropFocalLoss} that $p \mapsto \sfL_\gamma\left(p\right)$ is strictly convex and so it can be lower bounded by a tangent line as follows,
			\begin{align}
				&  \sfL_\gamma\left(Q(X)\right)  \\
				&\ge     \sfL_\gamma\left(P_\gamma^\star\left(X\right) \right)  \!+\! \sfL_\gamma' \left( P_\gamma^\star\left(X\right) \right) \left( Q(X) - P_\gamma^\star\left(X\right)  \right)  \\
				& =  \sfL_\gamma\left(P_\gamma^\star \left(X\right) \right)  - \tfrac{\alpha_\gamma^\star}{P_X(X)} \left( Q(X) - P_\gamma^\star\left(X\right) \right ). 
			\end{align}
			Consequently, for any choice of $Q$, we arrive at
			\begin{align}
				&H_\gamma ( P_X, Q) \\
				&=   \bbE \left [ \sfL_\gamma\left(Q\left(X\right)\right) \right ] \\
				&\ge  \bbE[ \sfL_\gamma\left(P_\gamma^\star\left(X\right) \right)] \!-\! \bbE \left[ \tfrac{\alpha_\gamma^\star}{P_X(X)} \left( Q(X) \!-\! P_\gamma^\star\left(X\right) \right ) \right]\\
				&=  \bbE[ \sfL_\gamma\left(P_\gamma^\star\left(X\right) \right)] - \alpha_\gamma^\star  \left( \sum_{x \in \cS} Q(x)  -1 \right)  \\
				& \geq   \bbE[ \sfL_\gamma\left(P_\gamma^\star\left(X\right) \right)]\label{eq:EqualityPropositionSupportS} \\
				& = H_\gamma ( P_X, P^\star_\gamma),
			\end{align}
            where all the expectations are taken with respect to $X \sim P_X$.
			The fact that $P^\star_\gamma$ in~\eqref{eq:optimizer} is the unique minimizer for the focal-entropy is due to the convexity result in Proposition~\ref{prop:convexity-Hgamma}. 
		\end{proof}

	\begin{figure}
		\raggedright
		\begin{tikzpicture}

\definecolor{darkgray176}{RGB}{176,176,176}
\definecolor{darkslateblue48103141}{RGB}{48,103,141}
\definecolor{gold25323136}{RGB}{253,231,36}
\definecolor{indigo68184}{RGB}{68,1,84}
\definecolor{lightgray204}{RGB}{204,204,204}
\definecolor{mediumseagreen53183120}{RGB}{53,183,120}

\tikzset{barborder/.style={draw=black, line width=0.3pt}}

\begin{axis}[
width=8.4cm,
height=4.3cm,
legend cell align={left},
legend style={
  at={(0,1)}, anchor=north west,
  fill=none, draw=none,
  font=\footnotesize, inner sep=1pt, row sep=-1pt,
  legend image post style={xscale=0.6,yscale=0.6}
},
tick align=outside,
tick pos=left,
x grid style={dashed, gray!50},
xmajorgrids,
xmin=-0.44, xmax=2.64,
xtick style={color=black},
xtick={0,1,2},
xticklabels={$x_1$,$x_2$,$x_3$},
y grid style={dashed, gray!50},
ymajorgrids,
ymin=0, ymax=0.485235539718907,
ytick style={color=black}
]
\draw[barborder,fill=mediumseagreen53183120,fill opacity=1] (axis cs:-0.3,0) rectangle (axis cs:-0.1,0.1820587756164);
\addlegendimage{area legend,draw=black,fill=mediumseagreen53183120}
\addlegendentry{$P_X$}

\draw[barborder,fill=mediumseagreen53183120,fill opacity=1] (axis cs:0.7,0) rectangle (axis cs:0.9,0.462129085446578);
\draw[barborder,fill=mediumseagreen53183120,fill opacity=1] (axis cs:1.7,0) rectangle (axis cs:1.9,0.355812138937022);

\draw[barborder,fill=gold25323136,fill opacity=1] (axis cs:-0.1,0) rectangle (axis cs:0.1,0.192363988704074);
\addlegendimage{area legend,draw=black,fill=gold25323136}
\addlegendentry{$ P_{\gamma=0.5}^\star$}

\draw[barborder,fill=gold25323136,fill opacity=1] (axis cs:0.9,0) rectangle (axis cs:1.1,0.447518135736118);
\draw[barborder,fill=gold25323136,fill opacity=1] (axis cs:1.9,0) rectangle (axis cs:2.1,0.360117875560263);

\draw[barborder,fill=indigo68184,fill opacity=1] (axis cs:0.1,0) rectangle (axis cs:0.3,0.205600699256593);
\addlegendimage{area legend,draw=black,fill=indigo68184}
\addlegendentry{$P_{\gamma=1}^\star $}

\draw[barborder,fill=indigo68184,fill opacity=1] (axis cs:1.1,0) rectangle (axis cs:1.3,0.433150548776666);
\draw[barborder,fill=indigo68184,fill opacity=1] (axis cs:2.1,0) rectangle (axis cs:2.3,0.361248751966286);

\draw[barborder,fill=darkgray176,fill opacity=1] (axis cs:0.3,0) rectangle (axis cs:0.5,0.230977415445977);
\addlegendimage{area legend,draw=black,fill=darkgray176}
\addlegendentry{$P_{\gamma=2}^\star $}

\draw[barborder,fill=darkgray176,fill opacity=1] (axis cs:1.3,0) rectangle (axis cs:1.5,0.410466794963668);
\draw[barborder,fill=darkgray176,fill opacity=1] (axis cs:2.3,0) rectangle (axis cs:2.5,0.3585557895899);
\end{axis}

\end{tikzpicture}
		\vspace{-0.4cm}
		\caption{Example of how a $P_X$ with $|\cS| =3$ is transformed into $P^\star_\gamma$ in~\eqref{eq:optimizer} for different values of $\gamma$.}
		\vspace{-0.4cm}
		\label{fig:P_gamma_start}
	\end{figure}
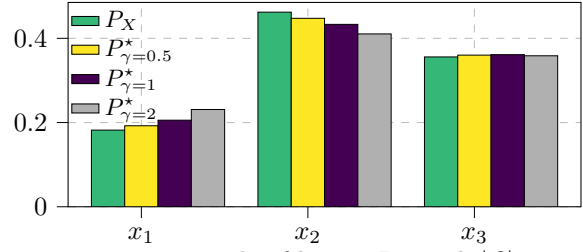
	Figure~\ref{fig:P_gamma_start} illustrates how a distribution $P_X$ with $|\cS|=3$ is transformed into $P^\star_\gamma$ in~\eqref{eq:optimizer} for different values of $\gamma$. The smallest probability (at $x=x_1$) is amplified, while the largest (at $x=x_2$) is reduced. Intuitively, the transformation shifts the mass from high- to low-probability outcomes, mitigating imbalance. Section~\ref{sec:prop_optimizer} gives a complete characterization of which probabilities are amplified or suppressed.

	Some immediate consequences of Theorem~\ref{thm:MinimizationFocalEntropy} (with Proposition~\ref{prop:AnalyticalPropFocalLossInverse}) are provided in the next corollary (proof in Appendix~\ref{app:ConseqTheorem1}).
	\begin{corollary}
		\label{cor:ConseqTheorem1}
		The distribution $P^\star_\gamma$ in Theorem~\ref{thm:MinimizationFocalEntropy} satisfies the following properties:
		\begin{enumerate}
			\item $P^\star_\gamma \ll P_X$;
			\item For any $x_1,x_2 \in \cS$
			\begin{equation}
				P_X(x_1) \ge P_X(x_2)
				\ \Longrightarrow \ 
				P^\star_\gamma(x_1) \ge P^\star_\gamma(x_2);	
			\end{equation}
			\item $P_\gamma^\star = P_X$ if and only if $\gamma=0$ or $P_X$ is a uniform distribution.
		\end{enumerate}
	\end{corollary}
	\begin{rem}
		It is useful to think of the mapping $P_X \mapsto P_\gamma^\star$ as an operator. The inverse of this operator is  well-defined and has a closed-form expression given~by 
		\begin{align}
			P_X(x) &=  \frac{\beta}{\sfL_\gamma'( P_\gamma^\star(x) ) }, \ x \in \cS,  \label{eq:Inverse}
			\\ \beta = & \left(    \sum_{v \in \cS} \frac{1}{\sfL_\gamma'( P_\gamma^\star(v) ) }   \right)^{-1}.  
		\end{align}
		In some cases, this inverse operator is easier to handle due to its closed form. This view relates to~\citep[Theorem~11]{charoenphakdee2021focal}, which derived~\eqref{eq:Inverse} for finite alphabets via Lagrangian duality. In contrast, Theorem~\ref{thm:MinimizationFocalEntropy} extends the result to countably infinite alphabets and gives a more geometric proof based on the inverse of the focal-loss derivative. Our setting is also more natural in practice: since the optimizer $P_\gamma^\star$ is usually unknown, one typically works with the forward mapping $P_X \mapsto P_\gamma^\star$.
    \end{rem}
The minimizer in~\eqref{eq:optimizer} depends on $\alpha_\gamma^\star$ and hence, finding bounds on $\alpha_\gamma^\star$ helps to reduce the search space. The next proposition (proof in Appendix~\ref{app:BoundsAlpha}) provides bounds on $\alpha_\gamma^\star$, which will be leveraged for theoretical derivations in the remainder of the paper.
	\begin{prop}
		\label{prop:BoundsAlpha}
		For all $\gamma \ge 0$, it holds that
		\begin{equation}
			\alpha_\gamma^\star \le \max_{t \in (0,1)} \phi_\gamma(t) \le 1+\gamma. \label{eq:alpha_bound_max_phi}  
		\end{equation}
		Moreover, suppose that  $ |\cS| =N<\infty$.
		Then: 
        {
        \begin{itemize} 
        \item For all $\gamma \ge 0$, it holds that
		\begin{equation}
			\label{eq:BoundsOptAlpha}
			p_{\min} \; c_{N,\gamma}  \le \alpha_\gamma^\star \le p_{\max} \; c_{N,\gamma},
		\end{equation}
		where $c_{N,\gamma} = - \sfL^{\prime}_\gamma(1/N)>0$.
		\item  For all $\gamma > \kappa(p_{\min})$, with $\kappa(\cdot)$ being defined in~\eqref{eq:kappaP}, it holds that
		\begin{equation}
			\phi_\gamma(p_{\max}) \le \alpha^\star_\gamma \le  \phi_\gamma(p_{\min}).
			\label{eq:BoundsOptAlpha2}
		\end{equation}
		\item  For all $\gamma < \kappa(p_{\max})$, with $\kappa(\cdot)$ being defined in~\eqref{eq:kappaP}, it holds that
		\begin{equation}
			\phi_\gamma(p_{\min}) < \alpha^\star_\gamma <  \phi_\gamma(p_{\max}).
			\label{eq:BoundsOptAlpha3}
		\end{equation}
	\end{itemize}}
	\end{prop}
	\begin{rem}       
		The lower and upper bounds in~\eqref{eq:BoundsOptAlpha},~\eqref{eq:BoundsOptAlpha2}, and~\eqref{eq:BoundsOptAlpha3} match when $P_X$ is a uniform distribution.
	\end{rem}
    \begin{rem}
    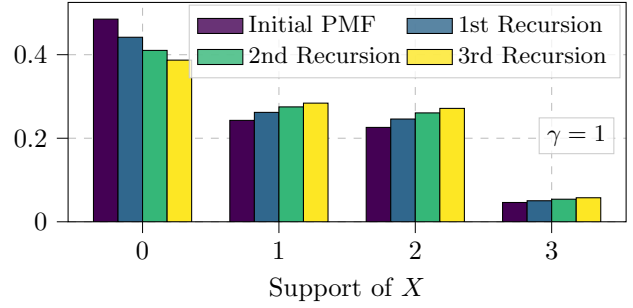
\begin{figure}
	\centering
	\begin{tikzpicture}
\definecolor{darkgray176}{RGB}{176,176,176}
\definecolor{darkslateblue48103141}{RGB}{48,103,141}
\definecolor{gold25323136}{RGB}{253,231,36}
\definecolor{indigo68184}{RGB}{68,1,84}
\definecolor{lightgray204}{RGB}{204,204,204}
\definecolor{mediumseagreen53183120}{RGB}{53,183,120}

\begin{axis}[
width = 9cm,
height = 4.5cm,
legend cell align={left},
legend columns=2,
legend style={
    fill opacity=0.8, 
    draw opacity=1, 
    text opacity=1, 
    draw=lightgray204,
    at={(0.99,0.99)},
    anchor=north east,
    font=\small
},
legend entries={Initial PMF, 1st Recursion, 2nd Recursion, 3rd Recursion},
legend image code/.code={
    \draw[#1] (0cm,-0.1cm) rectangle (0.3cm,0.1cm);
},
tick align=outside,
tick pos=left,
x grid style={dashed, gray!50}, 
xmajorgrids,                    
xlabel={Support of $X$},
xmin=-0.276, xmax=3.816,
xtick style={color=black},
xtick={0.27,1.27,2.27,3.27},
xticklabels={0,1,2,3},
y grid style={dashed, gray!50}, 
ymajorgrids,
ymin=0, ymax=0.52511729,
ytick style={color=black}
]

\addplot[
    ybar,
    bar width=0.18,
    fill=indigo68184,
    draw=black,
    area legend
] coordinates {
    (0,0.48511729)
    (1,0.24276922)
    (2,0.22591902)
    (3,0.04619447)
};

\addplot[
    ybar,
    bar width=0.18,
    fill=darkslateblue48103141,
    draw=black,
    area legend
] coordinates {
    (0.18,0.441756579555843)
    (1.18,0.261915248957649)
    (2.18,0.245980274388615)
    (3.18,0.0503478970978921)
};

\addplot[
    ybar,
    bar width=0.18,
    fill=mediumseagreen53183120,
    draw=black,
    area legend
] coordinates {
    (0.36,0.410292024188029)
    (1.36,0.275033035196676)
    (2.36,0.260648104474101)
    (3.36,0.0540268361411944)
};

\addplot[
    ybar,
    bar width=0.18,
    fill=gold25323136,
    draw=black,
    area legend
] coordinates {
    (0.54,0.386988832865333)
    (1.54,0.284036599989397)
    (2.54,0.271449668947819)
    (3.54,0.057524898197451)
};

\node[
    draw=lightgray204,
    fill=white,
    fill opacity=0.8,
    text opacity=1,
    inner sep=3pt,
    anchor=north east,
    font=\small
] at (axis cs:3.73,0.25) {$\gamma = 1$};
\end{axis}
\end{tikzpicture}
	\vspace{-0.4cm}
	\caption{Initial probability mass function: $P_X = \begin{bmatrix}0.485 & 0.243 & 0.226 & 0.046 \end{bmatrix}$.}
	\vspace{-0.4cm}
	\label{fig:Idempotent}
\end{figure} 
It is well known that the minimizer of the cross-entropy is \emph{idempotent}: if the minimization procedure is initialized at a distribution $P_0^\star$, the optimizer coincides with the initialization and no further change occurs. In contrast, the minimizer of the focal-entropy does \emph{not} exhibit this property. As illustrated in Figure~\ref{fig:Idempotent}, the optimizer $P_\gamma^\star$ in~\eqref{eq:optimizer} is, in general, \emph{not idempotent}.
The results shown in Figure~\ref{fig:Idempotent} were obtained as follows. Starting from the initial probability mass function
\begin{equation}
P_X=
\begin{bmatrix}
0.485 & 0.243 & 0.226 & 0.046
\end{bmatrix},
\end{equation}
we numerically determined (using the bounds in Proposition~\ref{prop:BoundsAlpha}) the value of $\alpha_\gamma^\star$ satisfying the normalization condition
$F(\alpha_\gamma^\star)=1$, where $F(\cdot)$ is defined in~\eqref{eq:normalization}. Using this value, we computed the corresponding optimizer $P_\gamma^\star$ via~\eqref{eq:optimizer} (Recursion~1 in Figure~\ref{fig:Idempotent}), which differs from the initial distribution $P_X$. Repeating this procedure by reapplying the optimization to the newly obtained distribution (Recursions~2 and~3 in Figure~\ref{fig:Idempotent}) reveals that successive iterates continue to change. Hence, unlike the cross-entropy case, focal-entropy minimization induces a non-trivial transformation of the distribution under repeated application.  

An interesting direction for future work is to determine whether this recursive procedure converges to a limiting distribution.
    \end{rem}
	\subsection{ Two-Point Support}
	We consider distributions with $|\cS| = 2$ for which we aim to study the structure of the optimizer $P_\gamma^\star$ in~\eqref{eq:optimizer}. Without loss of generality, we assume that $\mathcal{X}=\{0,1\}$. The next proposition (proof in Appendix~\ref{app:binary-focal-minimizer}) provides bounds on $P^\star_\gamma$.
	\begin{prop}\label{prop:binary-focal-minimizer}
		Let $P_X\in\cP(\{0,1\})$ with $P_X(1)=p\in \left (0,\tfrac12 \right]$, { without loss of generality.
        Then, for any $\gamma>0$, the unique minimizer in~\eqref{eq:optimizer}
		over $\cP(\{0,1\})$ satisfies
		\begin{equation} \label{eq:BoundsonOptimlDistrBinaryCase}
			\widetilde{Q}_\gamma(1) \; \le \; P^\star_\gamma(1)\;\le \; \widetilde{Q}_{\gamma+1} (1),
		\end{equation}
        where
		\begin{equation} \label{eq:powerLaw prob}
			\widetilde{Q}_\gamma(1) = q_\gamma := \frac{p^{\tfrac{1}{\gamma}}}{\,p^{\tfrac{1}{\gamma}} + (1-p)^{\tfrac{1}{\gamma}}\,},
		\end{equation}
        and $\widetilde{Q}_\gamma(0) = 1-\widetilde{Q}_\gamma(1)$.
		
		Furthermore, for any $\gamma > 0$,} it holds that
		\begin{equation}
			\widetilde{Q}_{\gamma+1} (1) - \widetilde{Q}_\gamma(1) \;\le\;\frac{\left|\log \left( \tfrac{p}{1-p} \right )\right|}{4\,\gamma^2}.
			\label{eq:DiffBoundsBinary}
		\end{equation}
	\end{prop}
	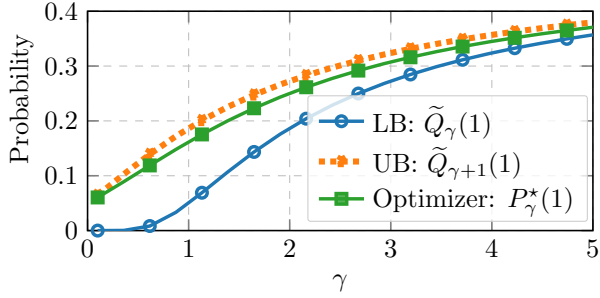
\begin{figure}
		\centering
\begin{tikzpicture}

\definecolor{steelblue31119180}{RGB}{31,119,180}
\definecolor{darkorange25512714}{RGB}{255,127,14}
\definecolor{forestgreen4416044}{RGB}{44,160,44}

\begin{axis}[
  width=8.3cm,
  height=4.5cm,
  xlabel={$\gamma$}, xlabel style={font=\normalsize},
  ylabel={Probability},
  y label style={font=\normalsize,at={(axis description cs:0,0.6)},anchor=south,yshift=17pt},
  tick label style={font=\normalsize},
  axis line style={semithick},
  tick style={semithick},
  xmin=0, xmax=5,
  ymin=0, ymax=0.4,
  xmajorgrids, ymajorgrids,
  grid style={dashed, gray!50},
  legend cell align={left},
  legend style={
    font=\normalsize,
    fill=white,
    fill opacity=0.8,
    text opacity=1,
    draw=gray!40,
    line width=0.4pt,
    rounded corners=1pt,
    at={(0.99,0.04)},
    anchor=south east
  }
]

\addplot [
  line width=1.3pt,
  steelblue31119180,
  mark=o,
  mark size=2,
  mark repeat=2
]
table {%
0.1 1.63103766612776e-13
0.357894736842105 0.000267236590634087
0.615789473684211 0.00831315998046325
0.873684210526316 0.033241777249815
1.13157894736842 0.0690058086523647
1.38947368421053 0.107253850821951
1.64736842105263 0.143396807790836
1.90526315789474 0.175748200483793
2.16315789473684 0.204049210270176
2.42105263157895 0.228609165847452
2.67894736842105 0.249908010631649
2.93684210526316 0.268433139691788
3.19473684210526 0.284621192159074
3.45263157894737 0.298843713608773
3.71052631578947 0.311409766374296
3.96842105263158 0.322574081376681
4.22631578947368 0.332546136246905
4.48421052631579 0.34149845422828
4.7421052631579 0.349573632191429
5 0.356890086257402
};
\addlegendentry{LB: $\widetilde{Q}_\gamma(1)$}

\addplot [
  line width=2.3pt,
  darkorange25512714,
  dotted,
  mark=triangle*,
  mark size=2,
  mark repeat=2
]
table {%
0.1 0.0643585401345818
0.357894736842105 0.102625642110769
0.615789473684211 0.139159171210372
0.873684210526316 0.172006898154139
1.13157894736842 0.200793596340831
1.38947368421053 0.225787131842278
1.64736842105263 0.247458448327724
1.90526315789474 0.26629849586888
2.16315789473684 0.282751437239597
2.42105263157895 0.297196841931213
2.67894736842105 0.309951036085814
2.93684210526316 0.321274911017297
3.19473684210526 0.331383016457072
3.45263157894737 0.340451996324585
3.71052631578947 0.348627778613377
3.96842105263158 0.356031471350471
4.22631578947368 0.362764109265354
4.48421052631579 0.368910444095333
4.7421052631579 0.374541963610539
5 0.379719298141306
};
\addlegendentry{UB: $\widetilde{Q}_{\gamma+1}(1)$}

\addplot [
  line width=1.3pt,
  forestgreen4416044,
  mark=square*,
  mark size=2,
  mark repeat=2
]
table {%
0.1 0.0603921952542986
0.357894736842105 0.0890827673049415
0.615789473684211 0.118817124615606
0.873684210526316 0.147921726320151
1.13157894736842 0.175307945852179
1.38947368421053 0.200416317589615
1.64736842105263 0.223072116342939
1.90526315789474 0.24333503425314
2.16315789473684 0.261385005240009
2.42105263157895 0.277449537058942
2.67894736842105 0.291762871167587
2.93684210526316 0.304545588025734
3.19473684210526 0.31599604554106
3.45263157894737 0.326288147460224
3.71052631578947 0.335572228507162
3.96842105263158 0.343977288928501
4.22631578947368 0.351613655358339
4.48421052631579 0.35857561489526
4.7421052631579 0.364943821593071
5 0.370787405133542
};
\addlegendentry{Optimizer: $P^\star_\gamma(1)$}

\end{axis}
\end{tikzpicture}
		\vspace{-0.6cm}
		\caption{True Probability: $0.05$}
		\label{fig:binary case}
	\end{figure}
	Figure~\ref{fig:binary case} illustrates the performance of the bounds in Proposition~\ref{prop:binary-focal-minimizer} in an imbalanced setting, i.e., $p=0.05$. We observe that the upper bound is tight across the entire range of $\gamma$, whereas the lower bound becomes tight for larger values of $\gamma$.

	\section{Properties of the Optimizer}
	\label{sec:prop_optimizer}
	In this section, we explore properties of the optimizer $P_\gamma^\star$ in~\eqref{eq:optimizer}. These are critical since they can shed light on important characteristics of the focal-loss, such as how it helps with class imbalance. 
	
	\subsection{Large $\gamma$ Asymptotics of $P_\gamma^\star$}
	Our goal here is to understand how $P^\star_\gamma$ in~\eqref{eq:optimizer} is affected as the value of $\gamma$ increases. The next proposition (proof in  Appendix~\ref{app:PGmmaStar0Inft}) answers this question.
	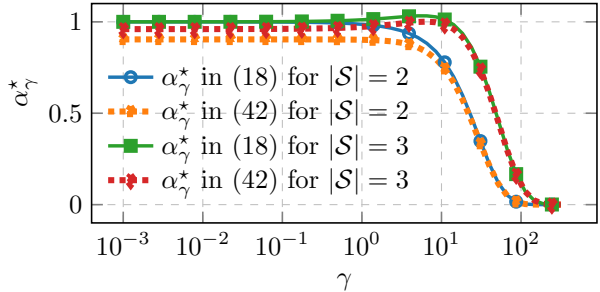
\begin{figure}
		\raggedright
\begin{tikzpicture}

\definecolor{steelblue31119180}{RGB}{31,119,180}
\definecolor{darkorange25512714}{RGB}{255,127,14}
\definecolor{forestgreen4416044}{RGB}{44,160,44}
\definecolor{crimson2143940}{RGB}{214,39,40}

\begin{axis}[
width=8.3cm,
height=4.5cm,
  xlabel={$\gamma$}, xlabel style={font=\normalsize},
  ylabel={$\alpha^\star_\gamma$}, 
  y label style={font=\normalsize,at={(axis description cs:0,0.6)},anchor=south,yshift=17pt},
  tick label style={font=\normalsize},
  axis line style={semithick},
  tick style={semithick},
  xmode=log,
  log basis x={10},
  xmin=4.0e-05, xmax=90,      
  ymin=-0.1, ymax=1.1,         
  log ticks with fixed point,
  xtick={1e-06,1e-05,0.0001,0.001,0.01,0.1,1,10,100,1000,10000},
  xticklabels={
    \(\displaystyle {10^{-5}}\),
    \(\displaystyle {10^{-4}}\),
    \(\displaystyle {10^{-3}}\),
    \(\displaystyle {10^{-2}}\),
    \(\displaystyle {10^{-1}}\),
    \(\displaystyle {10^{0}}\),
    \(\displaystyle {10^{1}}\),
    \(\displaystyle {10^{2}}\)
  },
  xmajorgrids, ymajorgrids,
  grid style={dashed, gray!50},
  legend cell align={left},
  legend style={
    font=\normalsize,
    fill=none,
    draw=none,
    at={(0.01,0.09)},
    anchor=south west
  }
]

\addplot [line width=1.3pt, steelblue31119180, mark=o, mark size=2, mark repeat=4]
table {%
0.0001 0.999987655559835
0.000129492584220526 0.999984014385518
0.000167683293681101 0.999979299011102
0.000217137430293752 0.999973192405378
0.000281176869797423 0.999965283891925
0.000364103194931067 0.999955041423163
0.000471486636345739 0.9999417756203
0.000610540229658533 0.999924593090327
0.00079060432109077 0.999902335780462
0.00102377396633958 0.999873502134051
0.00132571136559011 0.999836144457731
0.00171669790660786 0.999787735119753
0.0022229964825262 0.999724991801486
0.00287861559235457 0.999643648717665
0.00372759372031494 0.99953815619665
0.00482695743767787 0.999401284756459
0.00625055192527398 0.999223600892492
0.00809400121608312 0.998992769232245
0.0104811313415469 0.998692617401503
0.0135722878297165 0.998301873398759
0.0175751062485479 0.997792446038147
0.0227584592607479 0.99712706106493
0.0294705170255181 0.996255979416674
0.0381621340794936 0.99511239659887
0.0494171336132383 0.993605934567768
0.0639915233634927 0.991613364338718
0.0828642772854684 0.988965305570673
0.107303094052616 0.985427095180057
0.138949549437314 0.980671245280064
0.179929362329155 0.974237856321452
0.232995181051537 0.965477984845165
0.30171148105293 0.953473526890775
0.390693993705462 0.936926995746944
0.505919748843582 0.914020130632252
0.655128556859551 0.88226258432951
0.848342898244072 0.838409579870131
1.09854114198756 0.778633876759613
1.42252931348537 0.699260651932491
1.84206996932672 0.598378859838704
2.38534400643142 0.478284324563901
3.08884359647748 0.347798138778671
3.99982339560893 0.222342178733015
5.17947467923121 0.11956137326024
6.70703561118431 0.0510966479828313
8.68511373751352 0.0161351845276841
11.2465782211982 0.00342793421941688
14.5634847750124 0.000434256907738118
18.8586327877265 2.80740591734775e-05
24.4205309454865 7.57530415385146e-07
31.6227766016838 6.57973951997625e-09
};
\addlegendentry{$\alpha^\star_\gamma$ in~\eqref{eq:optimizer} for $|\cS|=2$}

\addplot [line width=2.3pt, darkorange25512714, dotted, mark=triangle*, mark size=2, mark repeat=4]
table {%
0.0001 0.90393391250239
0.000129492584220526 0.903933911032772
0.000167683293681101 0.903933908568527
0.000217137430293752 0.903933904436526
0.000281176869797423 0.903933897508124
0.000364103194931067 0.903933885890947
0.000471486636345739 0.903933866412173
0.000610540229658533 0.903933833752329
0.00079060432109077 0.903933778993316
0.00102377396633958 0.903933687184822
0.00132571136559011 0.903933533266019
0.00171669790660786 0.903933275232216
0.0022229964825262 0.903932842687908
0.00287861559235457 0.903932117676983
0.00372759372031494 0.903930902592044
0.00482695743767787 0.903928866480488
0.00625055192527398 0.903925455261989
0.00809400121608312 0.903919741726888
0.0104811313415469 0.903910175196739
0.0135722878297165 0.903894164317933
0.0175751062485479 0.903867383033976
0.0227584592607479 0.903822618814572
0.0294705170255181 0.903747867183296
0.0381621340794936 0.903623192151832
0.0494171336132383 0.903415580878664
0.0639915233634927 0.903070570343924
0.0828642772854684 0.902498747653924
0.107303094052616 0.90155425554061
0.138949549437314 0.900001141017844
0.179929362329155 0.897461878405383
0.232995181051537 0.893341149527941
0.30171148105293 0.886718231491837
0.390693993705462 0.876205874009948
0.505919748843582 0.859787479895388
0.655128556859551 0.834675684076484
0.848342898244072 0.797292656690935
1.09854114198756 0.743555408809005
1.42252931348537 0.669725419523179
1.84206996932672 0.574047900660253
2.38534400643142 0.459069463479042
3.08884359647748 0.333713099626598
3.99982339560893 0.213154856233075
5.17947467923121 0.114494653976707
6.70703561118431 0.0488749691926097
8.68511373751352 0.0154169294204775
11.2465782211982 0.00327221249622666
14.5634847750124 0.000414194258878697
18.8586327877265 2.67591577631462e-05
24.4205309454865 7.216597972549e-07
31.6227766016838 6.26516568457393e-09
};
\addlegendentry{$\alpha_\gamma^\star$ in~\eqref{eq: taylorExpansio of Alpha} for $|\cS|=2$}

\addplot [line width=1.3pt, forestgreen4416044, mark=square*, mark size=2, mark repeat=4]
table {%
0.0001 1.00001275510886
0.000129492584220526 1.00001651636785
0.000167683293681101 1.00002138654736
0.000217137430293752 1.00002769244338
0.000281176869797423 1.00003585705917
0.000364103194931067 1.00004642787326
0.000471486636345739 1.00006011334199
0.000610540229658533 1.00007783006322
0.00079060432109077 1.00010076360968
0.00102377396633958 1.00013044694469
0.00132571136559011 1.00016886135245
0.00171669790660786 1.00021856608282
0.0022229964825262 1.00028286452425
0.00287861559235457 1.00036601645706
0.00372759372031494 1.00047350808427
0.00482695743767787 1.00061239355955
0.00625055192527398 1.00079172370033
0.00809400121608312 1.00102307849056
0.0104811313415469 1.00132121891556
0.0135722878297165 1.00170486824709
0.0175751062485479 1.00219761949893
0.0227584592607479 1.0028289379793
0.0294705170255181 1.00363517492178
0.0381621340794936 1.00466041322716
0.0494171336132383 1.00595680161905
0.0639915233634927 1.00758375767191
0.0828642772854684 1.00960497130791
0.107303094052616 1.01208143334452
0.138949549437314 1.01505764153035
0.179929362329155 1.01853658847858
0.232995181051537 1.02243705041246
0.30171148105293 1.02652418012758
0.390693993705462 1.03030197633734
0.505919748843582 1.03285521577936
0.655128556859551 1.03263191290216
0.848342898244072 1.027171462159
1.09854114198756 1.01281984586194
1.42252931348537 0.984549624346528
1.84206996932672 0.936134201544064
2.38534400643142 0.861087119646527
3.08884359647748 0.754821899361104
3.99982339560893 0.618094166608444
5.17947467923121 0.460616373515677
6.70703561118431 0.302029927932836
8.68511373751352 0.166865486315601
11.2465782211982 0.0734760961281609
14.5634847750124 0.0240100755459026
18.8586327877265 0.0053119107465137
24.4205309454865 0.000706922362496798
31.6227766016838 4.85859222863928e-05
};
\addlegendentry{$\alpha^\star_\gamma$ in~\eqref{eq:optimizer} for $|\cS|=3$}

\addplot [line width=2.3pt, crimson2143940, dotted, mark=diamond*, mark size=2, mark repeat=4]
table {%
0.0001 0.960138024501967
0.000129492584220526 0.960142096669894
0.000167683293681101 0.960147369476659
0.000217137430293752 0.960154196785676
0.000281176869797423 0.960163036664105
0.000364103194931067 0.960174482007153
0.000471486636345739 0.960189300121144
0.000610540229658533 0.960208483858071
0.00079060432109077 0.960233317625559
0.00102377396633958 0.960265462520449
0.00132571136559011 0.960307065993247
0.00171669790660786 0.960360902888231
0.0022229964825262 0.960430556459579
0.00287861559235457 0.960520650060026
0.00372759372031494 0.960637142616984
0.00482695743767787 0.960787703648886
0.00625055192527398 0.960982186167395
0.00809400121608312 0.961233217800123
0.0104811313415469 0.961556930773777
0.0135722878297165 0.961973848038578
0.0175751062485479 0.962509932284009
0.0227584592607479 0.963197780840223
0.0294705170255181 0.964077902258541
0.0381621340794936 0.965199922810238
0.0494171336132383 0.966623415864636
0.0639915233634927 0.968417780914472
0.0828642772854684 0.970660155979779
0.107303094052616 0.973429631245214
0.138949549437314 0.976794913346389
0.179929362329155 0.980790916015897
0.232995181051537 0.985377397810931
0.30171148105293 0.990369770777711
0.390693993705462 0.995329089366505
0.505919748843582 0.999396647091993
0.655128556859551 1.0010624505606
0.848342898244072 0.997873787341266
1.09854114198756 0.986132785139287
1.42252931348537 0.960715706673426
1.84206996932672 0.915277602794585
2.38534400643142 0.843246576340221
3.08884359647748 0.740022705028775
3.99982339560893 0.606391982677119
5.17947467923121 0.452039925655465
6.70703561118431 0.296424420254511
8.68511373751352 0.163753203647585
11.2465782211982 0.0720934046465461
14.5634847750124 0.023553663944442
18.8586327877265 0.00520994763076953
24.4205309454865 0.000693234831858312
31.6227766016838 4.76382582545419e-05
};
\addlegendentry{$\alpha_\gamma^\star$ in~\eqref{eq: taylorExpansio of Alpha} for $|\cS|=3$}

\end{axis}
\end{tikzpicture}
		\vspace{-0.4cm}
		\caption{Comparison of $\alpha^\star_\gamma$ in~\eqref{eq:optimizer} with its approximation in~\eqref{eq: taylorExpansio of Alpha}. }
		\label{fig:alpha_vs_gamma}
		\vspace{-0.3cm}
	\end{figure}
	\begin{prop}
		\label{eq:PGmmaStar0Inft}
		Suppose that $|\cS| < \infty$. Then, as $\gamma \to \infty$, we have that
		\begin{align}
			\alpha_\gamma^\star &=  -  \sfL'_\gamma \left( \frac{1}{|\cS|} \right) \mathsf{HM}(P_X) +  O \left( \gamma \, \left (1\!-\!\frac{1}{|\cS|}\right )^{\gamma} \right),  \label{eq: taylorExpansio of Alpha}
			\\
			P_\gamma^\star(x) &=  \frac{1}{|\cS|} + O \left( \frac{1}{\gamma}\right ),
		\end{align}
		where  $\mathsf{HM}(P_X) =\frac{|\cS|}{ \sum_{x \in \cS} \frac{1}{P_X(x)}}$ is the \emph{harmonic mean} of~$P_X$.
	\end{prop}
	
	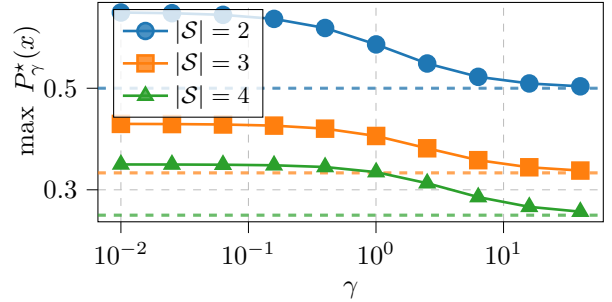
\begin{figure}
		\raggedright
		\begin{tikzpicture}
\definecolor{darkgray176}{RGB}{176,176,176}
\definecolor{darkorange25512714}{RGB}{255,127,14}
\definecolor{forestgreen4416044}{RGB}{44,160,44}
\definecolor{steelblue31119180}{RGB}{31,119,180}
\definecolor{gray}{RGB}{128,128,128}

\begin{axis}[
width=8.3cm,
height=4.5cm,
  log basis x={10},
  tick align=outside,
  tick pos=left,
  x grid style={dashed, gray!50},
  xmajorgrids,
  xminorgrids,
  log ticks with fixed point,
  minor grid style={dashed, gray!30},
  xlabel={\normalsize $\gamma$},
  xmin=0.00660536773049804, xmax=60.5568102065137,
  xmode=log,
  xtick style={color=black},
  major tick length=0.15cm,
  minor tick length=0.08cm,
  xtick={0.0001,0.001,0.01,0.1,1,10,100,1000},
  xticklabels={
    $10^{-4}$,
    $10^{-3}$,
    $10^{-2}$,
    $10^{-1}$,
    $10^{0}$,
    $10^{1}$,
    $10^{2}$,
    $10^{3}$
  },
  ylabel={\normalsize $\max \ P^\star_\gamma(x)$},
  y label style={at={(axis description cs:0,0.6)},anchor=south,yshift=17pt},
  ymin=0.237, ymax=0.669,
  ytick={0.3,0.5,0.7},
  ytick style={color=black},
  ymajorgrids,
  y grid style={dashed, gray!50},
  legend pos=north west,
  legend style={font=\small, fill=white, fill opacity=0.8, text opacity=1},
  tick label style={font=\normalsize}
]

\addplot [line width=1.2pt, steelblue31119180, dashed, opacity=0.7, forget plot]
table {%
0.00660536773049804 0.5
60.5568102065138 0.5
};
\addplot [line width=1.2pt, darkorange25512714, dashed, opacity=0.7, forget plot]
table {%
0.00660536773049804 0.333333333333333
60.5568102065138 0.333333333333333
};
\addplot [line width=1.2pt, forestgreen4416044, dashed, opacity=0.7, forget plot]
table {%
0.00660536773049804 0.25
60.5568102065138 0.25
};

\addplot [line width=1pt, steelblue31119180, mark=*, mark size=3, mark options={solid}]
table {%
0.01 0.649127146500632
0.0251321062979236 0.647812120133039
0.0631622766970132 0.644541922767985
0.15874010519682 0.636589218581776
0.398947319755006 0.618664736915434
1.00263864473545 0.586319994562927
2.51984209978975 0.548732590410964
6.3328939505899 0.522482300627871
15.9158963939703 0.509432634142286
40 0.503768920898437
};
\addplot [line width=1pt, darkorange25512714, mark=square*, mark size=3, mark options={solid}]
table {%
0.01 0.429786847961168
0.0251321062979236 0.429459691951251
0.0631622766970132 0.428615090173935
0.15874010519682 0.426379618215378
0.398947319755006 0.420400837831494
1.00263864473545 0.406003241416784
2.51984209978975 0.381750745105364
6.3328939505899 0.358322784026768
15.9158963939703 0.344348210444423
40 0.337882975290995
};
\addplot [line width=1pt, forestgreen4416044, mark=triangle*, mark size=3, mark options={solid}]
table {%
0.01 0.349909007183214
0.0251321062979236 0.349766375040872
0.0631622766970132 0.349382950863855
0.15874010519682 0.348279919546712
0.398947319755006 0.344872717657836
1.00263864473545 0.334799646263946
2.51984209978975 0.312935514269611
6.3328939505899 0.28564690983103
15.9158963939703 0.266426926257282
40 0.256898306657604
};

\legend{$|\cS|=2$, $|\cS|=3$, $|\cS|=4$}
\end{axis}
\end{tikzpicture}
		\vspace{-0.4cm}
		\caption{Convergence of $P^\star_\gamma$ to $1/ |\cS|$.}
		\vspace{-0.4cm}
		\label{fig:GammaInftyConvergence}
	\end{figure} 
	{Figure~\ref{fig:alpha_vs_gamma} compares $\alpha^\star_\gamma$ in~\eqref{eq:optimizer} with its approximation in~\eqref{eq: taylorExpansio of Alpha} for the two following initial distributions 
    \begin{align}
	P_X & = \begin{bmatrix}
		0.65 & 0.35
	\end{bmatrix}, \ \text{for } |\cS|=2,
	\\P_X &= \begin{bmatrix}
		0.43 & 0.32 & 0.25
	\end{bmatrix}, \ \text{for } |\cS|=3. 
\end{align}
 Figure~\ref{fig:GammaInftyConvergence} illustrates the convergence of $P_\gamma^\star$ to the uniform distribution for different values of $|\cS|$, again using the same initial distributions above for $|\cS|=2$ and $|\cS|=3$ and}
 \begin{equation}
 P_X  = \begin{bmatrix}
		0.35& 0.25& 0.25& 0.15 
	\end{bmatrix} , \ \text{for } |\cS|=4.
 \end{equation}

	\subsection{The Three Bins Property}
	We wish to understand when $P^\star_\gamma$ `dominates' $P_X$ and vice versa. To this end, we define the sequence:   
	\begin{equation} \label{eq:di def}
		d_i = p_{(i)} - p^\star_{(i)}, \, i \in \bbN,
	\end{equation}
	where $p_{(i)}$ is the $i$-th largest entry of $P_X$ and $p^\star_{(i)}$ is the $i$-th largest entry of $P^\star_\gamma$ in~\eqref{eq:optimizer}. By Corollary~\ref{cor:ConseqTheorem1}, $P_X$ and $P_\gamma^\star$ share the same ordering; thus, $d_i$ is well defined, preserving the correspondence between ordered entries. An immediate consequence of this and the fact that probabilities sum to one is the following lemma.
	\begin{lemma} \label{lem:minimum_num_sign_change}
		If $P_X$ is uniform or $\gamma=0$,  then $d_i=0$ for all $i \in \bbN $. Otherwise, the sequence $\{d_i \}$ has \emph{at least} one sign change. 
	\end{lemma}
	
	Before providing the main result of this section, which can be thought of as a converse to Lemma~\ref{lem:minimum_num_sign_change}, we present the following preparatory lemma.

	\begin{lemma}
		\label{lemma:Roots}
		For any fixed $\gamma>0$, $\phi_\gamma(p) - \alpha_\gamma^\star = 0$ has at least one and at most two roots. In particular, 
		\begin{itemize}
			\item  $ p_{\gamma,a}  \in [0,  p_{\gamma}^+)$  denotes the smallest root,  which exists only if $\alpha_\gamma^\star \geq 1$. For  $\alpha_\gamma^\star < 1$, we set $p_{\gamma,a} =0$. 
			\item  $  p_{\gamma,b}  \in (p_{\gamma}^+,1) $ denotes the largest root, which always exists. 
		\end{itemize}
	\end{lemma}
	\begin{proof}
		We recall that $\phi_\gamma(p)$ is unimodal (Proposition~\ref{prop:properties_of_phi_gamma}); hence $\phi_\gamma(p) - \alpha_\gamma^\star = 0$ has at most two roots. At least one root exists since
		\begin{equation}
			\alpha_\gamma^\star \leq \max_{p \in (0,1)} \phi_\gamma(p),
		\end{equation}
		as shown in Proposition~\ref{prop:BoundsAlpha}. Moreover, since $\phi_\gamma(0^+)=1$, $\phi_\gamma(1)=0$, and $\phi_\gamma$ is unimodal, if $\alpha_\gamma^\star < 1$ there is only one root, $p_{\gamma,b}$. This completes the proof.
	\end{proof}
	
	The next result (proof in  Appendix~\ref{app:Signdi}) characterizes how the sequence of sign changes in $\{d_i \}$ behaves. 
	\begin{theorem}
		\label{thm:Signdi} Suppose that $P_X$ is \emph{not} uniform and $\gamma > 0$.  Then, the following holds:
		\begin{itemize}
			\item if $p_{(i)} \in (0,p_{\gamma,a}]$, then $d_i \geq 0$;
			\item if $p_{(i)} \in (p_{\gamma,a},p_{\gamma,b})$, then $d_i < 0$;
			\item if $p_{(i)} \in [ p_{\gamma,b},1)$, then $d_i \geq 0$;
		\end{itemize}
		where $0 \!\leq\!p_{\gamma,a}\!\le\!  p^+_\gamma \!\le\! p_{\gamma,b} < 1$ are defined in Lemma~\ref{lemma:Roots}. 
	\end{theorem}
	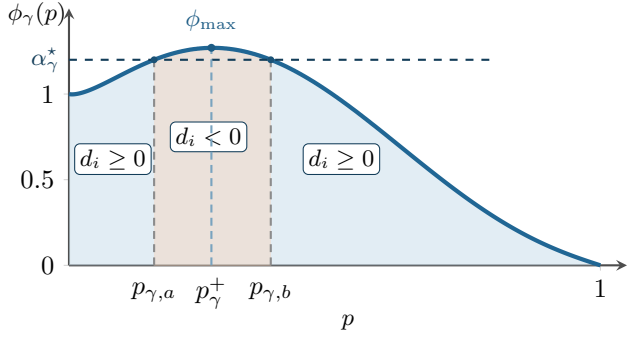
\begin{figure}
		\raggedright
		{
\pgfmathsetmacro{\A}{9.0}
\pgfmathsetmacro{\aa}{1.5}
\pgfmathsetmacro{\bb}{2.7}
\pgfmathsetmacro{\alphastar}{1.2}

\definecolor{acc}{RGB}{32,102,148}
\definecolor{accLight}{RGB}{181,209,226}
\definecolor{band}{RGB}{252,205,170}
\definecolor{grid}{RGB}{210,210,210}

\pgfmathdeclarefunction{phiBump}{1}{%
  \pgfmathparse{ 1 - (#1) + \A * (#1)^\aa * (1-#1)^\bb }%
}

\begin{tikzpicture}
\begin{axis}[
    width= 9cm, height=5cm,
    xmin=0, xmax=1.05,
    ymin=0, ymax=1.50,
    axis lines=left,
    axis line style={black!70, thick},
    axis background/.style={fill=none},
    xlabel={$p$}, 
    ylabel={$\phi_\gamma(p)$},
    y label style={
      at={(axis description cs:0,1)},
      anchor=south,
      xshift = - 10pt,
      yshift = 12pt,
      rotate=-90
    },
    label style={font=\small},
    tick style={draw=grid},
    tick align=outside,
    major tick length=2pt, minor tick length=1pt,
    xtick={1},
    xticklabels={$1$},
    ytick={0,0.5,1},
    yticklabels={$0$,$0.5$,$1$,$1.5$},
    grid=none,
    minor grid style={draw=grid, line width=0.2pt},
    major grid style={draw=grid, line width=0.3pt},
    clip=false,
    every axis plot/.style={line join=round}
  ]

    \addplot[
      name path=phi,
      domain=0:1, samples=800,
      ultra thick, acc
    ] { phiBump(x) };

    \path[name path=axisline] (axis cs:0,0) -- (axis cs:1,0);
    \addplot[accLight!70, opacity=0.55] fill between[of=phi and axisline];

    \addplot[name path=alpha, dashed, thick, acc!60!black]
      coordinates {(0,\alphastar) (0.8,\alphastar)};
    \node[anchor=east, font=\footnotesize, text=acc!60!black]
      at (axis cs:0,\alphastar) {$\alpha_\gamma^{\star}$};

    \path [name intersections={of=phi and alpha, by={Ia,Ib}}];

    \node[
      draw=acc!60!black,
      fill=white,
      rounded corners=2pt,
      inner sep=2pt,
      font=\small,
      anchor=west
    ] at ($(Ib)+(0.06,-0.58)$) {$d_i \ge 0$};
    
    \node[
      draw=acc!60!black,
      fill=white,
      rounded corners=2pt,
      inner sep=2pt,
      font=\small,
      anchor=west
    ] at ($(Ib)+(-0.19,-0.45)$) {$d_i < 0$};

\node[
      draw=acc!60!black,
      fill=white,
      rounded corners=2pt,
      inner sep=2pt,
      font=\small,
      anchor=west
    ] at ($(Ib)+(-0.37,-0.58)$) {$d_i \ge 0$};

    \addplot[band, opacity=0.35]
      fill between[
        of=phi and alpha,
        split,
        every segment no 0/.style={fill=none},
        every segment no 2/.style={fill=none}
      ];

    \path[name path=bandtop] (Ia|-{axis cs:0,\alphastar}) -- (Ib|-{axis cs:0,\alphastar});
    \path[name path=bandbot] (Ia|-{axis cs:0,0})          -- (Ib|-{axis cs:0,0});
    \addplot[band, opacity=0.35] fill between[of=bandtop and bandbot];

    \draw[densely dashed, thick, black!45]
      (Ia |- {axis cs:0,0}) -- (Ia |- {axis cs:0,\alphastar})
      (Ib |- {axis cs:0,0}) -- (Ib |- {axis cs:0,\alphastar});

    \fill[acc!80!black] (Ia) circle (1.3pt);
    \fill[acc!80!black] (Ib) circle (1.3pt);
    \node[below=3pt, font=\normalsize]
      at (Ia |- {axis cs:0,0}) {$p_{\gamma,a}$};
    \node[below=3pt, font=\normalsize]
      at (Ib |- {axis cs:0,0}) {$p_{\gamma,b}$};

    \pgfmathsetmacro{\steps}{2000}
    \pgfmathsetmacro{\pmax}{0}
    \pgfmathsetmacro{\vmax}{phiBump(0)}
    \pgfmathtruncatemacro{\N}{\steps}
    \pgfplotsinvokeforeach{0,...,\N}{%
      \pgfmathsetmacro{\p}{#1/\steps}%
      \pgfmathsetmacro{\val}{phiBump(\p)}%
      \ifdim \val pt > \vmax pt
        \xdef\pmax{\p}%
        \xdef\vmax{\val}%
      \fi
    }

    \draw[densely dashed, thick, acc!60]
      ({axis cs:\pmax,0}) -- ({axis cs:\pmax,\vmax});

    \fill[acc!90!black] ({axis cs:\pmax,\vmax}) circle (1.5pt)
      node[above, yshift=3pt, font=\footnotesize] {$\phi_{\max}$};
    \node[below=1pt, font=\normalsize]
      at ({axis cs:\pmax,0}) {$p^{+}_\gamma$};

  \end{axis}
\end{tikzpicture}}
		\vspace{-0.8cm}
		\caption{Illustration of Theorem~\ref{thm:Signdi}.}
		\label{fig:sign_changes}
		\vspace{-0.4cm}
	\end{figure}
	Figure~\ref{fig:sign_changes} illustrates the result in Theorem~\ref{thm:Signdi}. In words, Theorem~\ref{thm:Signdi} has the following consequences: 
	\begin{itemize}[leftmargin=*]
		\item \textbf{Small-to-moderate probabilities} $\boldsymbol{(p_{\gamma,a},\,p_{\gamma,b})}$.
		In this regime $d_i<0$ and hence, $p^{\star}_{(i)} > p_{(i)}$.
		The focal-loss amplifies these probabilities, which is the principal mechanism by which it mitigates class imbalance, i.e., the probability mass is shifted toward mid–range entries.
		
		\item \textbf{Large probabilities} $\boldsymbol{[p_{\gamma,b},\,1)}$.
		Here $d_i\geq 0$ and hence, $p^{\star}_{(i)} \leq p_{(i)}$.
		The focal-loss downweights high-probability (``easy”) samples. Part of this effect is normalization: since mid-range probabilities are amplified and since $\sum_i p^{\star}_{(i)}=1$, then the largest entries must decrease.
		
		\item \textbf{Very small probabilities} $\boldsymbol{(0,\,p_{\gamma,a}]}$.
		In this regime, which we refer to as {\em over-suppression regime}, we also have $d_i \geq 0$ and hence, $p^{\star}_{(i)} \leq p_{(i)}$.
		Thus, extremely small probabilities are further suppressed, which can, in principle, exacerbate imbalance at the extreme tail. However, as shown in~\eqref{eq:UBPgammaStar}, for sufficiently large values of $\gamma$ the interval $ (0,\,p_{\gamma,a}]$ is empty or negligible. Thus, while in practice this effect is typically muted, one should not ignore this regime especially in extremely imbalanced regimes.  
	\end{itemize}

	An immediate consequence of Theorem~\ref{thm:Signdi} is provided in the next corollary (see also Figure~\ref{fig:sign_changes}).
	\begin{corollary} 
		\label{cor:NumberOfSignChanges}
		Suppose that $P_X$ is \emph{not} uniform and $\gamma > 0$.  Then, 
		the sequence $\{ d_i \},i \in \bbN,$ has at least one and at most two sign changes. 
	\end{corollary}
	
	\subsection{Sign Changes for Small Support Sizes}
	We show that when the support is binary (proof in Appendix~\ref{app:binary_regime}) and possibly ternary, we have $p_{\min} > p_{\gamma,a}$, i.e., the {\em over-suppression regime} never occurs. 
	\begin{prop} \label{prop:binary_regime}
		Let $|\cS|=2$.
		Then, for all $\gamma>0$, it holds that $d_{1} \geq 0$ and $d_{2} < 0$. 
	\end{prop}
	\begin{conj}
		\label{conj:S3}
		Let $|\cS|=3$.
		Then, for all $\gamma>0$, it holds that $p_{\min} > p_{\gamma,a}$, i.e., $d_3 < 0$ (i.e., the over-suppression regime is absent for $|\cS|=3$). 
	\end{conj}
	We now provide numerical and theoretical evidence in support of Conjecture~\ref{conj:S3}.
	Numerically, we carried out extensive simulations (see Figure~\ref{fig: heatMap}) and we could not find cases where $d_3 \geq 0$.
	To provide theoretical evidence, we use the next result (proof in Appendix~\ref{app:ternary_regime}).

	\begin{figure}
		\centering
		\input{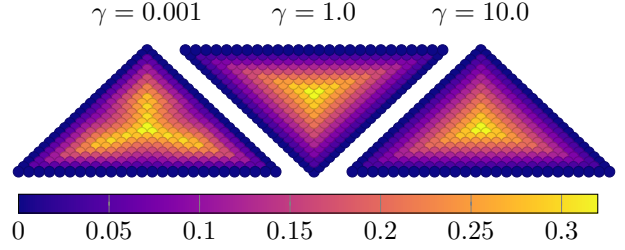}
		\caption{Heatmap of $p_{\min} - p_{\gamma,a}$ over the ternary probability simplex.}
		\label{fig: heatMap}
		\vspace{-0.3cm}
	\end{figure}
	\begin{prop} \label{prop:ternary_regime}
		Let $|\cS|=3$. Then, for all $\gamma>0$, it holds that $d_1 \geq 0$.
	\end{prop}
	Now, since $d_1 \geq 0$, two cases can occur:
	\begin{enumerate}
		\item $d_2 \geq 0$ (i.e., $p_{(2)}> p_{\gamma,b}$): in this case, from Corollary~\ref{cor:NumberOfSignChanges}, we need $d_3 <0$, i.e., $p_{\min}> p_{\gamma,a}$. Therefore, for this case, Conjecture~\ref{conj:S3} holds.
		\item $d_2 <0$ (i.e., $p_{\gamma,a}<p_{(2)}< p_{\gamma,b}$): For this regime, proving that $d_3 <0$ (i.e., $p_{\min}> p_{\gamma,a}$) is an open problem. In other words, this is the only regime that must be addressed to prove Conjecture~\ref{conj:S3}.
	\end{enumerate}
	
	\begin{rem}
		Although for the case $|\cS| \leq 3$ it appears that the {\em over-suppression regime} never occurs for all $\gamma > 0$, the same is not true for $|\cS| \geq 4$. 
		Appendix~\ref{app:N4WithOversuppression} gives an example with $|\cS|=4$ where $p_{\min} < p_{\gamma,a}$.
	\end{rem}

	\subsection{Interplay between $\gamma$, Cardinality of the Support, and~$P_X$}
	We now seek to study the interplay between $\gamma$, the cardinality of the support $\cS$, and the values of $P_X$.  In particular, we provide sufficient conditions under which the over-suppression regime  does not exist. 
	
	\begin{figure}
		\raggedright
		\input{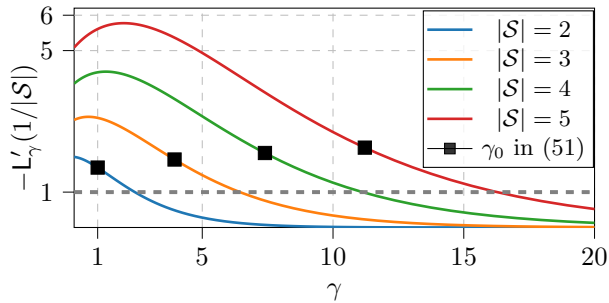}
		\vspace{-0.7cm}
		\caption{Illustration of  conditions in~\eqref{eq:sufficientCondition_p_a=0FurtherBounding} and~\eqref{eq:FinalGamma0Prop}.}
		\vspace{-0.4cm}
		\label{fig:focal_loss_prime}
	\end{figure}
	
	\begin{prop}
		\label{prop:SuffCond1Pa=0}
		Suppose that
		$|\cS| < \infty$ and 
        \begin{equation}
			-p_{\max} \sf  \sfL^{\prime}_\gamma \left( \frac{1}{|\cS|} \right) < 1. \label{eq:sufficientCondition_p_a=0}
		\end{equation}
		Then, $p_{\gamma,a}=0$, that is, the over-suppression regime  does not exist. 
	\end{prop}
	\begin{proof}
		Combining the upper bound on $\alpha_\gamma^\star$ in~\eqref{eq:BoundsOptAlpha} with the assumption in~\eqref{eq:sufficientCondition_p_a=0}, we arrive at $\alpha_\gamma^\star <1$.
		By Lemma~\ref{lemma:Roots},  $\alpha_\gamma^\star <1$ implies that $p_{\gamma,a}=0$. This concludes the proof of Proposition~\ref{prop:SuffCond1Pa=0}.
	\end{proof}
	The condition in~\eqref{eq:sufficientCondition_p_a=0} has three degrees of freedom, namely: $p_{\max}$, $\gamma$, and $|\cS|$. This condition can be further loosened to
	\begin{equation}
		-\sf  \sfL^{\prime}_\gamma \left( \frac{1}{| \cS|} \right)  <1.\label{eq:sufficientCondition_p_a=0FurtherBounding}
	\end{equation}
	From Figure~\ref{fig:focal_loss_prime}, for each $|\cS|$, there exists  $\bar{\gamma}$, that depends only on $|\cS|$, such that $p_{\gamma,a}=0$ for all $\gamma >\bar{\gamma}$, i.e., eventually the over-suppression regime  vanishes.
	
	The next proposition (proof in Appendix~\ref{app:SuffCondGammaMajor}) provides an additional bound on the minimum value of $\gamma$ that ensures that the over-suppression regime cannot occur.
	\begin{prop}
		\label{thm:SuffCondGammaMajor}
		Let $\kappa(\cdot)$ be given in~\eqref{eq:kappaP}. If $\gamma > \kappa(p_{\min})$, then $p_{\min} > p_{\gamma_a}$, that is, the over-suppression regime  does not exist.
	\end{prop}
	We now present a final result giving a sufficient condition for $p_{\max} > p_{\gamma,b}$, i.e., $d_1 \geq 0$. In other words, the next proposition (proof in Appendix~\ref{app:SuffCond1Positive}) ensures that the maximum probability is reduced.
	\begin{prop}
		\label{prop:SuffCond1Positive}
		Suppose that $|\cS| < \infty$ and $\gamma > \gamma_0$ where
		\begin{equation}
			\label{eq:FinalGamma0Prop}
			\gamma_0 = -1 - \frac{1}{n} \rmW_{-1} \left( -n\exp (-n) \right ),
		\end{equation}
		with $n \!=\! \log \left( \frac{|\cS|}{|\cS|-1} \right )$. Then, $p_{\max} \!>\! p_{\gamma, b}$, i.e., $d_1 > 0$.
	\end{prop}
	Figure~\ref{fig:focal_loss_prime} also reports $\gamma_0$ from~\eqref{eq:FinalGamma0Prop} for different values of $|\cS|$. The condition $\gamma \!>\! \gamma_0$ in Proposition~\ref{prop:SuffCond1Positive} is less stringent than~\eqref{eq:sufficientCondition_p_a=0FurtherBounding}, since the latter ensures both $d_1 \geq 0$ and $p_{\gamma,a}\!=\!0$, while the former guarantees only $d_1 \!\ge\! 0$.

	\subsection{Majorization}
	It has been observed empirically~\citep{mukhoti2020calibrating} that the focal-entropy optimizer $P_{\gamma}^\star$ typically has higher entropy than the original $P_X$. Here, we formalize this observation for distributions with finite support size, i.e., $|\cS| < \infty$. We start with the following definition.
	\begin{definition}[Majorization]
		For two vectors $x,y \in \bbR^n$, we say that $x$ majorizes $y$ (notation $x \succ y$) when
		\begin{equation}
			\sum_{i=1}^k x_{(i)} \geq \sum_{i=1}^k y_{(i)}, \ \text{for all} \ k =1,\ldots, n,
		\end{equation}
		where $x_{(i)}$ is the $i$-th largest entry of $x$. 
	\end{definition}
	An immediate consequence of Theorem~\ref{thm:Signdi} is given in the next proposition (proof in Appendix~\ref{app:Majorization}).
	\begin{prop}
		\label{prop:Majorization}
		Suppose that $p_{\min} > p_{\gamma_a}$. Then, $P_X$ majorizes $P_\gamma^\star$,   that is, $P_X \succ P^\star_\gamma$. 
	\end{prop}
	Note that sufficient conditions for which $p_{\min} > p_{\gamma_a}$ are provided in Proposition~\ref{prop:SuffCond1Pa=0} and Proposition~\ref{thm:SuffCondGammaMajor}.
	Combining the result in Proposition~\ref{prop:Majorization} with Schur-convexity, we have the following consequences.
	\begin{corollary}
		\label{cor:SchurConvexity}
		Suppose that $p_{\min} > p_{\gamma_a}$. Then:
		\begin{itemize}
			\item \emph{Entropy increase}: 
			\begin{equation}
				H(P_\gamma^\star) \ge  H(P_X),
				\label{eq:SchurEntropy}
			\end{equation}
			which follows since the Shannon entropy is Schur-concave~\citep{kvaalseth2022entropies}. The above result generalizes to the R\'enyi entropy~\citep{kvaalseth2022entropies}. 
			\item \emph{Non-uniformity decrease}: By subtracting $\log(|\cS|)$ from both sides of~\eqref{eq:SchurEntropy}, we have that
			\begin{equation}
				D_{\sf{KL}}(P_X \| P_U) \ge   D_{\sf{KL}}(P_\gamma^\star \| P_U),
				\label{eq:SchurKL}
			\end{equation} 
            where $P_U$ is the uniform distribution.
		\end{itemize}
	\end{corollary}
	\begin{rem}
Corollary~\ref{cor:SchurConvexity} formalizes the empirical finding of~\citep{mukhoti2020calibrating} that $H(P_\gamma^\star) \ge H(P_X)$. This is significant, since higher-entropy predictions help reduce overconfidence~\citep{mukhoti2020calibrating}.
	\end{rem}

\subsection{Relationship with Relative Entropy} \label{subsec: Relationship with Relative Entropy}
We aim at relating the focal-entropy in~\eqref{eq:FocalEntropy} to the well-studied cross-entropy in~\eqref{eq:CrossEntropy}.
We start by noting that the focal-loss in Definition~\ref{def:FocalLoss} can be equivalently written~as 
\begin{equation}
	\sfL_\gamma (P_X(x)) =  \sfL_0 (P^{(\gamma)}_X(x)) - h_\gamma (P_X), \label{eq:Identi_for_focal_loss}
\end{equation}
where
\begin{equation}
	P^{(\gamma)}_X(x) = \frac{P_{X}(x)^{(1- P_X(x))^\gamma}}{\exp \left( h_\gamma (P_X) \right)},
    \label{eq:PGammaX}
\end{equation}
and
\begin{equation}
	\label{eq:hGamma}
	h_\gamma (P_X) = \log \left( \sum_{x \in \cX}  P_{X}(x)^{(1- P_X(x))^\gamma} \right) .
\end{equation}
The function $h_\gamma(P_X)$ has recently appeared in other contexts, such as in a rate distortion framework with soft-reconstruction~\citep{focal2025lossy}.
The next lemma (proof in Appendix~\ref{app:Prophgamma}) further tightens the upper bound presented in~\citep{focal2025lossy} and shows that the correction term $ h_\gamma(P_X)$ is uniformly bounded.
\begin{lemma} 
	\label{lemma:Prophgamma}
	Let $P_X \in \cP(\cX)$. Then, for $\gamma \geq 0$, the two following properties hold:
	\begin{enumerate}
		\item $\gamma \mapsto h_\gamma(P_X)$ is strictly increasing; and 
		\item $0 \le   h_\gamma(P_X) \le \rme^{ \frac{1}{\rme}} \gamma $.
	\end{enumerate} 	
\end{lemma}
 The following proposition (proof in Appendix~\ref{app:RelatEntropyRelation}) formalizes the connection between the focal-entropy and other information-theoretic quantities. 

\begin{prop} 
		\label{prop:RelatEntropyRelation}
		For any $P_X, Q_X \in \cP(\cX)$ such that $P_X \ll Q_X$, it holds that
		\begin{align}
        \label{eq:rel_focal_ent}
			H_\gamma(P_X,Q_X) &= D_{\sf{KL}}\left(P_X \| Q_X^{(\gamma)}\right ) \!+\! H(P_X) - h_\gamma(Q_X)\\
			&= H_0 \left (P_X ,Q_X^{(\gamma)} \right )  - h_\gamma(Q_X), \label{eq:rel_focal_ent2}
		\end{align}
        where $Q_X^{(\gamma)}$ and $h_\gamma(Q_X)$ are defined in~\eqref{eq:PGammaX} and in~\eqref{eq:hGamma}, respectively. Moreover, it holds that
		\begin{equation}
			H_\gamma(P_X,Q_X) \!=\!\rho_\gamma(P_X,Q_X) \!\left( D_{\sf{KL}}(R_X\| Q_X) \!+\! H(R_X) \right),
		\end{equation}
where
\begin{equation}
		R_X(x)=  
		\frac{P_X(x) (1-Q_X(x))^\gamma}{\rho_\gamma(P_X,Q_X)},
	\end{equation}
	and
	\begin{equation}
		\rho_\gamma(P_X,Q_X) = 
         \bbE_{X \sim P_X} [(1- Q_X(X))^\gamma].
	\end{equation}
	\end{prop}
The next result (proof in Appendix~\ref{app:MinFocalCrossRel}) leverages Lemma~\ref{lemma:Prophgamma} and it establishes bounds on the difference between the minimum value of the cross-entropy and the minimum value of the focal-entropy. 
\begin{prop}  
	\label{prop:MinFocalCrossRel}
	Fix $\gamma \ge 0$. Then, it holds that
	\begin{align}
		0 \le  \inf_{Q_X} H(P_X,Q_X) -  \inf_{Q_X} H_\gamma(P_X,Q_X) \le \rme^{ \frac{1}{\rme}} \gamma. 
	\end{align}
\end{prop}
	
	\section{Experimental Validation}
	\label{sec:ExperimentalValidation}
	In this section, we validate our theoretical results on both {\em synthetic} and {\em real} data. 
	For the synthetic case (Appendix~\ref{app:SyntheticData}), we assume knowledge of the joint distribution over class labels and features, enabling direct theoretical validation. 
	
	For the {\em real} data, we used MNIST in binary classification: $C=0$ for `digit is not 1' and $C=1$ for `digit is 1', yielding class imbalance ($53{,}258$ of $60{,}000$ images in class $0$). From each $28 \times 28$ image, we extracted two zoning-style features capturing spatial asymmetry~\citep{impedovo2014zoning}: $F_1$, the intensity ratio of upper vs.\ lower halves, and $F_2$, the intensity ratio of left vs.\ right halves. Each image was represented by these two features, which were then quantized into four bins (empirical quantiles~\citep{dougherty1995supervised}).
	
	We trained a neural network (same architecture as in Appendix~\ref{app:SyntheticData}) using the focal loss in~\eqref{eq:FocalLoss} with $\gamma=1$. The trained model estimated the probability of $C$ for each of the $16$ bin combinations (four bins for $F_1$, four for $F_2$)—see the `Model $P_{C|F_1,F_2}$' bars in Figure~\ref{fig:MNIST data}.
	
	Since the true posterior is unknown, we approximated it by computing the empirical class proportions within each bin combination (`Estimated $P_{C|F_1,F_2}$' in Figure~\ref{fig:MNIST data}). Using this estimated posterior in~\eqref{eq:optimizer}, we derived $P_{\gamma}^\star$ for $\gamma=1$ (`Theory $P_{C|F_1,F_2}$'). Figure~\ref{fig:MNIST data} shows that the network outputs closely match $P_\gamma^\star$ (maximum difference $0.017$), supporting our theory and indicating convergence to the global minimum under focal loss.

	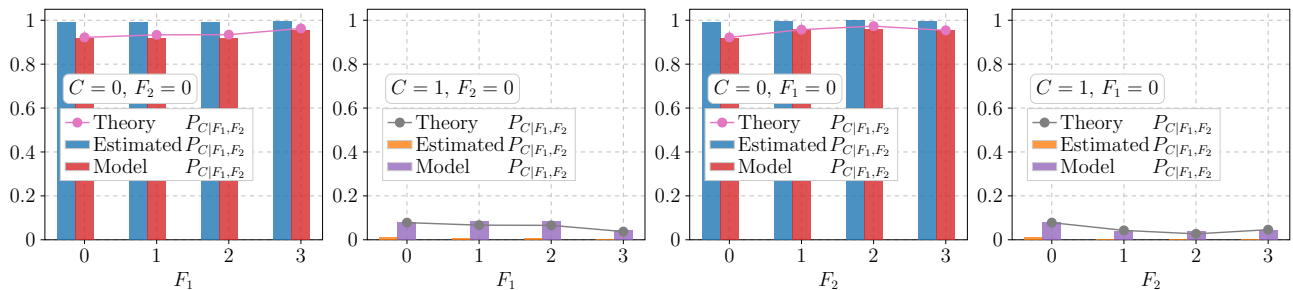
\begin{figure*}[!t]
		\centering
		\begin{subfigure}[t]{0.25\textwidth}\centering
			\resizebox{\linewidth}{!}{
\begin{tikzpicture}

\definecolor{crimson2143940}{RGB}{214,39,40}
\definecolor{darkgray176}{RGB}{176,176,176}
\definecolor{lightgray204}{RGB}{204,204,204}
\definecolor{orchid227119194}{RGB}{227,119,194}
\definecolor{steelblue31119180}{RGB}{31,119,180}

\begin{axis}[
legend cell align={left},
legend style={
  at={(0.75,0.57)},
  anchor=north east,
  fill opacity=1,
  draw opacity=1,
  text opacity=1,
  draw=lightgray204,
  font=\axisfont,
},
tick align=outside,
tick pos=left,
x grid style={dashed, gray!50},
y grid style={dashed, gray!50},
xlabel={$F_1$},
xlabel style={font=\axisfont},
xmajorgrids,
xmin=-0.55, xmax=3.3,
xtick style={color=black},
ymajorgrids,
ymin=0, ymax=1.05,
ytick style={color=black},
tick label style={font=\axisfont},
major tick length=3pt
]

\draw[draw=none,fill=steelblue31119180,fill opacity=0.8] (axis cs:-0.375,0) rectangle (axis cs:-0.125,0.988575458392102);
\draw[draw=none,fill=steelblue31119180,fill opacity=0.8] (axis cs:0.625,0)  rectangle (axis cs:0.875,0.991672218520986);
\draw[draw=none,fill=steelblue31119180,fill opacity=0.8] (axis cs:1.625,0)  rectangle (axis cs:1.875,0.991908293998651);
\draw[draw=none,fill=steelblue31119180,fill opacity=0.8] (axis cs:2.625,0)  rectangle (axis cs:2.875,0.997372262773723);

\draw[draw=none,fill=crimson2143940,fill opacity=0.8] (axis cs:-0.125,0) rectangle (axis cs:0.125,0.920122146606445);
\draw[draw=none,fill=crimson2143940,fill opacity=0.8] (axis cs:0.875,0)  rectangle (axis cs:1.125,0.917494297027588);
\draw[draw=none,fill=crimson2143940,fill opacity=0.8] (axis cs:1.875,0)  rectangle (axis cs:2.125,0.917146027088165);
\draw[draw=none,fill=crimson2143940,fill opacity=0.8] (axis cs:2.875,0)  rectangle (axis cs:3.125,0.954601228237152);

\addplot [line width=1pt, orchid227119194, mark=*, mark size=3, mark options={solid}]
table {%
0 0.921935646251313
1 0.93335977263035
2 0.934313859628219
3 0.962723070048503
};

\node[
  anchor=north west,
  draw=lightgray204,
  fill=white,
  fill opacity=1,
  text opacity=1,
  rounded corners,
  inner sep=4pt,
  font=\axisfont
] at (axis description cs:0.065,0.72) {$C=0$, $F_2=0$};

\addlegendimage{area legend,draw=none,fill=steelblue31119180,fill opacity=0.8}
\addlegendentry{\makebox[2.1cm][l]{Theory} $P_{C|F_1,F_2}$}

\addlegendimage{area legend,draw=none,fill=crimson2143940,fill opacity=0.8}
\addlegendentry{\makebox[2.1cm][l]{Estimated} $P_{C|F_1,F_2}$}

\addlegendimage{line legend,orchid227119194,line width=1pt,mark=*,mark options={solid},mark size=3}
\addlegendentry{\makebox[2.1cm][l]{Model} $P_{C|F_1,F_2}$}

\end{axis}
\end{tikzpicture}}
		\end{subfigure}\hfill
		\begin{subfigure}[t]{0.25\textwidth}\centering
			\resizebox{\linewidth}{!}{
\begin{tikzpicture}

\definecolor{darkgray176}{RGB}{176,176,176}
\definecolor{lightgray204}{RGB}{204,204,204}
\definecolor{darkorange25512714}{RGB}{255,127,14}
\definecolor{mediumpurple148103189}{RGB}{148,103,189}
\definecolor{gray127}{RGB}{127,127,127}

\begin{axis}[
  legend cell align={left},
  legend style={
    at={(0.75,0.57)},
    anchor=north east,
    fill opacity=1,
    draw opacity=1,
    text opacity=1,
    draw=lightgray204,
    font=\axisfont,
  },
  tick align=outside,
  tick pos=left,
  x grid style={dashed, gray!50},
  y grid style={dashed, gray!50},
  xlabel={$F_1$},
  xlabel style={font=\axisfont},
  xmajorgrids,
  xmin=-0.55, xmax=3.3,
  xtick style={color=black},
  ymajorgrids,
  ymin=0, ymax=1.05,
  ytick style={color=black},
  tick label style={font=\axisfont},
  major tick length=3pt
]

\draw[draw=none,fill=darkorange25512714,fill opacity=0.8] (axis cs:-0.375,0) rectangle (axis cs:-0.125,0.0114245416078984);
\draw[draw=none,fill=darkorange25512714,fill opacity=0.8] (axis cs:0.625,0)  rectangle (axis cs:0.875,0.00832778147901399);
\draw[draw=none,fill=darkorange25512714,fill opacity=0.8] (axis cs:1.625,0)  rectangle (axis cs:1.875,0.00809170600134862);
\draw[draw=none,fill=darkorange25512714,fill opacity=0.8] (axis cs:2.625,0)  rectangle (axis cs:2.875,0.00262773722627737);

\draw[draw=none,fill=mediumpurple148103189,fill opacity=0.8] (axis cs:-0.125,0) rectangle (axis cs:0.125,0.0798778682947159);
\draw[draw=none,fill=mediumpurple148103189,fill opacity=0.8] (axis cs:0.875,0)  rectangle (axis cs:1.125,0.0825057476758957);
\draw[draw=none,fill=mediumpurple148103189,fill opacity=0.8] (axis cs:1.875,0)  rectangle (axis cs:2.125,0.0828539803624153);
\draw[draw=none,fill=mediumpurple148103189,fill opacity=0.8] (axis cs:2.875,0)  rectangle (axis cs:3.125,0.0453987643122673);

\addplot [line width=1pt, gray127, mark=*, mark size=3, mark options={solid}]
table {%
0 0.0780643537486867
1 0.06664022736965
2 0.0656861403717812
3 0.0372769299514973
};

\node[
  anchor=north west,
  draw=lightgray204,
  fill=white,
  fill opacity=1,
  text opacity=1,
  rounded corners,
  inner sep=4pt,
  font=\axisfont
] at (axis description cs:0.065,0.72) {$C=1$, $F_2=0$};

\addlegendimage{area legend,draw=none,fill=darkorange25512714,fill opacity=0.8}
\addlegendentry{\makebox[2.1cm][l]{Theory} $P_{C|F_1,F_2}$}

\addlegendimage{area legend,draw=none,fill=mediumpurple148103189,fill opacity=0.8}
\addlegendentry{\makebox[2.1cm][l]{Estimated} $P_{C|F_1,F_2}$}

\addlegendimage{line legend,gray127,line width=1pt,mark=*,mark options={solid},mark size=3}
\addlegendentry{\makebox[2.1cm][l]{Model} $P_{C|F_1,F_2}$}

\end{axis}
\end{tikzpicture}}
		\end{subfigure}\hfill
		\begin{subfigure}[t]{0.25\textwidth}\centering
			\resizebox{\linewidth}{!}{
\begin{tikzpicture}

\definecolor{crimson2143940}{RGB}{214,39,40}
\definecolor{darkgray176}{RGB}{176,176,176}
\definecolor{lightgray204}{RGB}{204,204,204}
\definecolor{orchid227119194}{RGB}{227,119,194}
\definecolor{steelblue31119180}{RGB}{31,119,180}

\begin{axis}[
  legend cell align={left},
  legend style={
    at={(0.75,0.57)},
    anchor=north east,
    fill opacity=1,
    draw opacity=1,
    text opacity=1,
    draw=lightgray204,
    font=\axisfont,
  },
  tick align=outside,
  tick pos=left,
  x grid style={dashed, gray!50},
  y grid style={dashed, gray!50},
  xlabel={$F_2$},
  xlabel style={font=\axisfont}, 
  xmajorgrids,
  xmin=-0.55, xmax=3.3,
  xtick style={color=black},
  ymajorgrids,
  ymin=0, ymax=1.05,
  ytick style={color=black},
  tick label style={font=\axisfont},
  major tick length=3pt
]

\draw[draw=none,fill=steelblue31119180,fill opacity=0.8] (axis cs:-0.375,0) rectangle (axis cs:-0.125,0.988575458392102);
\draw[draw=none,fill=steelblue31119180,fill opacity=0.8] (axis cs:0.625,0)  rectangle (axis cs:0.875,0.996560350218887);
\draw[draw=none,fill=steelblue31119180,fill opacity=0.8] (axis cs:1.625,0)  rectangle (axis cs:1.875,0.998619737750173);
\draw[draw=none,fill=steelblue31119180,fill opacity=0.8] (axis cs:2.625,0)  rectangle (axis cs:2.875,0.99601593625498);

\draw[draw=none,fill=crimson2143940,fill opacity=0.8] (axis cs:-0.125,0) rectangle (axis cs:0.125,0.920122146606445);
\draw[draw=none,fill=crimson2143940,fill opacity=0.8] (axis cs:0.875,0)  rectangle (axis cs:1.125,0.960745930671692);
\draw[draw=none,fill=crimson2143940,fill opacity=0.8] (axis cs:1.875,0)  rectangle (axis cs:2.125,0.960874259471893);
\draw[draw=none,fill=crimson2143940,fill opacity=0.8] (axis cs:2.875,0)  rectangle (axis cs:3.125,0.955258786678314);

\addplot [line width=1pt, orchid227119194, mark=*, mark size=3, mark options={solid}]
table {%
0 0.921935646251313
1 0.957294679300958
2 0.973078982571224
3 0.954008887926193
};

\node[
  anchor=north west,
  draw=lightgray204,
  fill=white,
  fill opacity=1,
  text opacity=1,
  rounded corners,
  inner sep=4pt,
  font=\axisfont
] at (axis description cs:0.065,0.72) {$C=0$, $F_1=0$};

\addlegendimage{area legend,draw=none,fill=steelblue31119180,fill opacity=0.8}
\addlegendentry{\makebox[2.1cm][l]{Theory} $P_{C|F_1,F_2}$}

\addlegendimage{area legend,draw=none,fill=crimson2143940,fill opacity=0.8}
\addlegendentry{\makebox[2.1cm][l]{Estimated} $P_{C|F_1,F_2}$}

\addlegendimage{line legend,orchid227119194,line width=1pt,mark=*,mark options={solid},mark size=3}
\addlegendentry{\makebox[2.1cm][l]{Model} $P_{C|F_1,F_2}$}

\end{axis}
\end{tikzpicture}}
		\end{subfigure}\hfill
		\begin{subfigure}[t]{0.25\textwidth}\centering
			\resizebox{\linewidth}{!}{
\begin{tikzpicture}

\definecolor{darkgray176}{RGB}{176,176,176}
\definecolor{lightgray204}{RGB}{204,204,204}
\definecolor{darkorange25512714}{RGB}{255,127,14}
\definecolor{mediumpurple148103189}{RGB}{148,103,189}
\definecolor{gray127}{RGB}{127,127,127}

\begin{axis}[
  legend cell align={left},
  legend style={
    at={(0.75,0.57)},
    anchor=north east,
    fill opacity=1,
    draw opacity=1,
    text opacity=1,
    draw=lightgray204,
    font=\axisfont,
  },
  tick align=outside,
  tick pos=left,
  x grid style={dashed, gray!50},
  y grid style={dashed, gray!50},
  xlabel={$F_2$},
  xlabel style={font=\axisfont}, 
  xmajorgrids,
  xmin=-0.55, xmax=3.3,
  xtick style={color=black},
  ymajorgrids,
  ymin=0, ymax=1.05,
  ytick style={color=black},
  tick label style={font=\axisfont},
  major tick length=3pt
]

\draw[draw=none,fill=darkorange25512714,fill opacity=0.8] (axis cs:-0.375,0) rectangle (axis cs:-0.125,0.0114245416078984);
\draw[draw=none,fill=darkorange25512714,fill opacity=0.8] (axis cs:0.625,0)  rectangle (axis cs:0.875,0.0034396497811132);
\draw[draw=none,fill=darkorange25512714,fill opacity=0.8] (axis cs:1.625,0)  rectangle (axis cs:1.875,0.00138026224982747);
\draw[draw=none,fill=darkorange25512714,fill opacity=0.8] (axis cs:2.625,0)  rectangle (axis cs:2.875,0.00398406374501992);

\draw[draw=none,fill=mediumpurple148103189,fill opacity=0.8] (axis cs:-0.125,0) rectangle (axis cs:0.125,0.0798778682947159);
\draw[draw=none,fill=mediumpurple148103189,fill opacity=0.8] (axis cs:0.875,0)  rectangle (axis cs:1.125,0.0392540469765663);
\draw[draw=none,fill=mediumpurple148103189,fill opacity=0.8] (axis cs:1.875,0)  rectangle (axis cs:2.125,0.0391258001327515);
\draw[draw=none,fill=mediumpurple148103189,fill opacity=0.8] (axis cs:2.875,0)  rectangle (axis cs:3.125,0.0447412356734276);

\addplot [line width=1pt, gray127, mark=*, mark size=3, mark options={solid}]
table {%
0 0.0780643537486867
1 0.0427053206990422
2 0.0269210174287764
3 0.045991112073807
};

\node[
  anchor=north west,
  draw=lightgray204,
  fill=white,
  fill opacity=0.95,
  text opacity=1,
  rounded corners,
  inner sep=4pt,
  font=\axisfont
] at (axis description cs:0.065,0.72) {$C=1$, $F_1=0$};

\addlegendimage{area legend,draw=none,fill=darkorange25512714,fill opacity=0.8}
\addlegendentry{\makebox[2.1cm][l]{Theory} $P_{C|F_1,F_2}$}

\addlegendimage{area legend,draw=none,fill=mediumpurple148103189,fill opacity=0.8}
\addlegendentry{\makebox[2.1cm][l]{Estimated} $P_{C|F_1,F_2}$}

\addlegendimage{line legend,gray127,line width=1pt,mark=*,mark options={solid},mark size=3}
\addlegendentry{\makebox[2.1cm][l]{Model} $P_{C|F_1,F_2}$}

\end{axis}
\end{tikzpicture}}
		\end{subfigure}
		\vspace{-0.7cm}
		\caption{ MNIST dataset: 
			For each class, we fixed one of the features at its modal value and plotted: the true posterior, the posterior predicted by the trained neural network, and the theoretical focal-entropy minimizer. {``Model $P_{C|F_1,F_2}$'' denotes the predicted probability of class $C$ for each of the $16$ possible bin combinations (i.e., four bins for $F_1$ and four bins for $F_2$), which is estimated by using a neural network (see Appendix~\ref{app:SyntheticData} for the architecture) trained by using the focal-loss in~\eqref{eq:FocalLoss} with $\gamma=1$; ``Estimated $P_{C|F_1,F_2}$'' refers to an empirical approximation of the posterior (since the true posterior is unknown), which is obtained by computing, for each of the $16$ bin combinations, the proportion of samples with $C=0$ and $C=1$; and ``Theory $P_{C|F_1,F_2}$'' is the $P_\gamma^\star$ in~\eqref{eq:optimizer} obtained by using ``Estimated $P_{C|F_1,F_2}$'' as input.}}
		\label{fig:MNIST data}
		\vspace{-0.3cm}
	\end{figure*}
	
	\section{Conclusion}
	\label{sec:Conslusion}
	We conducted a systematic study of properties of the focal-entropy. We showed that it preserves desirable features (finiteness and convexity), while also exhibiting behaviors that distinguish it fundamentally from the cross-entropy. We established the existence and uniqueness of its minimizer and described how it reshapes distributions by amplifying mid-range probabilities, suppressing large probability outcomes, and, under certain conditions, over-suppressing very small probabilities. 
    This
finding is particularly important for practitioners, as it highlights the need to carefully select $\gamma$ to avoid the \emph{over-suppressing} regime.
    These results clarify the trade-offs that the focal-loss introduces and provide a theoretical foundation for its use in imbalanced learning. We hope that this perspective will guide both further theoretical developments and support more principled practical applications of the focal-loss.
	{An interesting direction for future work is to investigate the potential benefits (if any) of the focal-loss in classification problems with soft, real-valued labels~\citep{jeong2023demystifying}.}

\appendix

\onecolumn

\section{Proofs for Section~\ref{sec:AnalyticalPropFL}}

\subsection{Proof of Proposition~\ref{prop:AnalyticalPropFocalLoss}}
\label{app:AnalyticalPropFocalLoss}
To show the first property, note that for $p \in (0,1)$, the function $\gamma \mapsto p^{\left(1-p\right)^\gamma}$ is strictly increasing and hence, $\gamma \mapsto \sfL_\gamma(p)$ is strictly decreasing.

The first and second derivatives can be obtained through a straightforward calculation. It is also trivial to see that $\sfL^{\prime}_\gamma(p) < 0$ for all $p \in (0,1)$ and $\gamma \ge 0$.

We now show that $\sfL^{\prime \prime}_\gamma(p) > 0$ for all $p \in (0,1)$ and $\gamma \ge 0$, which implies that $p \;\mapsto\; \sfL_\gamma(p)$ is strictly convex on $(0,1)$. 
We consider two cases separately:
\begin{itemize}
	\item {\bf{Case} $\gamma > 1$ or $\gamma=0$:} it is trivial to see that the right-hand side of~\eqref{eq:secondDer} is always positive;
	\item {\bf{Case}  $0 < \gamma \leq 1$}: by using the fact that $\log(p) \geq 1-\frac{1}{p}$ for all $p>0$, we can lower bound the second order derivative in~\eqref{eq:secondDer} as follows,
	\begin{align}
		\sfL^{\prime \prime}_\gamma(p) & \ge  \frac{(1-p)^{\gamma-2}}{p^2} \left[\gamma(1-\gamma)\,p^2 \left( 1- \frac{1}{p}\right ) + 2\gamma\,p(1-p) + (1-p)^2\right]
		\\& = \frac{(1-p)^{\gamma-1}}{p^2} \left( \gamma^2 \; p +\gamma \; p + 1-p\right), 
	\end{align}
	which is always positive.
\end{itemize}\
This concludes the proof of Proposition~\ref{prop:AnalyticalPropFocalLoss}.

\subsection{Proof of Proposition~\ref{prop:AnalyticalPropFocalLossInverse}}
\label{app:AnalyticalPropFocalLossInverse}
The first property is directly derived from the inspection of~\eqref{eq:firstDer}.

For the second property, from~\eqref{eq:secondDer}, we have that $\sfL^{\prime \prime}_\gamma(p) > 0 $ on $p \in (0,1)$. Thus, the function $p \;\mapsto\;  \sfL^\prime_\gamma(p)$ is  strictly increasing on $p \in (0,1)$ and continuous. It therefore follows that $\left( \sfL^{\prime}_\gamma \right )^{-1}$ is also a strictly increasing function on its domain (i.e., the range of $\sfL^{\prime}_\gamma$), that is, $t \;\mapsto\; \left( \sfL^{\prime}_\gamma \right )^{-1}(t)$ is strictly increasing for $t<0$. 

Finally, the third property is a direct consequence of the first three properties. 

This concludes the proof of Proposition~\ref{prop:AnalyticalPropFocalLossInverse}.
\subsection{Proof of Proposition~\ref{prop:properties_of_phi_gamma}}
\label{app:properties_of_phi_gamma}

Property 1) follows by inspection.

To show Property 2), we start by proving that $p \mapsto \kappa(p)$ is strictly decreasing. By taking the first derivative of $\kappa(p)$ in~\eqref{eq:kappaP}, we obtain
\begin{equation}
	\kappa^\prime(p) = \frac{2p-1-\left( \log \left( \frac{1}{p}\right )-1 \right )^2}{p^2 \log^2 \left( \frac{1}{p} \right )}.
\end{equation}
The denominator is always positive and hence, we focus on studying the sign of the numerator, denoted as $f(p)$. We have that
$f^\prime(p) = \frac{2 \left( p-\log(p)-1\right )}{p}>0$, where we have used $\log \left(\frac{1}{p} \right ) > 1-p$ for all $p \in (0,1)$. Thus, $p \mapsto f(p)$ is strictly increasing, i.e., its maximum is $f(1)=0$. It therefore follows that $f(p)<0$ for all $p \in (0,1)$ and hence, $\kappa^\prime(p) <0$ for all $p \in (0,1)$, that is, $p \mapsto \kappa(p)$ is strictly decreasing.

We now prove that $\kappa(p) \leq 0$ for all $p \in [\rmw,1)$. Since $p \mapsto \kappa(p)$ is strictly decreasing, it suffices to show that $\kappa(\rmw) = 0$.
Setting~\eqref{eq:kappaP} equal to zero, we arrive at
\begin{equation}
	2p - \log(p)-2 = 0,
\end{equation}
or equivalently, after some manipulation,
\begin{equation}
	(-2p) \; {\rme}^{-2p} = -2 \; \rme^{-2},
\end{equation}
which is satisfied for $p = \rmw$ and $p=1$.

Now, by using the first and second derivatives of $\sfL_\gamma$ in~\eqref{eq:firstDer} and in~\eqref{eq:secondDer}, we have that
\begin{align}
	\phi^\prime (p) &= \gamma \left( 1-p\right )^{\gamma-2} \left ( \log \left( \frac{1}{p}\right ) (1-\gamma \, p) - 2(1-p) \right ) \\
	&= \gamma \left( 1-p\right )^{\gamma-2} p \log \left( \frac{1}{p} \right ) \left(\kappa(p)-\gamma \right),
\end{align}
which implies that the above derivative is negative whenever $\gamma > \kappa(p)$ and positive if $\gamma < \kappa(p)$.  

To show Property 3), note that
\begin{align}
	\phi_{\gamma}(p) &= p \; (1-p)^{\gamma-1} \left( \gamma \log \frac{1}{p} +\frac{1-p}{p}\right)  \\
	& \le (1-p)^{\gamma} (1+\gamma)\label{eq:bound_on_phi_max}\\
	& \le 1+\gamma,
\end{align}
where~\eqref{eq:bound_on_phi_max} follow by using the bound $\log x \le x-1, \, x>0$.

We now show Property 4). To show the first upper bound, note that by Property 2), $\phi_\gamma(p)$ is decreasing on $p \in [\rmw,1)$ since $\kappa(p)<0$. Hence, $p_\gamma^+ \le \rmw$.

The second upper bound follows by noting that $p^+_\gamma$ is such that $\kappa(p^+_\gamma) = \gamma$,
which implies that
\begin{equation}
	\frac{1}{p^+_\gamma}- \frac{2(1-p^+_\gamma)}{p^+_\gamma \log \left( \frac{1}{p^+_\gamma} \right )} = \gamma.
\end{equation}
By using the change of variable $p^+_\gamma = {\rme}^{-x}, x>0$, we can rewrite the above as follows,
\begin{equation}
	\gamma = {\rme}^x - \frac{2\left( \rme^x - 1\right )}{x}
	\leq \rme^x - 2 \; {\rme}^{\frac{x}{2}},
\end{equation}
where the inequality follows since $\frac{\rme^x - 1}{x} \geq {\rme}^{\frac{x}{2}}$ for $x>0$. By letting $y = {\rme}^{\frac{x}{2}}$, the above inequality can be rewritten as
\begin{equation}
	y^2 -2 y - \gamma \geq 0,
\end{equation}
which, since $\gamma \geq 0$, implies that
\begin{equation}
	y \geq 1+\sqrt{1+\gamma},
\end{equation}
and hence, proves the upper bound on $p^+_\gamma$.

To show Property 5), we use~\eqref{eq:PhiFunc} and write
\begin{align}
	\frac{\phi_\gamma(u)}{\phi_\gamma \left(\frac{1-u}{2}\right)} &= \frac{(1-u)^\gamma \left( \gamma \frac{u}{1-u} \log \left( \frac{1}{u} \right ) +1\right )}{\left(\frac{1+u}{2} \right )^\gamma \left(\gamma \frac{1-u}{1+u} \log \left( \frac{2}{1-u}\right ) +1  \right )}
	\\& = \left( \frac{2}{s_u} \right )^\gamma \left( \frac{1+ \gamma \frac{\log(1+t_u)}{t_u}}{1+ \gamma \frac{\log(1+s_u)}{s_u}} \right ),
\end{align}
where we let $t_u = \frac{1-u}{u}$ and $s_u = \frac{1+u}{1-u}$. By leveraging the equality above and letting $\sfT_u = \frac{\log(1+t_u)}{t_u}$ and $\sfS_u = \frac{\log(1+s_u)}{s_u}$, we can write
\begin{align}
	\phi_\gamma(u) - \phi_\gamma \left(\frac{1-u}{2}\right) 
	= \frac{\phi_\gamma \left( \frac{1-u}{2}\right)}{1+\gamma \sfS_u} \left [ \left( \frac{2}{s_u} \right )^\gamma (1+\gamma \sfT_u) - (1+\gamma \sfS_u) \right ].
	\label{eq:IntermediateDifferencePhy}
\end{align}
Our goal is to prove that the above expression is non-negative for all $u \in (0,1/3]$. We note that from Property 1) and the fact that $\sfS_u>0$, we have that $\frac{\phi_\gamma \left( \frac{1-u}{2}\right)}{1+\gamma \sfS_u} \geq 0$. Thus, proving that~\eqref{eq:IntermediateDifferencePhy} is non-negative for all $u \in (0,1/3 ]$ is equivalent to prove that
\begin{equation}
	f_u(\gamma) = \left( \frac{2}{s_u} \right )^\gamma (1+\gamma \sfT_u) - (1+\gamma \sfS_u)
	\label{eq:FGammaProp5}
\end{equation}
is non-negative for all $u \in (0,1/3 ]$. 
We next prove that this is the case by showing that $\gamma \mapsto f_u(\gamma)$ is non-decreasing for all $u \in (0,1/3 ]$, which implies that
\begin{equation}
	f_u(\gamma) \geq f_u(0) =0,
\end{equation}
for all $u \in (0,1/3 ]$.
To show that $\gamma \mapsto f_u(\gamma)$ in~\eqref{eq:FGammaProp5} is non-decreasing for all $u \in (0,1/3 ]$, it suffices to show that $f_u^\prime(\gamma) \geq 0$ for all $u \in (0,1/3 ]$.
To prove this, we first show that $\gamma \mapsto f_u^\prime(\gamma)$ is strictly increasing and then show that $f_u^\prime(0) \geq 0$ for all $u \in (0,1/3 ]$.

We have that
\begin{equation}
	f_u^\prime(\gamma) = \log \left( \frac{2}{s_u}\right ) \left( \frac{2}{s_u} \right )^\gamma (1+\gamma \sfT_u) + \left( \frac{2}{s_u} \right )^\gamma \sfT_u -\sfS_u,
\end{equation}
and
\begin{equation}
	f_u^{\prime \prime}(\gamma) =\left( \frac{2}{s_u} \right )^\gamma \log \left( \frac{2}{s_u}\right ) \left [ \log \left( \frac{2}{s_u}\right ) (1+\gamma \sfT_u) + 2\sfT_u\right ].
\end{equation}
Now, since $s_u \leq 2$ and $\sfT_u > 0$ for all $u \in (0,1/3 ]$, it follows that $f_u^{\prime \prime}(\gamma) > 0$ for all $u \in (0,1/3 ]$. Thus, $\gamma \mapsto f_u^\prime(\gamma)$ is strictly increasing for all $u \in (0,1/3 ]$.

We now show that, for all $u \in (0,1/3 ]$, it holds that $f_u^\prime(0) \geq 0$. We have that
\begin{align}
	f_u^\prime(0) & = \log \left( \frac{2}{s_u}\right )  +  \sfT_u -\sfS_u
	\\& = \log \left( \frac{2}{s_u}\right ) + \frac{\log(1+t_u)}{t_u} - \frac{\log(1+s_u)}{s_u}
	\\& =\frac{1}{(1-u^2)} \left( 2u(1-u)\log(2)-(1-u^2)\log(1+u) + 2(1-u)\log(1-u)-u(1+u)\log(u) \right ),
	\label{eq:DerivativefuGamma0}
\end{align}
where the last equality follows by substituting the definitions of $t_u$ and $s_u$.
Note that the denominator of $f_u^\prime(0)$ is always positive for all $u \in (0,1/3 ]$. Therefore, to prove $f_u^\prime(0) \geq 0$ for all $u \in (0,1/3 ]$, it suffices to prove that the numerator of~\eqref{eq:DerivativefuGamma0} is non-negative for all $u \in (0,1/3 ]$.
We denote the numerator of~\eqref{eq:DerivativefuGamma0} by $g_u(0)$.
In what follows, we will show that, for all $u \in (0,1/3 ]$, we have that $g_u(0)$ is concave and hence, it achieves minimum value at one of its end points, which are
\begin{equation}
	g_{0^+}(0) = g_{1/3}(0) = 0.
\end{equation}
To prove concavity, we take the second derivative of $g_u(0)$. It is a simple exercise to obtain
\begin{equation}
	\frac{\rmd^2}{\rmd u^2}g_u(0) = \frac{1-5u^2 -2u(1-u^2) \log \left( \frac{1+u}{4u} \right)}{u(u^2-1)}.
	\label{eq:DenSecDerivativeguGamma0}
\end{equation}
For all $u \in (0,1/3 ]$, the denominator of the above expression is always negative and hence, to prove concavity of $g_u(0)$ it suffices to prove that the numerator of~\eqref{eq:DenSecDerivativeguGamma0} is positive for all $u \in (0,1/3 ]$. We denote the numerator of~\eqref{eq:DenSecDerivativeguGamma0} by $h_u(0)$. We have that
\begin{align}
	h_u(0) & = 1-5u^2 -2u(1-u^2) \log \left( \frac{1+u}{4u} \right)
	\\& \geq 1-5u^2 -2u(1-u^2) \left( \frac{1+u}{4u}-1\right)
	\\& = \frac{1}{2}(1+3u-9u^2-3u^3)
	\\& = \frac{1}{2} ((1-9u^2)+3u(1-u^2))
	\\& >0,
\end{align}
where the first inequality follows by using the fact that $\log(x)\leq x-1$ for all $x>0$ and the second inequality is due to the fact that, for all $u \in (0,1/3 ]$, we have that  $(1-9u^2)\geq0$ and $3u(1-u^2)>0$.
Thus, for all $u \in (0,1/3 ]$, $h_u(0)>0$, which implies that $g_u(0)$ is concave with $g_u(0)\geq 0$ and hence, $f_u^\prime(0) \geq 0$, concluding the proof that $\gamma \mapsto f_u(\gamma)$ in~\eqref{eq:FGammaProp5} is non-decreasing for all $u \in (0,1/3 ]$.
This concludes the proof of Proposition~\ref{prop:properties_of_phi_gamma}.

\section{Proofs for Section~\ref{sec:focal_entropy}}

\subsection{Additional Properties of the Mapping $\gamma \mapsto H_\gamma(P_X,Q_X)$}
\label{app: focal-entropy vs gamma}
Consider the mapping $\gamma \mapsto H_\gamma(P_X,Q_X)$. The first and second order derivatives are given by
\begin{align}
	\frac{\rmd}{\rmd \gamma} H_\gamma(P_X,Q_X) & =\bbE \left [\log (1-Q_X(X)) \sfL_\gamma(Q_X(X)) \right ] \le 0, 
	\\
	\frac{\mathrm{d}^2}{\mathrm{d}\gamma^2} H_\gamma(P_X,Q_X) &= \bbE \left [\left(\log (1-Q_X(X) \right)^2 \sfL_\gamma(Q_X(X)) \right ] \ge 0,
\end{align}
where the expectations are taken with respect to $X \sim P_X$.

\subsection{Proof of Proposition~\ref{Prop:LimitGammaInf}}
\label{app:LimitGammaInf}
Define a positive measure
\begin{equation}
	\mu(x) = P_X(x) \log \frac{1}{Q_X(x)}, \, x \in \cX,
\end{equation}
which is well-defined, in view of the fact that $P_X \ll Q_X$, and has the same support as $P_X$, namely $\cS$. Moreover, $\mu$ is a finite measure since, from~\eqref{eq:CrossEntropy}, we have that
\begin{equation}
	\sum_x \mu(x) = H(P_X,Q_X)<\infty, 
\end{equation}
where the inequality holds from the assumption. 
Next, note that we can write 
\begin{align}
	H^{ \frac{1}{\gamma}}_\gamma(P_X,Q_X) &= \left( \sum_{x} \mu(x)  (1-Q_X(x))^\gamma \right)^{ \frac{1}{\gamma}} \nonumber
	\\& = \| 1-Q_X(\cdot) \|_\gamma,
\end{align}
where $\| \cdot\|_p$ is the $\ell_p(\cX, \mu), \, p>0,$ norm.   Now, taking
the limit for $\gamma \to \infty$ and using the standard result that $\ell_p$ norms converge to the essential supremum of the underlying function~\citep[Ch.~6, Ex.~7]{folland1999real}, we arrive at 
\begin{align}
	\lim_{\gamma \to \infty} \| 1-Q_X(\cdot) \|_\gamma &= \sup_{ x \in  \text{ support } \mu } (1-Q_X(x))\\
	&=\sup_{ x \in \cS } \,( 1-Q_X(x) ), \label{eq:ess_sup}
\end{align}
where in the last step we have used the fact that $\mu$ is supported on $\cS$. 
The proof of Proposition~\ref{Prop:LimitGammaInf} is concluded by noting that~\eqref{eq:ess_sup} is equal to one if $\cS$ is countably infinite. 

\noindent Figure~\ref{fig:focalEntropy_vs_gamma} provides an illustration of the convergence of $H_\gamma^{ \frac{1}{\gamma}}(P_X, Q_X)$ for a fixed $P_X$ with $|\cS|=3$ and two different distributions $Q_X$. From Figure~\ref{fig:focalEntropy_vs_gamma}, we note that this convergence is pretty fast.
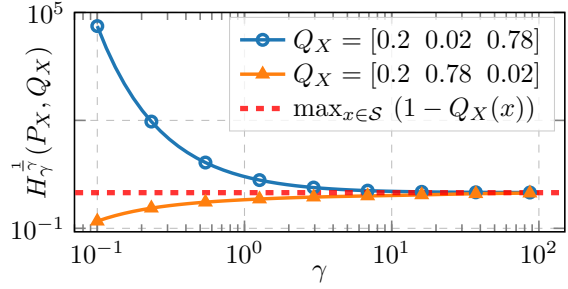
\begin{figure}
	\centering
\begin{tikzpicture}

\definecolor{steelblue31119180}{RGB}{31,119,180}
\definecolor{darkorange25512714}{RGB}{255,127,14}
\definecolor{lightgray204}{RGB}{204,204,204}

\begin{axis}[
  width=8.1cm,
  height=4.5cm,
  xlabel={$\gamma$}, 
  xlabel style={yshift=0.2cm},
  ylabel={$H_\gamma^{ \frac{1}{\gamma}}(P_X, Q_X)$},
  ylabel style={font=\normalsize, yshift=-0.8cm},
  tick label style={font=\normalsize},
  axis line style={semithick},
  tick style={semithick},
  xmode=log, ymode=log,
  log basis x={10}, log basis y={10},
  xmin=0.07, xmax=150,
  ymin=0.08, ymax=1e5,
  ytick={1e-1,1e2,1e5},
yticklabels={$10^{-1}$,, $10^{5}$}, 
  xmajorgrids, ymajorgrids,
  grid style={dashed, gray!50},
  legend cell align={left},
  legend style={fill opacity=0.8, draw opacity=1, text opacity=1, draw=lightgray204}
]

\addplot [line width=1.3pt, steelblue31119180, mark=o, mark size=2, mark repeat=6]
table {%
0.1 41815.2099713677
0.115139539932645 10230.3136842604
0.132571136559011 3011.97272896198
0.152641796717523 1041.47581744061
0.175751062485479 414.097159211114
0.202358964772516 185.881625373336
0.232995181051537 92.7100318206413
0.268269579527973 50.6709170745002
0.308884359647748 29.9853642818701
0.355648030622313 19.0126324237646
0.409491506238043 12.8002396590885
0.471486636345739 9.07853434052229
0.542867543932386 6.73696357526219
0.625055192527397 5.19976236225951
0.719685673001152 4.15272573348113
0.828642772854684 3.41639470831827
0.954095476349994 2.88402256143293
1.09854114198756 2.48976528919902
1.2648552168553 2.19166375318739
1.45634847750124 1.9621840290451
1.67683293681101 1.78277373567424
1.93069772888325 1.64062930957344
2.22299648252619 1.52672045895369
2.55954792269954 1.43455085862462
2.94705170255181 1.3593617841465
3.39322177189533 1.29760888561365
3.90693993705462 1.24661125664011
4.49843266896945 1.20431149304263
5.17947467923121 1.16910867724524
5.96362331659464 1.13974020727724
6.866488450043 1.11519698676401
7.9060432109077 1.09466188489316
9.10298177991522 1.07746481165396
10.4811313415469 1.06304995416708
12.0679264063933 1.05095210185423
13.8949549437314 1.0407798044539
15.9985871960606 1.03220354007455
18.4206996932672 1.02494728952834
21.2095088792019 1.018782085965
24.4205309454865 1.01352037167178
28.1176869797423 1.00901040183878
32.3745754281764 1.00513043356502
37.2759372031494 1.00178288877416
42.9193426012878 0.998888935482085
49.4171336132383 0.996383925368967
56.8986602901829 0.994213912249355
65.5128556859551 0.992333202468516
75.4312006335462 0.990702700410112
86.8511373751352 0.989288769587767
100 0.98806239268361
};
\addlegendentry{$Q_X = [0.2\;\;0.02\;\;0.78]$}

\addplot [line width=1.3pt, darkorange25512714, mark=triangle*, mark size=2, mark repeat=6]
table {%
0.1 0.158007997493135
0.115139539932645 0.191206916570655
0.132571136559011 0.225757024157375
0.152641796717523 0.260930243913529
0.175751062485479 0.296073805128654
0.202358964772516 0.330638932559467
0.232995181051537 0.364193932990498
0.268269579527973 0.396424891151601
0.308884359647748 0.427127563339744
0.355648030622313 0.45619376217457
0.409491506238043 0.483594893321555
0.471486636345739 0.509364578374228
0.542867543932386 0.533581629579394
0.625055192527397 0.556354111735168
0.719685673001152 0.57780486630062
0.828642772854684 0.598058685945421
0.954095476349994 0.61723130037572
1.09854114198756 0.635420436096679
1.2648552168553 0.652699390780412
1.45634847750124 0.669113732084911
1.67683293681101 0.684681773532731
1.93069772888325 0.699399269243722
2.22299648252619 0.713248224847442
2.55954792269954 0.726208889903707
2.94705170255181 0.738273106113823
3.39322177189533 0.749456613130343
3.90693993705462 0.759808027565322
4.49843266896945 0.769413135811304
5.17947467923121 0.778394594489395
5.96362331659464 0.786908464958001
6.866488450043 0.795139472935777
7.9060432109077 0.803295982816178
9.10298177991522 0.811603393036967
10.4811313415469 0.820291562043252
12.0679264063933 0.829569425098702
13.8949549437314 0.839581328192603
15.9985871960606 0.850349540271048
18.4206996932672 0.861726696938669
21.2095088792019 0.873396198069705
24.4205309454865 0.884941734171524
28.1176869797423 0.895958922134927
32.3745754281764 0.906149229616646
37.2759372031494 0.915355498173345
42.9193426012878 0.923543342124458
49.4171336132383 0.930758170508614
56.8986602901829 0.9370842139243
65.5128556859551 0.942617081435143
75.4312006335462 0.947449739227422
86.8511373751352 0.951667189819782
100 0.955345340368282
};
\addlegendentry{$Q_X = [0.2\;\;0.78\;\;0.02]$}

\addplot [line width=2pt, red, dashed, opacity=0.8]
table {%
0.0707945784384138 0.98
141.253754462275 0.98
};
\addlegendentry{$\max_{x \in \cS} \ (1-Q_{X}(x))$}

\end{axis}
\end{tikzpicture}
	\vspace{-0.3cm}
	\caption{$P_X = \begin{bmatrix}0.4 & 0.58 & 0.02\end{bmatrix}$.}
	\vspace{-0.4cm}
	\label{fig:focalEntropy_vs_gamma}
\end{figure} 
\subsection{Proof of Proposition~\ref{prop:focal-finite-iff-cross-finite}}
\label{app:focal-finite-iff-cross-finite}
If $P_X$ is not absolutely continuous with respect to  $Q_X$, the statement trivially holds by the definition in~\eqref{eq:FocalEntropy}.     
Hence, suppose that $P_X \ll Q_X$. Then, on the one hand, from  Proposition~\ref{prop: focal-entropy vs gamma},  we have that 
\begin{equation}
	H_\gamma(P_X,Q_X)
	\;\le\;
	H(P_X,Q_X). \label{Eq:use_monote_In_finite}
\end{equation}
On the other hand, fix some $\delta >0$ and let $\cQ_\delta = \{x \in \cX :Q_X(x) \le \delta\}$. Then,
\begin{align}
	H_\gamma(P_X, Q_X) & \geq \sum_{x \in \cQ_\delta} P_X(x) \left( 1-Q_X(x)\right )^\gamma \log \left( \frac{1}{Q_X(x)}\right ) \notag
	\\& \ge  \left( 1-\delta\right )^\gamma \sum_{x \in \cQ_\delta} P_X(x)  \log \left( \frac{1}{Q_X(x)}\right )
	\\& = (1-\delta)^\gamma \left( H(P_X,Q_X) -\Delta_\delta \right), \label{eq:finitness_of_cross}
\end{align}
where we have defined
\begin{equation}
	\Delta_\delta = \sum_{x \in \cQ^c_\delta} P_X(x)  \log \left( \frac{1}{Q_X(x)}\right ). 
\end{equation}
Next, note that  since $\cQ^c_\delta$ is a finite set, then $\Delta_\delta$ is a finite sum for all $\delta>0$.  Combining~\eqref{Eq:use_monote_In_finite} and~\eqref{eq:finitness_of_cross}, we conclude that  $H_\gamma(P_X, Q_X) < \infty$ if and only if $H(P_X, Q_X) < \infty$, which concludes the proof of Proposition~\ref{prop:focal-finite-iff-cross-finite}.

\subsection{Proof of Proposition~\ref{prop:convexity-Hgamma}}
\label{app:convexity-Hgamma}
Consider a sequence $Q_n$ that converges to a distribution $Q_{\infty}$. Then, the lower semicontinuity property follows from the application of Fatou's lemma since
\begin{align}
	\liminf_{n \to \infty}   H_{\gamma}(P_X,Q_n )  & \ge  \bbE \left [  \liminf_{n \to \infty} \sfL_\gamma(Q_n(X)   ) \right ] 
	\\& = H_{\gamma}(P_X,Q_{\infty} ).
\end{align}
To show the convexity property, recall from Proposition~\ref{prop:AnalyticalPropFocalLoss} that $p \;\mapsto\; \sfL_\gamma(p)$ is strictly convex on $(0,1)$ for $\gamma \geq 0$.
Consequently,
\begin{align}
	H_\gamma(P_X,Q) =\sum_{x\in\cX}P_X(x)\,\sfL_{\gamma}(Q(x)) 
\end{align}
is a non-negative weighted sum of strictly convex functions and hence, it is convex in $Q$.  Moreover, $H_\gamma(P_X,Q)$ is strictly convex in $Q$ on $\cS$ due to the fact that it is a positive weighted sum of strictly convex functions. This concludes the proof of Proposition~\ref{prop:convexity-Hgamma}.

\subsection{Proof of Corollary~\ref{cor:ConseqTheorem1}}  
\label{app:ConseqTheorem1}

Property 1) follows directly from the expression of $P_\gamma^\star$ in~\eqref{eq:optimizer}.

For Property 2), note that the fact that $P_X\left(x_1\right) \geq P_X\left(x_2\right)$ implies that
\begin{equation}
	-\frac{\alpha_\gamma^\star}{P_X\left(x_1\right)}  \geq -\frac{\alpha_\gamma^\star}{P_X\left(x_2\right)}.
\end{equation}
By applying $\left( \sfL^{\prime}_\gamma \right )^{-1}$ to both sides of the above inequality and recalling that $t \;\mapsto\; \left( \sfL^{\prime}_\gamma \right )^{-1}(t)$ for $t \in (-\infty,0)$ is strictly increasing (see Proposition~\ref{prop:AnalyticalPropFocalLossInverse}), we arrive at 
\begin{equation}
	P^\star_\gamma(x_1) \ge P^\star_\gamma(x_2).
\end{equation}

For Property 3), note  that to have $P_\gamma^\star = P_X$, from~\eqref{eq:optimizer} we would, in fact, need that
\begin{equation}
	(\mathsf L_\gamma')^{-1} \left (-\tfrac{\alpha_\gamma^\star}{P_X(x)}\right ) = P_X(x) \ \text{for all} \ x \in \cS,
\end{equation}
or equivalently
\begin{equation}
	\alpha_\gamma^\star = -P_X(x) \,  \mathsf L_\gamma' \left( P_X(x)\right ), \ \text{for all} \ x \in \cS.
\end{equation}
Since the function $p \;\mapsto\;  \sfL^\prime_\gamma(p)$ is strictly increasing on $p \in (0,1)$ and continuous (see Proposition~\ref{prop:AnalyticalPropFocalLoss}), then the above equation is satisfied if and only if $\gamma=0$ or $P_X$ is a uniform distribution.

This concludes the proof of Corollary~\ref{cor:ConseqTheorem1}.

\subsection{Proof of Proposition~\ref{prop:BoundsAlpha}}
\label{app:BoundsAlpha}
We start by showing~\eqref{eq:alpha_bound_max_phi}. Toward this end, let $\phi_{\gamma,M}=\max_{t \in (0,1)} \phi_\gamma(t)$ and let
\begin{equation}
	q_i = \left( \sfL^{\prime}_\gamma \right )^{-1}\left(-\frac{\phi_{\gamma,M}}{p_i}\right),
\end{equation}
from which we obtain
\begin{equation}
	p_i = - \frac{\phi_{\gamma,M}}{\sfL^{\prime}_\gamma(q_i)} = \frac{\phi_{\gamma,M}}{\phi_\gamma(q_i)}  \, q_i \geq q_i,
\end{equation}
for all $i$'s. Thus, we have that
\begin{equation}
	F(\phi_{\gamma,M}) = \sum_{i}(\sfL^\prime_\gamma)^{-1} \left(-\tfrac{\phi_{\gamma,M}}{p_i}\right) = \sum_{i} q_i \leq \sum_{i} p_i =1. 
\end{equation}
Since $\alpha \mapsto F(\alpha)$ is strictly decreasing and $F(\alpha^\star_\gamma)=1$, it follows that $\alpha^\star_\gamma \le \phi_{\gamma,M}$. Finally, the bound $\phi_{\gamma,M} \le 1+\gamma$ is given in Proposition~\ref{prop:properties_of_phi_gamma}.

We now show~\eqref{eq:BoundsOptAlpha}. From the definition of $\alpha_\gamma^\star$, we have that
\begin{equation}
	\frac{1}{N} =  \frac{1}{N} \sum_{x \in \cS}  (\mathsf L_\gamma')^{-1} \left (-\tfrac{\alpha_\gamma^\star}{P_X(x)}\right ),  \label{eq:equality_norm_N}
\end{equation}
and an additional note that 
\begin{align}
	(\mathsf L_\gamma')^{-1} \left (-\tfrac{\alpha_\gamma^\star}{p_{\min}}\right )  & \le  \frac{1}{N} \sum_{x \in \cS}  (\mathsf L_\gamma')^{-1} \left (-\tfrac{\alpha_\gamma^\star}{P_X(x)}\right )  \notag\\
	&\le       (\mathsf L_\gamma')^{-1} \left (-\tfrac{\alpha_\gamma^\star}{p_{\max}}\right ), \label{eq:bounds_to_get_alpha}
\end{align}
where the inequalities follow from the fact that $t \mapsto (\sfL_\gamma')^{-1} \left ( t\right )$  for $t \in (-\infty,0)$ is strictly increasing,  as shown in Proposition~\ref{prop:AnalyticalPropFocalLossInverse}.  Now, using the fact that $p \;\mapsto\;  \sfL^\prime_\gamma(p)$ is  strictly increasing on $p \in (0,1)$,
as shown in Proposition~\ref{prop:AnalyticalPropFocalLoss}, and applying it to~\eqref{eq:equality_norm_N} and~\eqref{eq:bounds_to_get_alpha} we arrive at 
\begin{equation}
	-\tfrac{\alpha_\gamma^\star}{p_{\min}}   \le     \mathsf L_\gamma' \left (\frac{1}{N}\right )  \le      - \tfrac{\alpha_\gamma^\star}{p_{\max}},
\end{equation}
which leads to~\eqref{eq:BoundsOptAlpha}. 

{
We now prove~\eqref{eq:BoundsOptAlpha2}.
	Let $p_{(i)}$ denote the $i$th largest entry of $P_X$.
	Under the assumption $\gamma > \kappa(p_{\min})$,
	we know from Proposition~\ref{prop:properties_of_phi_gamma} that $\phi(\cdot)$ is strictly decreasing, i.e.,
	\begin{equation}
		\phi(p_{(1)})< \cdots < \phi(p_{(N)}),    
	\end{equation}
	which,  for all $i \in \{1,\ldots, N\}$, implies that
	\begin{equation}
		-\frac{\phi(p_{(1)})}{p_{(i)}}
		>
		-\frac{\phi(p_{(i)})}{p_{(i)}}
		=  \sfL_\gamma^\prime(p_{(i)}). 
	\end{equation}
	Now, by applying $\left( \sfL^{\prime}_\gamma \right )^{-1}$ on both sides of the above inequality and recalling that $t \mapsto (\sfL_\gamma')^{-1} \left ( t\right )$  for $t \in (-\infty,0)$ is strictly increasing
	(see Proposition~\ref{prop:AnalyticalPropFocalLossInverse}), we get
	\begin{equation}
		p_{(i)} <    (\sfL_\gamma^\prime)^{-1} \left(-\frac{\phi(p_{(1)})}{p_{(i)}}\right).   
	\end{equation}
	Similarly,   for all $i \in \{1,\ldots, N\}$, we conclude that 
	\begin{equation}
		p_{(i)} >    (\sfL_\gamma^\prime)^{-1} \left(-\frac{\phi(p_{(N)})}{p_{(i)}}\right). 
	\end{equation}
	Now, evaluating $F(\cdot)$ in~\eqref{eq:normalization} in $\phi(p_{(1)})$, we obtain
	\begin{equation}
		F\left(\phi(p_{(1)})\right)
		=\sum_{i=1}^N(\sfL_\gamma^\prime)^{-1} \left(-\frac{\phi(p_{(1)})}{p_{(i)}}\right)
		> \sum_{i=1}^N p_{(i)}=1,    
	\end{equation}  
	and  evaluating $F(\cdot)$ in~\eqref{eq:normalization} in $\phi(p_{(N)})$, we obtain
	\begin{equation}
		F\left(\phi(p_{(N)})\right)
		=\sum_{i=1}^N(\sfL_\gamma^\prime)^{-1} \left(-\frac{\phi(p_{(N)})}{p_{(i)}}\right)
		< \sum_{i=1}^N p_{(i)}=1.   
	\end{equation}
	Now recall that $F(\alpha)$ is continuous and strictly decreasing in~$\alpha$. Thus, by monotonicity, we conclude the proof of~\eqref{eq:BoundsOptAlpha2}.
	The proof of~\eqref{eq:BoundsOptAlpha3} follows by the same steps as those used for the proof of~\eqref{eq:BoundsOptAlpha2} with the only difference that, under the assumption $\gamma < \kappa(p_{\max})$,
	we know from Proposition~\ref{prop:properties_of_phi_gamma} that $\phi(\cdot)$ is now strictly increasing, i.e.,
	\begin{equation}
		\phi(p_{(1)})>\cdots > \phi(p_{(N)}),     
	\end{equation}
    This concludes the proof of Proposition~\ref{prop:BoundsAlpha}.
}

\subsection{Proof of Proposition~\ref{prop:binary-focal-minimizer}}
\label{app:binary-focal-minimizer}
		From Proposition~\ref{prop:convexity-Hgamma}, we have that the mapping $Q_X \mapsto H_\gamma(P_X,Q_X)$ is strictly convex {on $\{0,1\}$} provided that the distribution $P_X$ is non-degenerate. Thus, by letting $q = Q_X(1)$ (or equivalently, $Q_X(0) = 1-q$),
        the derivative
		\begin{align}
			D(q) & = \frac{\rmd }{\rmd q} H_\gamma(P_X,\{q,1-q\}) 
            \\& = p \, (1-q)^{\gamma-1} \left( \gamma \log (q) - \frac{1-q}{q} \right ) +(1-p) \, q^{\gamma-1} \left( - \gamma \log(1-q) + \frac{q}{1-q} \right )\label{eq:DerivativeDqBinary}
		\end{align}
		is strictly increasing with a unique root $q^\star_\gamma$.  Note also that the fact that $p\in \left (0,\tfrac12 \right ]$ from~\eqref{eq:powerLaw prob} implies that $q_{\gamma}\leq\tfrac12$ for all $\gamma>0$.

We prove lower bound in~\eqref{eq:BoundsonOptimlDistrBinaryCase}. From~\eqref{eq:powerLaw prob}, we have that
\begin{equation}
\frac{q_\gamma}{1-q_\gamma} = \frac{p^{\tfrac{1}{\gamma}}}{\,p^{\tfrac{1}{\gamma}} + (1-p)^{\tfrac{1}{\gamma}}\,} \frac{\,p^{\tfrac{1}{\gamma}} + (1-p)^{\tfrac{1}{\gamma}}\,}{(1-p)^{\tfrac{1}{\gamma}}} = \frac{p^{\tfrac{1}{\gamma}}}{(1-p)^{\tfrac{1}{\gamma}}},
\end{equation}
or equivalently
\begin{equation}
\frac{1-p}{p}=\left(\frac{1-q_\gamma}{q_\gamma}\right)^{\gamma}.	
\end{equation}
			Using the above inside~\eqref{eq:DerivativeDqBinary}, we obtain
			\begin{equation}
				D(q_\gamma)
				=p\,\frac{(1-q_\gamma)^{\gamma-1}}{q_\gamma} \left[\gamma\,\Delta(q_\gamma)+2q_\gamma-1\right],
			\end{equation}
			where $\Delta(q)=(1-q)\log \left( \frac1{1-q} \right )-q\log \left( \frac1{q} \right ) \leq 0 $ for $q \leq \tfrac12$.  
			Thus, since $q_\gamma \leq \frac{1}{2}$, we have that 
			\begin{equation}
				D(q_\gamma)
				\leq p\,\frac{(1-q_\gamma)^{\gamma-1}}{q_\gamma}(2q_\gamma-1)
				\leq 0,
			\end{equation}
			which implies that $q_\gamma \leq q^\star_\gamma$. 

We now prove the upper bound in~\eqref{eq:BoundsonOptimlDistrBinaryCase}. {From~\eqref{eq:powerLaw prob}, we have that}
			\begin{equation}
				\frac{1-p}{p}
				= \left(\frac{1-q_{\gamma+1}}{q_{\gamma+1}}\right)^{\gamma+1}.
			\end{equation}
			Using the above inside~\eqref{eq:DerivativeDqBinary}, we obtain
			\begin{equation}
				D(q_{\gamma+1})
				    =
				\gamma \,  p\,\frac{(1-q_{\gamma+1})^{\gamma-1}}{q_{\gamma+1}^2}\,\left[(1-q_{\gamma+1})^2\log \left( \tfrac1{1-q_{\gamma+1}} \right )
				- q_{\gamma+1}^2\log \left( \tfrac1{q_{\gamma+1}} \right )\right].
			\end{equation}
			For $q \leq \tfrac12$, we have that $(1-q)^2\log \left( \frac1{1-q} \right ) \geq q^2\log \left( \frac1{q} \right )$ and hence, $D(q_{\gamma+1})\geq 0$, leading to $q^\star_\gamma \leq q_{\gamma+1}$.
			
{Combining the fact that $D(q_\gamma) \leq 0 \leq D(q_{\gamma+1})$ with the strict monotonicity of $D$} yields
			\begin{equation}
				{q_\gamma \leq q^\star_\gamma \leq q_{\gamma+1}.}
			\end{equation}
		We now prove~\eqref{eq:DiffBoundsBinary}.
        {By the mean value theorem, there exists $\xi\in(\gamma,\gamma+1)$ such that
        \begin{equation}
        \label{eq:qPrimeXi}
        q'(\xi) = q_{\gamma+1} - q_\gamma = \widetilde{Q}_{\gamma+1}(1) - \widetilde{Q}_{\gamma}(1),
        \end{equation}
        where $q'(\xi) = \left. \frac{\rmd}{\rmd \gamma} q_\gamma \right |_{\gamma=\xi}$, which can be computed from~\eqref{eq:powerLaw prob}, that is,
        \begin{align}
				q'(\xi) & = - \log \left( \frac{p}{1-p}\right ) \frac{q_{\xi} (1-q_\xi)}{\xi^2}
                \\& \leq \left| \log \left( \frac{p}{1-p}\right ) \right | \frac{1}{4 \xi^2} \label{eq:IntermDistanceQOpt}
                \\& \leq \left| \log \left( \frac{p}{1-p}\right ) \right | \frac{1}{4 \gamma^2}, \label{eq:IntermDistanceQOpt2}
        \end{align}
        where~\eqref{eq:IntermDistanceQOpt} follows since $q_{\gamma} (1-q_\gamma) \leq \frac{1}{4}$ for $q_\gamma \leq \frac{1}{2}$ and~\eqref{eq:IntermDistanceQOpt2} follows since $\xi\in(\gamma,\gamma+1)$. Combining~\eqref{eq:qPrimeXi} with~\eqref{eq:IntermDistanceQOpt2}, we arrive at
        \begin{equation}
        \widetilde{Q}_{\gamma+1}(1) - \widetilde{Q}_{\gamma}(1) \leq \left| \log \left( \frac{p}{1-p}\right ) \right | \frac{1}{4 \gamma^2}.
        \end{equation}
        }
		This concludes the proof of Proposition~\ref{prop:binary-focal-minimizer}.

\section{Proofs for Section~\ref{sec:prop_optimizer}}

\subsection{Proof of Proposition~\ref{eq:PGmmaStar0Inft}}
\label{app:PGmmaStar0Inft}
Let $f_0$ be such that $(\mathsf L_\gamma')^{-1} \left (f_0\right ) = \frac{1}{|\cS|}$. 
Then, using the Taylor remainder theorem, we have that 
\begin{align}
	P_\gamma^\star(x) &=  (\mathsf L_\gamma')^{-1} \left ( -\frac{\alpha_\gamma^\star}{P_X(x)}\right ) 
	\\&= \frac{1}{|\cS|} - \left(  \frac{\alpha_\gamma^\star}{P_X(x)} + f_0  \right)  \frac{\rmd}{\rmd t} (\mathsf L_\gamma')^{-1}(t) |_{t=f_0}  + R_{x,\infty}
	\\&= \frac{1}{|\cS|} -  \frac{\left(  \frac{\alpha_\gamma^\star}{P_X(x)} + f_0  \right)}{\mathsf L_\gamma'' \left( (\mathsf L_\gamma')^{-1}(f_0) \right)}     + R_{x,\infty}
	\\&= \frac{1}{|\cS|} -  \frac{\left(  \frac{\alpha_\gamma^\star}{P_X(x)} + \sfL_\gamma' \left( \frac{1}{|\cS|} \right)  \right)}{\mathsf L_\gamma'' \left( \frac{1}{|\cS|} \right)}     + R_{x,\infty}, \label{eq:P_x_remenider_Expression}
\end{align}
where  the last step follows from the fact that $f_0 = \sfL_\gamma' \left( \frac{1}{|\cS|} \right) $ and the remainder term is given by 
\begin{align}
	R_{x,\infty} &=  \frac{\left (  \frac{\alpha_\gamma^\star}{P_X(x)} + \sfL_\gamma' \left( \frac{1}{|\cS|} \right) \right )^2}{2}  \left. \frac{\rmd^2}{\rmd t^2} (\mathsf L_\gamma')^{-1} (t) \right |_{t=c} 
	\\&=  - \frac{ \left (  \frac{\alpha_\gamma^\star}{P_X(x)} + \sfL_\gamma' \left( \frac{1}{|\cS|} \right) \right )^2}{2}  \left. \frac{ \mathsf L_\gamma''' (u) }{ \left(  \mathsf L_\gamma'' (u)\right)^3} \right  |_{u = (\mathsf L_\gamma')^{-1} (c) },
\end{align}
for some $ c \in \left (\sfL_\gamma' \left( \frac{1}{|\cS|} \right),  -\frac{\alpha_\gamma^\star}{P_X(x)} \right )$. By combining the bounds on $\alpha^\star_\gamma$ in~\eqref{eq:BoundsOptAlpha},  it is not difficult to see that 
\begin{equation}
	R_{x,\infty} = O \left ( \frac{1}{\gamma}\right ).\label{eq:order_bounds_remeinder}
\end{equation}
Summing both sides of~\eqref{eq:P_x_remenider_Expression} over $x \in \cS$, we arrive at 
\begin{equation}
	\alpha_\gamma^\star \sum_{x \in \cS} \frac{1}{P_X(x)}   + |\cS| \, \sfL_\gamma' \left( \frac{1}{|\cS|} \right)   -  \sfL_\gamma'' \left( \frac{1}{|\cS|} \right)  \sum_{x \in \cS}  R_{x,\infty} = 0,
\end{equation}
which implies that 
\begin{align}
	\alpha_\gamma^\star &=    -  \sfL'_\gamma \left( \frac{1}{|\cS|} \right) \mathsf{HM}(P_X) +  \frac{\mathsf{HM}(P_X) }{|\cS|} \sfL_\gamma'' \left( \frac{1}{|\cS|} \right)  \sum_{x \in \cS}  R_{x,\infty} \\
	&= -  \sfL'_\gamma \left( \frac{1}{|\cS|} \right) \mathsf{HM}(P_X) +  O \left( \gamma \, \left (1-\frac{1}{|\cS|}\right )^{\gamma} \right), \label{eq:alpha_approx_inside_proof}
\end{align}
where the last step follows by using~\eqref{eq:order_bounds_remeinder}.

Now, inserting~\eqref{eq:alpha_approx_inside_proof} into~\eqref{eq:P_x_remenider_Expression}, we arrive at the conclusion that as $\gamma \to \infty$, we have that
\begin{align}
	P_\gamma^\star(x) & = \frac{1}{|\cS|} + \frac{\sfL'_\gamma \left( \frac{1}{|\cS|} \right)}{ \sfL''_\gamma \left( \frac{1}{|\cS|} \right)} \left(  \frac{\mathsf{HM}(P_X)}{P_X(x) } -1 \right) + O \left( \frac{1}{\gamma}\right ) \\  
	&= \frac{1}{|\cS|} + O \left( \frac{1}{\gamma}\right ), 
\end{align}
where in the last equality we have used that  $\frac{\sfL'_\gamma \left( \frac{1}{|\cS|} \right)}{ \sfL''_\gamma \left( \frac{1}{|\cS|} \right)} = O \left( \frac{1}{\gamma} \right)$. 
This concludes the proof of Proposition~\ref{eq:PGmmaStar0Inft}.

\subsection{Proof of Theorem~\ref{thm:Signdi}}
\label{app:Signdi}
We consider two cases:

\noindent {\bf{Case~1:}} If $p_{(i)} \in (0,p_{\gamma,a}]$ or $p_{(i)} \in [ p_{\gamma,b},1)$, then $\alpha_\gamma^\star \geq \phi_\gamma(p_{(i)})$. Since $t \mapsto (\sfL_\gamma')^{-1} \left ( t\right )$  for $t \in (-\infty,0)$ is strictly increasing 
(see Proposition~\ref{prop:AnalyticalPropFocalLossInverse}), we have that
\begin{equation}
	p_{(i)} \geq (\mathsf L_\gamma')^{-1} \left (-\tfrac{\alpha_\gamma^\star}{p_{(i)}}\right ).
\end{equation}
From~\eqref{eq:optimizer}, the right-hand side of the above inequality is precisely $p^\star_{(i)}$ and hence, $d_i \geq 0.$

\noindent {\bf{Case~2:}} If $p_{(i)} \in (p_{\gamma,a},p_{\gamma,b})$, then $\alpha_\gamma^\star < \phi_\gamma(p_{(i)})$. Since $t \mapsto (\sfL_\gamma')^{-1} \left ( t\right )$  for $t \in (-\infty,0)$ is strictly increasing 
(see Proposition~\ref{prop:AnalyticalPropFocalLossInverse}), this implies that
\begin{equation}
	p_{(i)} < (\mathsf L_\gamma')^{-1} \left (-\tfrac{\alpha_\gamma^\star}{p_{(i)}}\right ).
\end{equation}
From~\eqref{eq:optimizer}, the right-hand side of the above inequality is precisely $p^\star_{(i)}$ and hence, $d_i < 0$. This concludes the proof of Theorem~\ref{thm:Signdi}.
\subsection{Proof of Proposition~\ref{prop:binary_regime}}
\label{app:binary_regime}
We focus on the case $\alpha_\gamma^\star \geq 1$ (if $\alpha_\gamma^\star < 1$, then from Lemma~\ref{lemma:Roots}, the result trivially holds). 
If $\alpha_\gamma^\star \geq 1$, then we need $\alpha_\gamma^\star = \phi_\gamma(p_{\gamma,b})\geq 1$. 
However, we now prove that for all $p \in [1/2,1)$, we have that $\phi_\gamma(p)<1$, which implies that $p_{\gamma,b} < 1/2$ (to ensure that $\phi_\gamma(p_{\gamma,b})\geq 1$) and hence, $p_{(1)} > p_{\gamma,b}$, i.e., $d_1\geq 0$. The fact that $\phi_\gamma(p)<1$, for all $p \in [1/2,1)$ is because:
(i) from Proposition~\ref{prop:properties_of_phi_gamma}, $\phi_\gamma(p)$ is unimodal and it is such that $\phi_\gamma(0^+) = 1$ and $\phi_\gamma(1)=0$; 
(ii) $\phi_\gamma(1/2) < 1$;
(iii) $p \mapsto \phi_\gamma(p)$ is strictly decreasing on $p \in [1/2,1)$ (from Proposition~\ref{prop:properties_of_phi_gamma}). Thus, $p_{\max} \geq p_{\gamma,b}$, i.e., $d_1\geq 0$. The fact that $d_2<0$ follows from Corollary~\ref{cor:NumberOfSignChanges}.
This concludes the proof of Proposition~\ref{prop:binary_regime}.
\subsection{Proof of Proposition~\ref{prop:ternary_regime}}
\label{app:ternary_regime} 
From Proposition~\ref{prop:properties_of_phi_gamma}, since $p_{\gamma,a}<1/3$, we have that
\begin{equation}
	\phi_\gamma(p_{\gamma,a}) > \phi_\gamma \left(\frac{1-p_{\gamma,a}}{2}\right).
\end{equation}
Moreover, we also need
\begin{equation}
	\alpha_\gamma^\star = \phi_\gamma(p_{\gamma,a}) = \phi_\gamma(p_{\gamma,b}).
\end{equation}
The above two equations, together with the unimodality of $\phi_\gamma$ (see Proposition~\ref{prop:properties_of_phi_gamma}), imply that
\begin{equation}
	p_{\gamma,b} <  \frac{1-p_{\gamma,a}}{2}.
\end{equation}
We are now ready to prove that $d_1 \geq 0$. The proof is by contradiction. Assume that $d_1 <0$. Then, from Theorem~\ref{thm:Signdi}, we have that
\begin{equation}
	p_{(1)} < p_{\gamma,b} <  \frac{1-p_{\gamma,a}}{2},
\end{equation}
which leads to
\begin{align}
	p_{(1)}+p_{(2)}+p_{(3)} &\leq  2 \, p_{(1)} + p_{(3)}
	\\& < 2 \frac{1-p_{\gamma,a}}{2} + p_{\gamma,a} 
	\\& =1,
\end{align}
where the second inequality follows since we need at least one sign change in the sequence $\{d_i\}, i \in \{1,2,3\}$ (see Corollary~\ref{cor:NumberOfSignChanges}).
Thus, we have found a contradiction since we need $p_{(1)}+p_{(2)}+p_{(3)}=1$. This concludes the proof of Proposition~\ref{prop:ternary_regime}.

\subsection{Example: $p_{\min} < p_{\gamma,a}$}
\label{app:N4WithOversuppression}
Let $\gamma=0.2$ and $\cS = \{x_1,x_2,x_3,x_4\}$ with
\begin{equation}
	P_X(x) = \frac{1}{51}\left \{
	\begin{array}{ll}
		1 & x = x_1,
		\\
		20 & x = x_2, x_4,
		\\
		10 & x = x_3,
		\\
		0 & \text{otherwise.}
	\end{array}
	\right.
\end{equation}
With this, from~\eqref{eq:optimizer}, we obtain
\begin{equation}
	P_\gamma^\star = \left \{
	\begin{array}{ll}
		3/154 & x = x_1,
		\\
		195/499 & x = x_2, x_4,
		\\
		191/960 & x = x_3,
		\\
		0 & \text{otherwise.}
	\end{array}
	\right.
\end{equation}
From the above, we have that
\begin{align}
	&d_1 = d_2 = \frac{20}{51} - \frac{195}{499} = \frac{35}{25449} > 0,
	\\
	& d_3 = \frac{10}{51} - \frac{191}{960} = -\frac{47}{16320} < 0,
	\\
	& d_4 = \frac{1}{51}  - \frac{3}{154} = \frac{1}{7854} >0,  
\end{align}
that is, $p_{\min} = \frac{1}{51} < p_{\gamma,a}$ (see also Figure~\ref{fig:sign_changes}).

\subsection{Proof of Proposition~\ref{thm:SuffCondGammaMajor}}
\label{app:SuffCondGammaMajor}
We start by noting that the bounds in~\eqref{eq:BoundsOptAlpha2}, together with the fact that $\phi(p_{(i)})$ is strictly increasing in $i$ (see Proposition~\ref{prop:properties_of_phi_gamma}), implies that there is a unique $k_M$ such that
\begin{equation}
	\label{eq:kM}
	k_M
	=\max\{i:\phi(p_{(i)})\le \alpha^\star_\gamma\},\ k_M \in \mathbb{N}. 
\end{equation}
Now, we consider two cases:
\begin{enumerate}
	\item {\bf{Case~1:}} $i \leq k_M$. For this case, we have that $\phi(p_{(i)})\le \alpha^\star_\gamma$, which implies
	\begin{equation}
		-\frac{\alpha^\star_{\gamma}}{p_{(i)}}
		\le 
		-\frac{\phi(p_{(i)})}{p_{(i)}}
		=\sfL^\prime_\gamma(p_{(i)}),
	\end{equation}
	and hence, by applying $\left( \sfL^{\prime}_\gamma \right )^{-1}$ on both sides of the above inequality and recalling that $\left( \sfL^{\prime}_\gamma \right )^{-1}$ is strictly increasing on its domain $(-\infty,0)$ (see Proposition~\ref{prop:AnalyticalPropFocalLossInverse}), we arrive at
	\begin{equation}
		p^\star_{(i)} \leq p_{(i)},
	\end{equation}
	that is, $d_i \geq 0$.
	\item {\bf{Case~2:}} $i > k_M$. For this case, we have that $\phi(p_{(i)}) > \alpha^\star_\gamma$, which implies
	\begin{equation}
		-\frac{\alpha^\star_{\gamma}}{p_{(i)}}
		> 
		-\frac{\phi(p_{(i)})}{p_{(i)}}
		=\sfL^\prime_\gamma(p_{(i)}),
	\end{equation}
	and hence, by using similar steps as in Case~1, we arrive at
	\begin{equation}
		p^\star_{(i)} > p_{(i)},
	\end{equation}
	that is, $d_i < 0$.
\end{enumerate}
In summary, we have proved that $d_i \geq 0$ for all $i \leq k_M$ and $d_i < 0$ for all $i >k_M$, that is, $p_{\min} > p_{\gamma_a}$ (see also Figure~\ref{fig:sign_changes}). This concludes the proof of Proposition~\ref{thm:SuffCondGammaMajor}.
\begin{rem}
	By following similar steps as above, it is not difficult to prove that, whenever $\gamma < \kappa(p_{\max})$, then $p_{\max} < p_{\gamma,b}$, that is, for these values of $\gamma$, the over-suppression regime {\em does} exist.
\end{rem}

\subsection{Proof of Proposition~\ref{prop:SuffCond1Positive}}
\label{app:SuffCond1Positive}
We start by noting that, if $\alpha_\gamma^\star < 1$, then the result trivially holds since from Corollary~\ref{cor:NumberOfSignChanges}, the sequence $\{ d_i \}$ has at least one sign change. Therefore, we next assume $\alpha_\gamma^\star \geq  1$.

Using the definition of $\phi_\gamma(p)$ and the fact that $\phi_\gamma(p_{\gamma,b}) = \alpha_\gamma^\star$, we have that
\begin{align}
	\alpha_\gamma^\star &= p_{\gamma,b} \; (1-p_{\gamma,b})^{\gamma-1} \left( \gamma \log \frac{1}{p_{\gamma,b}} +\frac{1-p_{\gamma,b}}{p_{\gamma,b}} \right) \label{eq:InitialAlphaNoBound}
	\\& \leq p_{\gamma,b} \; (1-p_{\gamma,b})^{\gamma-1} \left(  \frac{\gamma}{p_{\gamma,b}} - \gamma +\frac{1-p_{\gamma,b}}{p_{\gamma,b}} \right)
	\\& = (1-p_{\gamma,b})^{\gamma} (\gamma+1),
\end{align}
where the inequality follows since $\log(x) \leq x-1$ for $x>0$. Now, since $\alpha_\gamma^\star \geq 1$, we can further bound the above as follows,
\begin{equation}
	1 \leq (1-p_{\gamma,b})^{\gamma} (\gamma+1),
\end{equation}
which implies
\begin{equation}
	p_{\gamma,b} \leq 1 - \left( \frac{1}{\gamma+1}\right )^{\frac{1}{\gamma}}.
\end{equation}
To show that $p_{\max} > p_{\gamma,b}$, it suffices to show that $p_{\gamma,b} < \frac{1}{|\cS|}$ (since $p_{\max} \geq \frac{1}{|\cS|}$). A sufficient condition for this can be obtained by setting the above upper bound to be smaller than $\frac{1}{|\cS|}$, that is,
\begin{equation}
	1 - \left( \frac{1}{\gamma+1}\right )^{\frac{1}{\gamma}} < \frac{1}{|\cS|},
\end{equation}
or equivalently
\begin{equation}
	\label{eq:GeneralN}
	\left( \frac{1}{\gamma+1}\right )^{\frac{1}{\gamma}} > \frac{|\cS|-1}{|\cS|}.
\end{equation}
Since $\gamma \mapsto \left( \frac{1}{\gamma+1}\right )^{\frac{1}{\gamma}}$ is strictly increasing, it suffices to find $\gamma_0 > 0$ such that
\begin{equation}
	\label{eq:gamma0}
	\left( \frac{1}{\gamma_0+1}\right )^{\frac{1}{\gamma_0}} = \frac{|\cS|-1}{|\cS|}
\end{equation}
to ensure that~\eqref{eq:GeneralN} is satisfied for all $\gamma > \gamma_0$. We now focus on solving~\eqref{eq:gamma0}, which is equivalent to
\begin{equation}
	\label{eq:IntermExcludeW0}
	\left( \frac{|\cS|}{|\cS|-1} \right )^{\gamma_0} = \gamma_0 + 1,
\end{equation}
or
\begin{equation}
	\exp \left( \gamma_0 \log \left( \frac{|\cS|}{|\cS|-1} \right ) \right ) = \gamma_0 + 1.
\end{equation}
Letting $x=\gamma_0 +1$ (where $x >1$) and $n = \log \left( \frac{|\cS|}{|\cS|-1} \right )$, we arrive at
\begin{equation}
	\exp \left( (x-1) \, n \right ) = x,
\end{equation}
or equivalently
\begin{equation}
	\exp (-n) = x \exp(-nx),
\end{equation}
or equivalently
\begin{equation}
	-n\exp (-n) = -n \; x \exp(-nx).
\end{equation}
By letting $y = -n \;x$ and $t = -n\exp (-n)$, we obtain
\begin{equation}
	\label{eq:FinalLambert}
	y \exp(y) = t.
\end{equation}
Note that the above equation can be solved for $y$ only if $t\geq -\frac{1}{\rme}$, which in our case holds since
\begin{align}
	t & = -n\exp (-n)
	\\& = - \log \left( \frac{|\cS|}{|\cS|-1} \right ) \exp \left ( -\log \left( \frac{|\cS|}{|\cS|-1} \right ) \right )
	\\& = \frac{|\cS|-1}{|\cS|} \log \left( \frac{|\cS|-1}{|\cS|} \right ) \label{eq:EqTInc}
	\\& \geq \frac{1}{2} \log \left ( \frac{1}{2}\right ) \label{eq:FirstIneT}
	\\& > -\frac{1}{\rme},
\end{align}
where the inequality in~\eqref{eq:FirstIneT} follows since the function in~\eqref{eq:EqTInc} is increasing in $|\cS|$. 
Thus,~\eqref{eq:FinalLambert} can be solved and, since $t<0$, it has two distinct real solutions, namely
\begin{equation}
	\label{eq:SolPrincBranch}
	y  = \rmW_0 (t),
\end{equation}
and 
\begin{equation}
	y  = \rmW_{-1} (t).
\end{equation}
We now show that the solution $y  = \rmW_0 (t)$ leads to the trivial solution $\gamma_0=0$. We start by noting that~\eqref{eq:SolPrincBranch} is equivalent to
\begin{equation}
	\label{eq:Gamma0PrincipalBranch}
	\gamma_0 = -1 - \frac{1}{n} \rmW_0 \left( -n \exp(-n)\right ).
\end{equation}
Now, we know that $\rmW_0 (t) \in (-1,0)$ when $-\frac{1}{\rme}<t<0$ and hence, $-\frac{1}{n} \rmW_0(t) \in \left (0, \frac{1}{n} \right )$. Thus, from~\eqref{eq:Gamma0PrincipalBranch}, we obtain
\begin{equation}
	\gamma_0 = -1 - \frac{1}{n} \rmW_0 \left( -n \exp(-n)\right ) \in \left (-1,-1 + \frac{1}{n} \right ).
\end{equation}
We note that $0 \in \left (-1,-1 + \frac{1}{n} \right )$ and indeed $\gamma_0 =0$ satisfies~\eqref{eq:IntermExcludeW0}. However, we already know that $\gamma>0$.
Therefore, we only consider $y  = \rmW_{-1} (t)$, which is equivalent to
\begin{equation}
	\label{eq:FinalGamma0}
	\gamma_0 = -1 - \frac{1}{n} \rmW_{-1} \left( -n\exp (-n) \right ),
\end{equation}
where recall that $n = \log \left( \frac{|\cS|}{|\cS|-1} \right )$. Note also that the right-hand side of~\eqref{eq:FinalGamma0} is always positive. This is because $\rmW_{-1} (t) \in (-\infty,-2 \log(2))$ when $\frac{1}{2} \log \left ( \frac{1}{2}\right )\leq t<0$ and hence, $-\frac{1}{n} \rmW_{-1}(t) \in \left (\frac{2 \log(2)}{n}, \infty \right )$ and $\frac{2 \log(2)}{n} \geq 2$ for all $|\cS| \geq 2$.
Therefore, we conclude that $p_{\max} > p_{\gamma,b}$, whenever $\gamma > \gamma_0$.
This concludes the proof of Proposition~\ref{prop:SuffCond1Positive}.

\subsection{Proof of Proposition~\ref{prop:Majorization}}
\label{app:Majorization}
We start by noting that when $p_{\min} > p_{\gamma_a}$, from Theorem~\ref{thm:Signdi} we have that either $p_{(i)} \in (p_{\gamma,a},p_{\gamma,b})$ (in which case $d_i< 0$) or $p_{(i)} \in [p_{\gamma,b},1)$ (in which case $d_i \geq 0$) for all $i \in \{1,2,\ldots,|\cS|\}$. Let
\begin{equation}
	v = \min \{i: d_i < 0\}.
\end{equation}
Then, the sequence $\{ d_i \}$ has one sign change, i.e., $d_i \geq 0$ for all $i \in \{1,2,\ldots,v-1\}$ and $d_i < 0$ for all 
$i \in \{v,v+1, \ldots , |\cS|\}$.
In addition, we also know that
\begin{equation}
	{\sum_{i=1}^{|\cS|} d_i = 0.}
\end{equation}
A vector whose entries change sign 
at most
once (from positive to negative) and
sum to zero has all partial sums non‐negative, 
that is, for all 
{$k \in \{1,\ldots,|\cS|\}$,} 
we have that 
\begin{equation}
	\sum_{i=1}^k d_i \ge0,
\end{equation}
or equivalently
\begin{equation}
	\sum_{i=1}^k p_{(i)}\;\ge\;\sum_{i=1}^k p^\star_{(i)},
\end{equation}
which, together with $\sum_{i=1}^{|\cS|} p_{(i)}=\sum_{i=1}^{|\cS|} p^\star_{(i)}=1$,
is equivalent
to $P_X\succ P^\star_\gamma$. This concludes the proof of Proposition~\ref{prop:Majorization}.

\subsection{Proof of Lemma~\ref{lemma:Prophgamma}}
\label{app:Prophgamma}
To show the first property, note that for $p \in (0,1)$, the function $\gamma \mapsto p^{\left(1-p\right)^\gamma}$ is strictly increasing and hence, $\gamma \mapsto h_\gamma(P_X)$ is strictly increasing.  

To show the lower bound in the second property, note that
\begin{equation}
	h_\gamma(P_X) \ge h_0(P_X) =\log \left( \sum_{x \in \cX} P_X(x) \right )  = 0, 
\end{equation}
where we have used the fact that $\gamma \mapsto h_\gamma(P_X)$ is strictly increasing.
To prove the upper bound in the second property, we consider two cases separately:
\begin{itemize}
	\item {\bf{Case $\gamma > 1$:}} In this case, we have that
	\begin{align}
		&\sum_{x \in \cX} P_X(x)^{\left(1-P_X(x)\right)^\gamma} 
		\\ &= \sum_{x \in \cX} \rme^{ {\left(1-P_X(x)\right)^\gamma \log P_X(x)}  }\\
		&\le   \sum_{x \in \cX} \rme^{\left(1- \gamma P_X(x)\right) \log P_X(x)}  \label{eq:Bernoulli_inq}\\
		&=   \sum_{x \in \cX} P_X(x) \rme^{ - \gamma P_X(x) \log P_X(x)} \\
		& \le \rme^{  \frac{\gamma}{\rme} }   \sum_{x \in \cX} P_X(x) = \rme^{  \frac{\gamma}{\rme} } , \label{eq:xlogx_bound}
	\end{align}
	where~\eqref{eq:Bernoulli_inq} follows from the Bernoulli's inequality  $(1+x)^r \ge 1+rx, \, x \ge -1, r \ge 1$; and~\eqref{eq:xlogx_bound} follows from the bound $x \log x \ge -\frac{1}{\rme}$.
	Thus, by using the above inside~\eqref{eq:hGamma}, we arrive at 
	\begin{equation}
		h_\gamma(P_X) \leq \frac{\gamma}{\rme}.
	\end{equation}
	\item {\bf{Case $0 \leq \gamma \leq 1$:}}
	For this case, we let
	\begin{equation}
		f(\gamma) = \sum_{x \in \cX} P_X(x)^{(1- P_X(x))^\gamma},
	\end{equation}
	and we note that
	\begin{align}
		f'(\gamma) &= \sum_{x  \in\cX} P_X(x)^{(1- P_X(x))^\gamma}  (1- P_X(x))^\gamma  \log 
		(P_X(x)) \log (1-P_X(x)) \\
		&\le  \sum_{x \in\cX}  P_X(x)^{(1- P_X(x))^\gamma}
		\log 
		\left (\frac{1}{P_X(x)} \right ) \log \left (\frac{1}{1-P_X(x)} \right ) \label{eq:first_bound_(1-p)^t}
		\\& \leq \sum_{x \in\cX}  P_X(x)^{(1- P_X(x))^\gamma} \label{eq:second_bound_(1-p)^t}
		\\& \leq \sum_{x \in\cX}  P_X(x)^{(1- P_X(x))} \label{eq:third_bound_(1-p)^t}
		\\& \le \rme^{ \frac{1}{\rme}},
		\label{eq:Bound_on_der}
	\end{align}
	where: in~\eqref{eq:first_bound_(1-p)^t}, we have used the fact that $(1- P_X(x))^\gamma \le 1$;~\eqref{eq:second_bound_(1-p)^t} follows from $\log(x) \le x-1, x>0$;~\eqref{eq:third_bound_(1-p)^t} is due to the fact that $\gamma \mapsto p^{\left(1-p\right)^\gamma}$ is strictly increasing; and in~\eqref{eq:Bound_on_der} we have used the bound in \eqref{eq:xlogx_bound}.
	
	Now, by using the mean value theorem, there exists $c \in (0, \gamma)$ such that
	\begin{align}
		f(\gamma)  &= f'(c) (\gamma -0)  + f(0)\\
		& \le \rme^{ \frac{1}{\rme}} \gamma  + 1,
	\end{align}
	where in the last step we have used the bound in~\eqref{eq:Bound_on_der} and the fact that $f(0)=1$.
	
	Next, by using the above inside~\eqref{eq:hGamma}, we arrive at 
	\begin{align}
		h_\gamma(P_X) &= \log \left( f(\gamma) \right ) \\
		&\le \log ( \rme^{ \frac{1}{\rme}} \gamma +1)\\
		&\le \rme^{ \frac{1}{\rme}} \gamma ,
	\end{align}
	where in the last step we have used the bound $\log(1+x) \le x.$
\end{itemize}

Taking the maximum of both bounds concludes the proof of Lemma~\ref{lemma:Prophgamma}.

\subsection{Proof of Proposition~\ref{prop:RelatEntropyRelation}}
\label{app:RelatEntropyRelation}
		To prove the first expression,  we note that by using the identity in~\eqref{eq:Identi_for_focal_loss}, we have that 
		\begin{align}
			&H_\gamma(P_X,Q_X) \\
			&= \bbE_{X \sim P_X} \left[ \log \left( \frac{1}{Q_X^{(\gamma)}(X)} \right ) \right] - h_\gamma(Q_X)\\
			&= \bbE_{X \sim P_X} \left[ \log \left( \frac{P_X(X)}{Q_X^{(\gamma)}(X)} \right ) \right] -  \bbE \left[ \log P_X(X)  \right] - h_\gamma(Q_X).
		\end{align}
		To show the second expression, we observe that 
		\begin{align}
			&H_\gamma(P_X,Q_X) \\
			&= \sum_{x \in \cX} P_X(x) (1- Q_X(x))^\gamma \log  \left( \frac{1}{Q_X(x)} \right )\\
			&= \rho_\gamma(P_X,Q_X) \left( \sum_{x \in \cX} R_X(x)  \log  \left( \frac{1}{Q_X(x)} \right ) \right) \\
			&=   \rho_\gamma(P_X,Q_X) \left( D_{\sf{KL}}(R_X\| Q_X) + H(R_X) \right). 
		\end{align}
		This concludes the proof of Proposition~\ref{prop:RelatEntropyRelation}.
 
\subsection{Proof of Proposition~\ref{prop:MinFocalCrossRel}}
\label{app:MinFocalCrossRel}

The lower bound follows since, from Proposition~\ref{prop: focal-entropy vs gamma}, we know that the mapping $\gamma \mapsto H_\gamma(P_X,Q_X)$ is non-increasing and hence,
\begin{align}
&\inf_{Q_X} H(P_X,Q_X) -  \inf_{Q_X} H_\gamma(P_X,Q_X) 
\\& \geq \inf_{Q_X} H(P_X,Q_X) -  \inf_{Q_X} H_0(P_X,Q_X) 
\\& = 0.
\end{align}
For the upper bound, we have that
	\begin{align}
    & \inf_{Q_X} H(P_X,Q_X) -  \inf_{Q_X} H_\gamma(P_X,Q_X) \\
		& = \inf_{Q_X} H(P_X,Q_X) -  \inf_{Q_X} \left ( H \left( P_X,Q^{(\gamma)}_X \right )  - h_\gamma(Q_X) \right) \label{eq:FirstEqInf}\\
		& \le  \inf_{Q_X} H(P_X,Q_X) -  \inf_{Q_X} H \!\left (P_X,Q^{(\gamma)}_X \right )  +     \sup_{Q_X}  h_\gamma(Q_X)\label{eq:bound_minmax} \\
		& \le  \inf_{Q_X} H(P_X,Q_X) -  \inf_{Q_X} H(P_X,Q_X)  +     \sup_{Q_X}  h_\gamma(Q_X) \label{eq:bound_without_Q_hat} \\
		&= \sup_{Q_X} h_\gamma(Q_X)\\
		& \le  \rme^{ \frac{1}{\rme}} \gamma, \label{eq:using_lemma}
	\end{align}
	where:~\eqref{eq:FirstEqInf} follows from~\eqref{eq:rel_focal_ent2}; ~\eqref{eq:bound_minmax} is due to the fact that $\inf_x (f(x)-g(x)) \geq \inf_x f(x) - \sup_{x} g(x)$, which follows since $f(x) \geq \inf_x f(x)$ and $g(x) \leq \sup_x g(x)$;~\eqref{eq:bound_without_Q_hat} is due to the fact that the space over which $Q^{(\gamma)}_X$ is defined is a subset of $\cP(\cX)$; and 
	in~\eqref{eq:using_lemma} we have used the bound in Lemma~\ref{lemma:Prophgamma}.
    This concludes the proof of Proposition~\ref{prop:MinFocalCrossRel}.

\section{Experiments on Synthetic Data}
\label{app:SyntheticData}
We consider two classes $C \in \{0,1\}$ 
with $P[C=0]=0.95$ and $P[C=1]=0.05$, and two features $F_1 \in \{0,1,2,3\}$ and $F_2 \in \{0,1,2,3\}$,  
Note that this is an example of an imbalanced dataset. We let
\begin{align}
	P[F_1|C = 0] &= [0.65, 0.20, 0.10, 0.05],
	\\
	P[F_1|C = 1] &= [0.10, 0.25, 0.45, 0.20],
	\\
	P[F_2|C = 0] &= [0.50, 0.30, 0.15, 0.05],
	\\
	P[F_2|C = 1] &= [0.20, 0.50, 0.20, 0.10].
\end{align}
We assume that the two features are 
conditionally independent given the class, i.e., $P[F_1, F_2|C ] =
P [F_1|C ] \,P [F_2|C]$. 
From this, we can therefore obtain the joint probability mass function as follows,
\begin{equation}
	P [F_1, F_2, C ] = P [F_1|C ] \, P [F_2|C ] \, P[C],
\end{equation}
and, by the Bayes' rule, also get $P[C|F_1,F_2]$ (see the bars labeled as ``True $P_{C|F_1,F_2}$'' in Figure~\ref{fig:MLP experiment} for some examples).
We used this $P_{C|F_1,F_2}$  inside~\eqref{eq:optimizer} to obtain $P_{\gamma}^\star$ for $\gamma=1$, i.e., $P_{\gamma}^\star = \arg \min_{Q_X} H_\gamma(P_{C|F_1,F_2},Q_X)$ (see the bars labeled as ``Theory $P_{C|F_1,F_2}$'' in Figure~\ref{fig:MLP experiment} for some examples).

We generated $10,000$ samples by repeatedly sampling from the $16$ possible two-dimension feature vector space (i.e., two features with four possible values each, and two classes). These samples were used to train a feedforward neural network with a single hidden layer of $64$ neurons, each using the ReLU activation function.
We trained this model by using the focal loss in~\eqref{eq:FocalLoss} with $\gamma=1$ and performed the optimization with the Adam optimizer at a learning rate of $10^{-3}$. Training was carried out with a batch size of 64 over 30 epochs. The trained neural network produced predicted probabilities for each of the two classes by applying the softmax function to the output layer (see the bars labeled as ``Model $P_{C|F_1,F_2}$'' in Figure~\ref{fig:MLP experiment} for some examples).
From Figure~\ref{fig:MLP experiment}, we observe that the neural network output probability is close to $P_\gamma^\star$ in~\eqref{eq:optimizer} (i.e., with a maximum difference of $0.027$). This not only supports our theoretical analysis, but it also shows that the neural network has effectively converged to the global minimum, achieving the best possible performance under the focal loss.

\begin{figure*}[t]
	\centering
	\begin{subfigure}[t]{0.25\textwidth}\centering
		\resizebox{\linewidth}{!}{
\begin{tikzpicture}

\definecolor{crimson2143940}{RGB}{214,39,40}
\definecolor{darkgray176}{RGB}{176,176,176}
\definecolor{lightgray204}{RGB}{204,204,204}
\definecolor{orchid227119194}{RGB}{227,119,194}
\definecolor{steelblue31119180}{RGB}{31,119,180}

\begin{axis}[
legend cell align={left},
legend style={
  at={(0.65,0.57)},
  anchor=north east,
  fill opacity=1,
  draw opacity=1,
  text opacity=1,
  draw=lightgray204,
  font=\axisfont,
},
tick align=outside,
tick pos=left,
x grid style={dashed, gray!50},
xlabel={$F_1$},
xlabel style={font=\axisfont},
xmajorgrids,
xmin=-0.55, xmax=3.3,
xtick style={color=black},
y grid style={dashed, gray!50},
ymajorgrids,
ymin=0, ymax=1.05,
ytick style={color=black},
tick label style={font=\axisfont},
major tick length=3pt
]

\draw[draw=none,fill=steelblue31119180,fill opacity=0.8] (axis cs:-0.375,0) rectangle (axis cs:-0.125,0.996771589991929);
\draw[draw=none,fill=steelblue31119180,fill opacity=0.8] (axis cs:0.625,0)  rectangle (axis cs:0.875,0.974358974358974);
\draw[draw=none,fill=steelblue31119180,fill opacity=0.8] (axis cs:1.625,0)  rectangle (axis cs:1.875,0.913461538461538);
\draw[draw=none,fill=steelblue31119180,fill opacity=0.8] (axis cs:2.625,0)  rectangle (axis cs:2.875,0.922330097087379);

\draw[draw=none,fill=crimson2143940,fill opacity=0.8] (axis cs:-0.125,0) rectangle (axis cs:0.125,0.948001027107239);
\draw[draw=none,fill=crimson2143940,fill opacity=0.8] (axis cs:0.875,0)  rectangle (axis cs:1.125,0.860628306865692);
\draw[draw=none,fill=crimson2143940,fill opacity=0.8] (axis cs:1.875,0)  rectangle (axis cs:2.125,0.767855882644653);
\draw[draw=none,fill=crimson2143940,fill opacity=0.8] (axis cs:2.875,0)  rectangle (axis cs:3.125,0.781163275241852);

\addplot [line width=1pt, orchid227119194, mark=*, mark size=3, mark options={solid}]
table {%
0 0.958639176152701
1 0.883408162284922
2 0.789506398979483
3 0.800144160879881
};

\node[
  anchor=north west,
  draw=lightgray204,
  fill=white,
  fill opacity=1,
  text opacity=1,
  rounded corners,
  inner sep=4pt,
  font=\axisfont
] at (axis description cs:0.065,0.72) {$C=0$, $F_2=0$};

\addlegendimage{area legend,draw=none,fill=steelblue31119180,fill opacity=0.8}
\addlegendentry{\makebox[1.5cm][l]{Theory} $P_{C|F_1,F_2}$}

\addlegendimage{area legend,draw=none,fill=crimson2143940,fill opacity=0.8}
\addlegendentry{\makebox[1.5cm][l]{True} $P_{C|F_1,F_2}$}

\addlegendimage{line legend,orchid227119194,line width=1pt,mark=*,mark options={solid},mark size=3}
\addlegendentry{\makebox[1.5cm][l]{Model} $P_{C|F_1,F_2}$}

\end{axis}
\end{tikzpicture}}
	\end{subfigure}\hfill
	\begin{subfigure}[t]{0.25\textwidth}\centering
		\resizebox{\linewidth}{!}{
\begin{tikzpicture}

\definecolor{darkgray176}{RGB}{176,176,176}
\definecolor{darkorange25512714}{RGB}{255,127,14}
\definecolor{gray127}{RGB}{127,127,127}
\definecolor{lightgray204}{RGB}{204,204,204}
\definecolor{mediumpurple148103189}{RGB}{148,103,189}

\begin{axis}[
  legend pos=south east,
  legend cell align={left},
legend style={
  at={(0.65,0.57)},
  anchor=north east,
  fill opacity=1,
  draw opacity=1,
  text opacity=1,
  draw=lightgray204,
  font=\axisfont,
},
  tick align=outside,
  tick pos=left,
  x grid style={dashed, gray!50},
  xlabel={$F_1$},
  xlabel style={font=\axisfont}, 
  xmajorgrids,
  xmin=-0.55, xmax=3.3,
  xtick style={color=black},
  y grid style={dashed, gray!50},
  ymajorgrids,
  ymin=0, ymax=1.05,
  ytick style={color=black},
  tick label style={font=\axisfont},
  major tick length=3pt
]

\draw[draw=none,fill=darkorange25512714,fill opacity=0.8] (axis cs:-0.375,0) rectangle (axis cs:-0.125,0.00322841000807103);
\draw[draw=none,fill=darkorange25512714,fill opacity=0.8] (axis cs:0.625,0)  rectangle (axis cs:0.875,0.0256410256410256);
\draw[draw=none,fill=darkorange25512714,fill opacity=0.8] (axis cs:1.625,0)  rectangle (axis cs:1.875,0.0865384615384616);
\draw[draw=none,fill=darkorange25512714,fill opacity=0.8] (axis cs:2.625,0)  rectangle (axis cs:2.875,0.0776699029126214);

\draw[draw=none,fill=mediumpurple148103189,fill opacity=0.8] (axis cs:-0.125,0) rectangle (axis cs:0.125,0.0519989915192127);
\draw[draw=none,fill=mediumpurple148103189,fill opacity=0.8] (axis cs:0.875,0)  rectangle (axis cs:1.125,0.139371693134308);
\draw[draw=none,fill=mediumpurple148103189,fill opacity=0.8] (axis cs:1.875,0)  rectangle (axis cs:2.125,0.23214416205883);
\draw[draw=none,fill=mediumpurple148103189,fill opacity=0.8] (axis cs:2.875,0)  rectangle (axis cs:3.125,0.218836724758148);

\addplot [line width=1pt, gray127, mark=*, mark size=3, mark options={solid}]
table {%
0 0.0413608238472989
1 0.116591837715077
2 0.210493601020516
3 0.199855839120119
};

\node[
  anchor=north west,
  draw=lightgray204,
  fill=white,
  fill opacity=1,
  text opacity=1,
  rounded corners,
  inner sep=4pt,
  font=\axisfont
] at (axis description cs:0.065,0.72){$C=1$, $F_2=0$};

\addlegendimage{area legend,draw=none,fill=darkorange25512714,fill opacity=0.8}
\addlegendentry{\makebox[1.4cm][l]{Theory} $P_{C|F_1,F_2}$}

\addlegendimage{area legend,draw=none,fill=mediumpurple148103189,fill opacity=0.8}
\addlegendentry{\makebox[1.4cm][l]{True} $P_{C|F_1,F_2}$}

\addlegendimage{line legend,gray127,line width=1pt,mark=*,mark options={solid},mark size=3}
\addlegendentry{\makebox[1.4cm][l]{Model} $P_{C|F_1,F_2}$}

\end{axis}

\end{tikzpicture}}
	\end{subfigure}\hfill
	\begin{subfigure}[t]{0.25\textwidth}\centering
		\resizebox{\linewidth}{!}{
\begin{tikzpicture}

\definecolor{crimson2143940}{RGB}{214,39,40}
\definecolor{darkgray176}{RGB}{176,176,176}
\definecolor{lightgray204}{RGB}{204,204,204}
\definecolor{orchid227119194}{RGB}{227,119,194}
\definecolor{steelblue31119180}{RGB}{31,119,180}

\begin{axis}[
  legend pos=south east,
  legend cell align={left},
  legend style={
  at={(0.65,0.57)},
  anchor=north east,
  fill opacity=1,
  draw opacity=1,
  text opacity=1,
  draw=lightgray204,
  font=\axisfont,
},
  tick align=outside,
  tick pos=left,
  x grid style={dashed, gray!50},
  xlabel={$F_2$},
  xlabel style={font=\axisfont}, 
  xmajorgrids,
  xmin=-0.55, xmax=3.3,
  xtick style={color=black},
  y grid style={dashed, gray!50},
  ymajorgrids,
  ymin=0, ymax=1.05,
  ytick style={color=black},
  tick label style={font=\axisfont},
  major tick length=3pt
]

\draw[draw=none,fill=steelblue31119180,fill opacity=0.8] (axis cs:-0.375,0) rectangle (axis cs:-0.125,0.996771589991929);
\draw[draw=none,fill=steelblue31119180,fill opacity=0.8] (axis cs:0.625,0) rectangle (axis cs:0.875,0.986684420772304);
\draw[draw=none,fill=steelblue31119180,fill opacity=0.8] (axis cs:1.625,0) rectangle (axis cs:1.875,0.98931909212283);
\draw[draw=none,fill=steelblue31119180,fill opacity=0.8] (axis cs:2.625,0) rectangle (axis cs:2.875,0.98406374501992);

\draw[draw=none,fill=crimson2143940,fill opacity=0.8] (axis cs:-0.125,0) rectangle (axis cs:0.125,0.948001027107239);
\draw[draw=none,fill=crimson2143940,fill opacity=0.8] (axis cs:0.875,0) rectangle (axis cs:1.125,0.906934261322021);
\draw[draw=none,fill=crimson2143940,fill opacity=0.8] (axis cs:1.875,0) rectangle (axis cs:2.125,0.904766798019409);
\draw[draw=none,fill=crimson2143940,fill opacity=0.8] (axis cs:2.875,0) rectangle (axis cs:3.125,0.902139008045197);

\addplot [line width=1pt, orchid227119194, mark=*, mark size=3, mark options={solid}]
table {%
0 0.958639176152701
1 0.915736205751727
2 0.924517380021371
3 0.907855099213066
};

\node[
  anchor=north west,
  draw=lightgray204,
  fill=white,
  fill opacity=0.95,
  text opacity=1,
  rounded corners,
  inner sep=4pt,
  font=\axisfont
] at (axis description cs:0.065,0.72) {$C=0$, $F_1=0$};

\addlegendimage{area legend,draw=none,fill=steelblue31119180,fill opacity=0.8}
\addlegendentry{\makebox[1.4cm][l]{Theory} $P_{C|F_1,F_2}$}

\addlegendimage{area legend,draw=none,fill=crimson2143940,fill opacity=0.8}
\addlegendentry{\makebox[1.4cm][l]{True} $P_{C|F_1,F_2}$}

\addlegendimage{line legend,orchid227119194,line width=1pt,mark=*,mark options={solid},mark size=3}
\addlegendentry{\makebox[1.4cm][l]{Model} $P_{C|F_1,F_2}$}

\end{axis}
\end{tikzpicture}}
	\end{subfigure}\hfill
	\begin{subfigure}[t]{0.25\textwidth}\centering
		\resizebox{\linewidth}{!}{
\begin{tikzpicture}

\definecolor{darkgray176}{RGB}{176,176,176}
\definecolor{darkorange25512714}{RGB}{255,127,14}
\definecolor{gray127}{RGB}{127,127,127}
\definecolor{lightgray204}{RGB}{204,204,204}
\definecolor{mediumpurple148103189}{RGB}{148,103,189}

\begin{axis}[
  legend pos=north east,
  legend cell align={left},
  legend style={
  at={(0.65,0.57)},
  anchor=north east,
  fill opacity=1,
  draw opacity=1,
  text opacity=1,
  draw=lightgray204,
  font=\axisfont,
},
  tick align=outside,
  tick pos=left,
  x grid style={dashed, gray!50},
  xlabel={$F_2$},
  xlabel style={font=\axisfont}, 
  xmajorgrids,
  xmin=-0.55, xmax=3.3,
  xtick style={color=black},
  y grid style={dashed, gray!50},
  ymajorgrids,
  ymin=0, ymax=1.05,
  ytick style={color=black},
  tick label style={font=\axisfont},
  major tick length=3pt
]

\draw[draw=none,fill=darkorange25512714,fill opacity=0.8] (axis cs:-0.375,0) rectangle (axis cs:-0.125,0.00322841000807103);
\draw[draw=none,fill=darkorange25512714,fill opacity=0.8] (axis cs:0.625,0)  rectangle (axis cs:0.875,0.0133155792276964);
\draw[draw=none,fill=darkorange25512714,fill opacity=0.8] (axis cs:1.625,0)  rectangle (axis cs:1.875,0.0106809078771696);
\draw[draw=none,fill=darkorange25512714,fill opacity=0.8] (axis cs:2.625,0)  rectangle (axis cs:2.875,0.0159362549800797);

\draw[draw=none,fill=mediumpurple148103189,fill opacity=0.8] (axis cs:-0.125,0) rectangle (axis cs:0.125,0.0519989915192127);
\draw[draw=none,fill=mediumpurple148103189,fill opacity=0.8] (axis cs:0.875,0)  rectangle (axis cs:1.125,0.0930656641721725);
\draw[draw=none,fill=mediumpurple148103189,fill opacity=0.8] (axis cs:1.875,0)  rectangle (axis cs:2.125,0.0952332094311714);
\draw[draw=none,fill=mediumpurple148103189,fill opacity=0.8] (axis cs:2.875,0)  rectangle (axis cs:3.125,0.0978609770536423);

\addplot [line width=1pt, gray127, mark=*, mark size=3, mark options={solid}]
table {%
0 0.0413608238472989
1 0.0842637942482724
2 0.0754826199786294
3 0.0921449007869341
};

\node[
  anchor=north west,
  draw=lightgray204,
  fill=white,
  fill opacity=0.95,
  text opacity=1,
  rounded corners,
  inner sep=4pt,
  font=\axisfont
] at (axis description cs:0.065,0.72) {$C=1$, $F_1=0$};

\addlegendimage{area legend,draw=none,fill=darkorange25512714,fill opacity=0.8}
\addlegendentry{{\makebox[1.4cm][l]{Theory} $P_{C|F_1,F_2}$}}

\addlegendimage{area legend,draw=none,fill=mediumpurple148103189,fill opacity=0.8}
\addlegendentry{\makebox[1.4cm][l]{True} $P_{C|F_1,F_2}$}

\addlegendimage{line legend,gray127,line width=1pt,mark=*,mark options={solid},mark size=3}
\addlegendentry{\makebox[1.4cm][l]{Model} $P_{C|F_1,F_2}$}

\end{axis}
\end{tikzpicture}}
	\end{subfigure}
	\vspace{-0.7cm}
	\caption{ Synthetic dataset. For each class, we fixed one of the features at its modal value and plot: the true posterior, the posterior predicted by the trained neural network, and the theoretical focal-entropy minimizer.}
	\label{fig:MLP experiment}
\end{figure*}
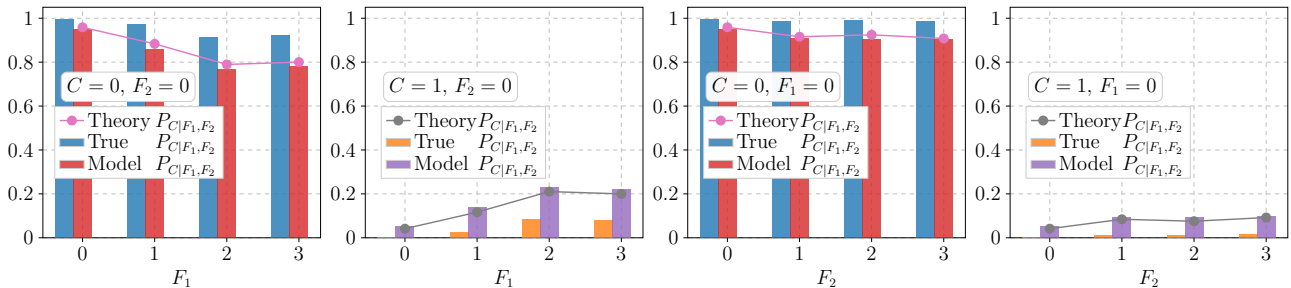

\bibliographystyle{plainnat}
	\bibliography{refs.bib}
\end{document}